\documentclass{lmcs}
\pdfoutput=1

\usepackage{lastpage}
\lmcsdoi{15}{4}{2}
\lmcsheading{}{\pageref{LastPage}}{}{}%
{Jul.~10,~2018}{Oct.~14,~2019}{}


\usepackage[utf8]{inputenc}
\usepackage{amsmath}
\usepackage{amsthm}
\usepackage{amssymb}
\usepackage{amsfonts}
\usepackage{mathpartir}
\usepackage[sort]{cite}
\usepackage{tikz}
\usepackage{hyperref}

\newcommand{\ie}{i.e.}
\newcommand{\etc}{etc.}
\newcommand{\etal}{\emph{et al.}}
\newcommand{\Atom}{\mathbb{A}}
\newcommand{\One}{\textup{\bf 1}}
\newcommand{\ASet}{A}
\newcommand{\asort}{\alpha}
\newcommand{\Var}{\mathcal{V}}
\newcommand{\Sta}{S}
\newcommand{\Res}{R}
\newcommand{\Proc}{\mathsf{pr}}
\newcommand{\Act}{\mathsf{ac}}
\newcommand{\Chan}{\mathsf{ch}}
\newcommand{\nTerm}[1]{\mathbb{N}(#1)} % chktex 36
\newcommand{\rTerm}[1]{\mathbb{T}(#1,\Var)} % chktex 36
\newcommand{\gTerm}[1]{\mathbb{T}(#1)} % chktex 36
\newcommand{\Sort}{\mathsf{S}}

\newcommand{\lowrel}[1]{\mathrel{\raisebox{-1mm}{$\buildrel#1\over\longrightarrow$}}}
\newcommand{\rel}[1]{\buildrel#1\over\longrightarrow}

\newcommand{\nullPA}{\mathit{null}}
\newcommand{\tauPA}[1][]{
	\ifthenelse{\equal{#1}{}}{{\mathit{tau}}}{{\mathit{tau}(#1)}}} % chktex 36
\newcommand{\inPA}[2][]{
	\ifthenelse{\equal{#1}{}}{{\mathit{in}}}{{\mathit{in}(#1,#2)}}} % chktex 36
\newcommand{\outPA}[3][]{
	\ifthenelse{\equal{#1}{}}{{\mathit{out}}}{{\mathit{out}(#1,#2,#3)}}} % chktex 36
\newcommand{\newPA}[1][]{
	\ifthenelse{\equal{#1}{}}{{\mathit{new}}}{{\mathit{new}(#1)}}} % chktex 36
\newcommand{\parPA}[2][]{
	\ifthenelse{\equal{#1}{}}{{\mathit{par}}}{{\mathit{par}(#1,#2)}}} % chktex 36
\newcommand{\sumPA}[2][]{
	\ifthenelse{\equal{#1}{}}{{\mathit{sum}}}{{\mathit{sum}(#1,#2)}}} % chktex 36
\newcommand{\repPA}[1][]{
	\ifthenelse{\equal{#1}{}}{{\mathit{rep}}}{{\mathit{rep}(#1)}}} % chktex 36

\newcommand{\tauAA}{{\mathit{tauA}}}
\newcommand{\inAA}[2][]{
	\ifthenelse{\equal{#1}{}}{{\mathit{inA}}}{{\mathit{inA}(#1,#2)}}} % chktex 36
\newcommand{\binAA}[2][]{
	\ifthenelse{\equal{#1}{}}{{\mathit{binA}}}{{\mathit{binA}(#1,#2)}}} % chktex 36
\newcommand{\outAA}[2][]{
	\ifthenelse{\equal{#1}{}}{{\mathit{outA}}}{{\mathit{outA}(#1,#2)}}} % chktex 36
\newcommand{\boutAA}[2][]{
	\ifthenelse{\equal{#1}{}}{{\mathit{boutA}}}{{\mathit{boutA}(#1,#2)}}} % chktex 36

\newcommand{\resAF}[2]{{(#1,#2)}}
\newcommand{\resA}[3]{{[#1](#2,#3)}} % chktex 36
\newcommand{\tr}[2]{(#1\,#2)}
\newcommand{\peract}[2]{#1\cdot #2}
\newcommand{\susp}[2]{#1\boldsymbol{\{}\!\!\!\boldsymbol{\{}#2\boldsymbol{\}}\!\!\!\boldsymbol{\}}}
\newcommand{\rep}[2]{#1/#2}
\newcommand{\ren}[2]{#1\{#2\}}
\newcommand{\inj}{\mathrm{inj}}
\newcommand{\fs}{\mathrm{fs}}
\newcommand{\dom}{\mathrm{dom}}

\newcommand{\R}{\mathcal{R}}
\newcommand{\Ru}{\textsc{Ru}}
\newcommand{\NT}[1]{\mathit{NT}[\![#1]\!]}

\newcommand{\fra}{\!\mathrel{\not\!{\not\hspace{-1pt}\approx}}}

\newcommand{\pset}{\mathcal{P}}

\newcommand{\bn}{\mathrm{bn}}
\newcommand{\supp}{\mathrm{supp}}
\newcommand{\perm}[1]{\mathrm{Perm}\;#1}
\newcommand{\permso}[1]{\mathrm{Perm}_s\,#1}
\newcommand{\nf}[1]{\langle#1\rangle \mathit{nf}} % chktex 36 chktex 1
\newcommand{\SigmaNTS}{\Sigma_{\textup{NTS}}}
\newcommand{\SigmaNTSAbs}{\Sigma_{\textup{NTS}}^{[\Chan]}}

\newcommand{\Tr}{\mathcal{T}}
\newcommand{\TrAbs}{{\mathcal{T}^{[\Chan]}}}

\newcommand{\Trans}[1][]{
  \ifthenelse{\equal{#1}{}}{\mathfrak{T}}{\mathfrak{T}(#1)}} % chktex 36
\newcommand{\TransAbs}[1][]{
  \ifthenelse{\equal{#1}{}}{\mathfrak{T}^{[\Chan]}}{\mathfrak{T}^{[\Chan]}(#1)}} % chktex 36

\newcommand{\trel}{\rel{}}
\newcommand{\trelAbs}{\rel{}_{[\Chan]}}

\newcommand{\RE}{\mathcal{R}_{\textsc{E}}}
\newcommand{\RL}{\mathcal{R}_{\textsc{L}}}

\newcommand{\REAbs}{\mathcal{R}_{\textsc{E}}^{[\Chan]}}
\newcommand{\RLAbs}{\mathcal{R}_{\textsc{L}}^{[\Chan]}}

\newcommand{\TrE}{\mathcal{T}_{\textsc{E}}}
\newcommand{\TrEAbs}{{\mathcal{T}_{\textsc{E}}^{[\Chan]}}}

\newcommand{\TrL}{\mathcal{T}_{\textsc{L}}}
\newcommand{\TrLAbs}{{\mathcal{T}_{\textsc{L}}}^{[\Chan]}}

\newcommand{\bnE}{\bn_{\textsc{E}}}
\newcommand{\bnL}{\bn_{\textsc{L}}}

\newcommand{\curry}{\mathrm{curry}}
\newcommand{\app}{\mathrm{app}}
\newcommand{\fa}{\mathrm{fa}}

\newcommand{\myby}[1]{\textup{\quad\textit{by #1}}}

\newcommand{\blockqed}{\hfill$\blacksquare$}
\def\eg{\emph{e.g.}}

\begin{document}

\title[Rule Formats for Nominal Process Calculi]{Rule Formats for Nominal Process Calculi\rsuper*}
\titlecomment{\lsuper{*}Extended version of a paper with the same title presented at CONCUR 2017.}

\author[L.~Aceto]{Luca Aceto\rsuper{{a,b}}}
\author[I.~Fábregas]{Ignacio Fábregas\rsuper{c}}
\author[A.~García-Pérez]{Álvaro García-Pérez\rsuper{c}}
\author[A.~Ingólfsdóttir]{\texorpdfstring{\\}{}Anna Ingólfsdóttir\rsuper{b}}
\author[Y.~Ortega-Mallén]{Yolanda Ortega-Mallén\rsuper{d}\texorpdfstring{\vspace{-7mm}}{}}

\address{\lsuper{a}Gran Sasso Science Institute, L'Aquila, Italy}
\email{luca.aceto@gssi.it}
\address{\lsuper{b}ICE-TCS, School of Computer Science, Reykjavik University, Iceland}
\email{\{luca,annai\}@ru.is}
\address{\lsuper{c}IMDEA Software Institute, Madrid, Spain}
\email{\{ignacio.fabregas,alvaro.garcia.perez\}@imdea.org}
\address{\lsuper{d}Departamento de Sistemas Informáticos y Computación, Universidad Complutense de Madrid, Spain}
\email{yolanda@ucm.es}

\thanks{Research partially supported by the project {Nominal SOS} (nr.~141558-051) of the % chktex 8
  Icelandic Research Fund, the project 001-ABEL-CM-2013 within the {NILS} Science and % chktex 8
  Sustainability Programme, the Spanish Projects TRACES (TIN2015-67522-C3-3-R) and % chktex 8
  Bosco (PGC2018-102210-B-I00), and by Comunidad de Madrid as part of the program % chktex 8
  S2018/TCS-4339 (BLOQUES-CM) co-funded by EIE Funds of the European Union, and the projects % chktex 8
  {RACCOON} (H2020-EU~714729) and MATHADOR (COGS 724.464) of the European Research % chktex 8
  Council, and the Spanish addition to MATHADOR (TIN2016-81699-ERC)} % chktex 8

%% etc.

%% required for running head on odd and even pages, use suitable
%% abbreviations in case of long titles and many authors:

%%%%%%%%%%%%%%%%%%%%%%%%%%%%%%%%%%%%%%%%%%%%%%%%%%%%%%%%%%%%%%%%%%%%%%%%%%%

%% the abstract has to PRECEDE the command \maketitle:
%% be sure not to issue the \maketitle command twice!

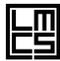
\begin{abstract}
  \noindent The nominal transition systems (NTSs) of Parrow et al.\ describe the
  operational semantics of nominal process calculi. We study NTSs in terms of the nominal
  residual transition systems (NRTSs) that we introduce. We provide rule formats for the
  specifications of NRTSs that ensure that the associated NRTS is an NTS and apply them to
  the operational specifications of the early and late pi-calculus. We also explore
  alternative specifications of the NTSs in which we allow residuals of abstraction sort,
  and introduce translations between the systems with and without residuals of abstraction
  sort. Our study stems from the Nominal SOS of Cimini et al.\ and from earlier works in
  nominal sets and nominal logic by Gabbay, Pitts and their collaborators.
\end{abstract}

\maketitle

%%%%%%%%%%%%%%%%%%%%%%%%%%%%%%%%%%%%%%%%%%%%%%%%%%%%%%%%%%%%%%%%%%%%%%%%%%%
%% Intro
\section{Introduction}%
\label{sec:intro}

The goal of this paper is to develop the foundations of a framework for studying the
meta-theory of structural operational semantics (SOS)~\cite{Plo04} for process calculi
with names and name-binding operations, such as the
$\pi$-calculi~\cite{MPW92,SW01}. To this end, we build on the large body of work on
rule formats for SOS, as surveyed in~\cite{AFV01,MRG07}, and on the nominal techniques of
Gabbay, Pitts and their co-workers~\cite{UPG04,CP07,GM09,Pit13}.

Rule formats provide syntactic templates guaranteeing that the models of the calculi,
whose semantics they specify, enjoy some desirable properties. A first design decision
that has to be taken in developing a theory of rule formats for a class of languages is
therefore the choice of the semantic objects specified by the rules. The target semantic
model we adopt in our study is that of \emph{nominal transition systems} (NTSs), which
have been introduced by Parrow \etal\ in~\cite{PBEGW15,PWBE17} as a uniform model to
describe the operational semantics of a variety of calculi with names and name-binding
operations.  Based on this choice, a basic sanity criterion for a collection of rules
describing the operational semantics of a nominal calculus is that they specify an NTS,
and we present a rule format guaranteeing this property (Thm.~\ref{the:alpha-conversion}).

As a first stepping stone in our study, we introduce \emph{nominal residual transition
  systems} (NRTSs), and study NTSs in terms of NRTSs (Section~\ref{sec:preliminaries}).
More specifically, the only requirement of an NRTS is that its transition relation is
equivariant, which means that it treats names uniformly. This is a desirable property of
models of nominal calculi, such as NTSs. Moreover, NTSs are NRTSs that, in addition to
having an equivariant transition relation, satisfy a property Parrow \etal\ call
\emph{alpha-conversion of residuals} (see Def.~\ref{def:nts} for the details). The latter
property formalises a key aspect of calculi in which names can be scoped to represent
local resources.  To wit, one crucial feature of the $\pi$-calculus is scope
opening~\cite{MPW92}. Consider a transition $p\lowrel{\overline{a}(\nu b)}p'$ in which a
process $p$ exports a private/local channel name $b$ along channel $a$. Since the name $b$
is local, it `can be subject to alpha-conversion'~\cite{PBEGW15} and the transitions
$p\lowrel{\overline{a}(\nu c)}\ren{p}{\rep{c}{b}}$ should also be present for each `fresh
name' $c$.

In contrast to related work~\cite{CMRG12,FG07}, our approach uses \emph{nominal terms}~\cite{Pit13} to connect the specification system with the semantic model. This has the
advantage of capturing the requirement that transitions be `up to alpha-equivalence'
(typical in nominal calculi) without instrumenting alpha-conversion explicitly in the
specification system.

We specify an NRTS by means of a nominal residual transition system specification (NRTSS),
which describes the syntax of a nominal calculus in terms of a nominal signature
(Section~\ref{sec:terms}) and its semantics by means of a set of inference rules
(Section~\ref{sec-specification-nrts}). We develop the basic theory of the NRTS/NRTSS
framework, building on the nominal algebraic datatypes of Pitts~\cite{Pit13} and the
nominal rewriting framework of Fernández and Gabbay~\cite{FG07}. Based on this framework,
we provide rule formats~\cite{AFV01,MRG07} for NRTSSs (Section~\ref{sec-rule-format-nrts})
that ensure that the induced transition relation is equivariant
(Thm.~\ref{thm:rule-format-equivariance}) and enjoys alpha-conversion of residuals
(Thm.~\ref{the:alpha-conversion}), and is therefore an NTS\@. Section~\ref{sec:example-nts}
presents an example of application of these rule formats to the setting of the
$\pi$-calculus. Section~\ref{sec:nts-abstraction-sort} explores alternative specifications
of the NTSs in which we allow a residual to be an atom abstraction (hereafter referred to
as \emph{residual with abstraction sort}). We introduce translations between the systems
with and without residuals of abstraction sort (Defs.~\ref{def:nts-to-nrts} and~\ref{def:nrts-to-nts}). We develop a rule format that guarantees that these translations
are the inverse of each other (Thms.~\ref{thm-composition-identity} and~\ref{thm-composition-identity-abstraction}). Section~\ref{sec:application-BA-format}
presents an example of application of this rule format to the early $\pi$-calculus and to
a slightly modified version of the late $\pi$-calculus. We also show that both the
specification with and without residuals of abstraction sort induce the same model of
computation. Finally, Section~\ref{sec:conclusions} discusses related work, as well as
avenues for future work, and concludes.

The appendix accompanying the paper collects some proofs that are omitted in the main
text.

This paper is an extended version of a paper with the same title presented at CONCUR 2017~\cite{AFGIO17}. The novel content in this extended version is summarised below:
\begin{itemize}
\item In Section~\ref{sec:preliminaries} we recall the notion of finite renamings, which
  play a prominent role throughout this paper since they replace the permutations in the
  moderated terms of the CONCUR 2017 paper.
\item In Section~\ref{sec:early-pi-calculus} we introduce an NRTSS that faithfully
  captures the original semantics of the early $\pi$-calculus~\cite{San96}. The NRTSS of
  the CONCUR 2017 paper induced a semantics that failed to capture some transitions in the
  original early $\pi$-calculus.
\item In Section~\ref{sec:late-pi-calculus} we introduce an NRTSS whose induced semantics
  differs minimally from the original semantics of the late $\pi$-calculus~\cite{San96}
  (see Remark~\ref{rem:prevent-capture-in} for further discussion). We also apply the rule
  format for alpha-conversion of residuals to this version of the late
  $\pi$-calculus. This section is entirely novel.
\item Section~\ref{sec:nts-abstraction-sort}, where we study alternative formulations of
  the NTSs in which we allow residuals of abstraction sorts, is entirely novel.
\item In Section~\ref{sec:application-BA-format} we apply the rule formats from
  Section~\ref{sec:nts-abstraction-sort} to the early $\pi$-calculus and to our version of
  the late $\pi$-calculus. This section is entirely novel too.
\item We have included the detailed proofs of all lemmas and theorems in the paper, some
  of which were missing in the conference version.
\end{itemize}

%%%%%%%%%%%%%%%%%%%%%%%%%%%%%%%%%%%%%%%%%%%%%%%%%%%%%%%%%%%%%%%%%%%%%%%%%%%

\section{Preliminaries}%
\label{sec:preliminaries}
This section collects some earlier foundational work by Gabbay and Pitts on nominal sets
and finitary renamings~\cite{GH08,GP02,Pit13,Pit16} on which our work builds, and recalls
the nominal transition systems of Parrow \etal~\cite{PBEGW15}.

\subsection*{Nominal Sets}
We assume a countably infinite set $\Atom$ of \emph{atoms} and consider $\perm{\Atom}$ as
the group of \emph{finite permutations of atoms} (hereafter \emph{permutations}) ranged
over by $\pi$, where we write $\iota$ for the \emph{identity}, $\circ$ for
\emph{composition} and $\pi^{-1}$ for the \emph{inverse} of permutation $\pi$. We are
particularly interested in \emph{transpositions} of two atoms: $\tr{a}{b}$ stands for the
permutation that swaps $a$ with $b$ and leaves all other atoms fixed. Every permutation
$\pi$ is equal to the composition of a finite number of transpositions, \ie\
$\pi=\tr{a_1}{b_1}\circ\ldots\circ\tr{a_n}{b_n}$ with $n\geq 0$.

An \emph{action} of the group $\perm{\Atom}$ on a set $S$ is a binary operation mapping
each $\pi\in\perm{\Atom}$ and $s\in S$ to an element $\peract{\pi}{s}\in S$, and
satisfying the identity law $\peract{\iota}{s}=s$ and the composition law
$\peract{(\pi_1\circ\pi_2)}{s}=\peract{\pi_1}{(\peract{\pi_2}{s})}$. A
\emph{$\perm{\Atom}$-set} is a set equipped with an action of $\perm{\Atom}$.

We say that a set of atoms $A$ \emph{supports} an object $s$ iff $\peract{\pi}{s}=s$ for
every permutation $\pi$ that leaves each element $a\in A$ invariant. In particular, we are
interested in sets all of whose elements have finite support (Def.~2.2 of~\cite{Pit13}).

\begin{defi}[Nominal sets] A \emph{nominal set} is a $\perm{\Atom}$-set all of whose
  elements are finitely supported.
\end{defi}

For each element $s$ of a nominal set, we write $\supp(s)$ for the least set that supports
$s$, called the \emph{support} of $s$. (Intuitively, the action of permutations on a set
$S$ determines that a finitely supported $s\in S$ only depends on atoms in $\supp(s)$, and
no others.)  The set $\Atom$ of atoms is a nominal set when $\peract{\pi}{a}=\pi\,a$ since
$\supp(a)=\{a\}$ for each atom $a\in\Atom$. The set $\perm{\Atom}$ of finite permutations
is also a nominal set where the permutation action on permutations is given by
conjugation, \ie\ $\peract{\pi}{\pi'}=\pi\circ\pi'\circ\pi^{-1}$, and the support of a
permutation $\pi$ is $\supp(\pi)=\{a\mid \pi a\not=a\}$.

Given two $\perm{\Atom}$-sets $S$ and $T$ and a function $f:S\to T$, the action of
permutation $\pi$ on function $f$ is given by conjugation, \ie\
$(\peract{\pi}{f})(s)=\peract{\pi}{f(\peract{\pi^{-1}}{s})}$ for each $s\in S$. We say
that a function $f:S\to T$ is \emph{equivariant} iff
$\peract{\pi}{f(s)}=f(\peract{\pi}{s})$ for every $\pi\in\perm{\Atom}$ and every $s\in S$.
The intuition is that an equivariant function $f$ is atom-blind, in that $f$ does not
treat any atom preferentially. It turns out that a function $f$ is equivariant iff
$\supp(f)=\emptyset$ (Rem.~2.13 of~\cite{Pit13}). The function $\supp$ is equivariant
(Prop.~2.11 of~\cite{Pit13}).

Let $S$ be a $\perm{\Atom}$-set, we write $S_\fs$ for the nominal set that contains the
elements in $S$ that are finitely supported. Let $S_1$ and $S_2$ be nominal sets. The
product $S_1\times S_2$ is a nominal set (Prop.~2.14 of~\cite{Pit13}). The permutation
action for products is given componentwise (Eq~(1.12) of~\cite{Pit13}).

Conjugation yields that, for every $\perm{\Atom}$-set $S$, the action of $\pi$ on $s\in S$
is equivariant. Indeed,
\begin{displaymath}
  \peract{\pi}{(\peract{\pi_1}{s})} = \peract{(\pi\circ\pi_1)}{s} =
  \peract{(\pi\circ\pi_1\circ\pi^{-1}\circ\pi)}{s} =
  \peract{((\peract{\pi}{\pi_1})\circ\pi)}{s} =
  \peract{(\peract{\pi}{\pi_1})}{(\peract{\pi}{s})}.
\end{displaymath}

\noindent
It is also straightforward to show that composition of permutations is equivariant. In
fact,
\begin{displaymath}
  \peract{\pi}{(\pi_1\circ\pi_2)} = \pi\circ(\pi_1\circ\pi_2)\circ\pi^{-1} =
  (\pi\circ\pi_1\circ\pi^{-1})\circ(\pi\circ\pi_2\circ\pi^{-1}) =
  (\peract{\pi}{\pi_1})\circ(\peract{\pi}{\pi_2}).
\end{displaymath}

An element $s_1\in S_1$ \emph{is fresh in}
$s_2\in S_2$, written $s_1\#s_2$, iff $\supp(s_1)\cap\supp(s_2)=\emptyset$. The freshness
relation is equivariant (Eq.~(3.2) of~\cite{Pit13}).

We consider \emph{atom abstractions} (Sec.~4 of~\cite{Pit13}), which represent
alpha-equivalence classes of elements.

\begin{defi}[Atom abstraction]%
  \label{def:atom-abstraction}
  Given a nominal set $S$, the \emph{atom abstraction} of atom $a$ in element $s\in S$,
  written $\langle a\rangle s$, is the $\perm{\Atom}$-set
  $\langle a\rangle s=\{(b,\peract{\tr{b}{a}}{s})\mid b=a\lor b\# s\}$, whose permutation
  action is
  $ \peract{\pi}{\langle a\rangle s}=\{
  (\peract{\pi}{b},\peract{\pi}{(\peract{\tr{b}{a}}{s})})\mid
  \peract{\pi}{b}=\peract{\pi}{a}\lor \peract{\pi}{b}\# \peract{\pi}{s}\}$.

  We write $[\Atom]S$ for the set of \emph{atom abstractions} in elements of $S$, which is
  a nominal set (Def.~4.4 of~\cite{Pit13}), since
  $\supp(\langle a\rangle s)=\supp(s)\setminus\{a\}$ for each atom $a$ and element
  $s\in S$.
\end{defi}

\begin{rem}%
  \label{rem:equality-bodies}
  Notice that, by Lemma~4.3 in~\cite{Pit11}, $s=s'$ whenever
  $\langle a\rangle s = \langle a\rangle s'$.\blockqed
\end{rem}

Nominal sets are the objects of a category $\mathbf{Nom}$ whose morphisms are the
equivariant functions. The category $\mathbf{Nom}$ is closed under finite products and
both finite and infinite coproducts.\footnote{In $\mathbf{Nom}$, coproducts correspond to
  disjoint unions.} We write $s=\inj_i s'$ with $i\in I$ and $s'\in S_i$ for an element
$s$ in a coproduct $\sum_{i\in I}(S_i)$. (For a finite coproduct ${S_1 + \cdots + S_n}$ we
let $I=\{1,\ldots,n\}$.)  For other set-theoretical operations (\ie\ infinite products,
functions, partial functions, power sets) the following caveat applies. The category of
nominal sets is closed under the variant of each operation that restricts any universal
quantification that is involved in the operation to quantify only over finitely supported
elements (see Sections~2.2 to 2.5 of~\cite{Pit13}).

The \emph{nominal function set} between nominal sets $S$ and $T$ (Definition~2.18 of~\cite{Pit13}) is the nominal set ${(T^S)}_\fs$ of finitely supported functions from $S$ to
$T$---be they equivariant or not; recall that an equivariant function has empty
support. (We may write $S\to_\fs T$ in lieu of ${(T^S)}_\fs$.) The application and currying
functions can be respectively restricted to equivariant functions
${\app : (X\to_\fs Y) \times X \to Y}$ and
$\curry : (Z \times X \to_\fs Y) \to Z \to (X \to_\fs Y)$ such that the nominal function
set coincides with the \emph{exponential object} in $\mathbf{Nom}$, \ie\ there is a
bijection between hom-sets $\mathbf{Nom}(Z\times X,Y)$ and $\mathbf{Nom}(Z,X\to_\fs Y)$
given by sending $f \in \mathbf{Nom}(Z\times X,Y)$ to
$\curry(f)\in\mathbf{Nom}(Z,X\to_\fs Y)$.  (Section~2.4 in~\cite{Pit13} spells out all the
details on this isomorphism.)

Finally, the category $\mathbf{Nom}$ is Cartesian closed (Theorem~2.19 in~\cite{Pit13}),
\ie, $\mathbf{Nom}$ admits all the finite products (including the empty product $\One$
which is the terminal object) and all the exponentials.

\subsection*{Renamings}

We consider the \emph{finitely supported renamings} (hereafter \emph{renamings}) ranged
over by $\rho$, which are finitely supported functions $\rho:\Atom\to_\fs \Atom$, that is,
functions that act like the identity on all but finitely many atoms. We write $\iota$ for
the \emph{identity function} and `$;$' for \emph{diagrammatical composition}, that is, % chktex 40
$f;g$ denotes the function $g\circ f$. We are particularly interested in % chktex 40
\emph{replacements} of an atom by another: $\rep{b}{a}$ stands for the replacement that
substitutes $a$ with $b$ and leaves all other atoms fixed. Every renaming $\rho$ is equal
to the composition of a finite number of replacements~\cite{GH08}, \ie\
$\rho=\rep{b_1}{a_1};\ldots;\rep{b_n}{a_n}$ with $n\geq 0$. Notice that
$\Atom\to_\fs \Atom$ with `$;$' as composition operator and $\iota$ as identity element is % chktex 40
a monoid~\cite{GH08}.

An \emph{action} of the monoid $\Atom\to_\fs \Atom$ on a set $S$ is a binary operation
mapping each $\rho\in\Atom\to_\fs \Atom$ and $s\in S$ to an element $\ren{s}{\rho}\in S$,
and satisfying the identity law $\ren{s}{\iota}=s$ and the composition law
$\ren{(\ren{s}{\rho_1})}{\rho_2}=\ren{s}{\rho_1;\rho_2}$. We will provide an action of
renaming for the raw terms to be defined in Section~\ref{sec:terms}. An action of renaming
could be defined for every object in $\mathbf{Nom}$, which ultimately gives rise to the
category $\mathbf{Ren}$ of renamings as described in~\cite{GH08}, which is a
generalisation of $\mathbf{Nom}$. We are interested in interpreting our terms as the
nominal algebraic datatypes of~\cite{Pit13}, which live in $\mathbf{Nom}$, and therefore
we refrain ourselves from interpreting our terms in the category $\mathbf{Ren}$, and we
treat renamings as the exponential objects $\Atom\to_\fs \Atom$ in the former category.

Notice that every permutation is an instance of a renaming. For every permutation $\pi$,
we may write $\ren{s}{\pi} = \peract{\pi}{s}$ for the action of renaming $\pi$ on $s$, and
for every renaming $\rho$, we may write $\pi;\rho$ for the diagrammatical composition of
$\rho$ after $\pi$. As we have mentioned above, the renamings are the exponential object
$\Atom\to_\fs \Atom$ in the category $\mathbf{Nom}$, and therefore they are equipped with
a permutation action given by $\peract{\pi}{\rho} = \pi^{-1};\rho;\pi$. As for any other
element of an object in $\textbf{Nom}$, the support of a renaming $\rho$ is the least set
$A$ such that $\peract{\pi}{\rho}=\rho$ for every permutation $\pi$ that leaves each
element of $A$ invariant.

\begin{exa}
  Consider the replacement $\rep{b}{a}$. Its support is $\supp(\rep{b}{a})=\{a,b\}$, as we
  show next. Let $\pi$ be a permutation such that $\peract{\pi}{a}=a$ and
  $\peract{\pi}{b}=b$. We show that $\pi^{-1};\rep{b}{a};\pi = \rep{b}{a}$. For atom $a$,
  \begin{displaymath}
    \ren{a}{\pi^{-1};\rep{b}{a};\pi} = \ren{(\peract{\pi^{-1}}{a})}{\rep{b}{a};\pi}
    = \ren{a}{\rep{b}{a};\pi} = \ren{b}{\pi} = \peract{\pi}{b} = b
    = \ren{a}{\rep{b}{a}}.
  \end{displaymath}
  For any other atom $c\#a$,
  \begin{displaymath}
    \ren{c}{\pi^{-1};\rep{b}{a};\pi} = \ren{(\peract{\pi^{-1}}{c})}{\rep{b}{a};\pi}
    = \ren{(\peract{\pi^{-1}}{c})}{\pi} = \peract{\pi}{(\peract{\pi^{-1}}{c})} = c
    = \ren{c}{\rep{b}{a}},
  \end{displaymath}
  since $\peract{\pi^{-1}}{c}\not= a$ by the assumptions on $\pi$. Therefore $\{a,b\}$
  supports $\rep{b}{a}$, while it is not hard to see that no subset of $\{a,b\}$ does so.
  \blockqed
\end{exa}

Since every renaming $\rho$ is finitary, its support can be defined alternatively as in
the proposition below.
\begin{prop}\label{prop:support-rho}
  Let $\rho$ be a renaming. The support $\supp(\rho)=\{a,(\rho\ a)\mid\rho\ a\not=a\}$.
\end{prop}
The proof of Proposition~\ref{prop:support-rho} is in Appendix~\ref{ap:preliminaries}.

\subsection*{Nominal Transition Systems}
Nominal transition systems adopt the state/residual presentation for transitions of~\cite{BP09}, where a residual is a pair consisting of an action and a state. In~\cite{PBEGW15}, Parrow \etal\ develop modal logics à la Hennessy-Milner for process
nominal calculi. Here we are mainly interested in the transition relation and we adapt
Definition~1 in~\cite{PBEGW15} by removing the predicates. We write $\pset_\omega(\Atom)$
for the \emph{finite power set} of $\Atom$.

\begin{defi}[Nominal transition system]%
  \label{def:nts}
  A \emph{nominal transition system} (NTS) is a quadruple $(S,\mathit{Act},\bn,\rel{})$
  where $S$ and $\mathit{Act}$ are nominal sets of \emph{states} and \emph{actions}
  respectively, $\bn:\mathit{Act}\to\pset_\omega(\Atom)$ is an equivariant function that
  delivers the \emph{binding names} in an action, and
  ${\rel{}}\subseteq S\times(\mathit{Act}\times S)$ is an equivariant binary transition
  relation from states to \emph{residuals} (we let $\mathit{Act}\times S$ be the set of
  residuals). The function $\bn$ is such that $\bn(\ell)\subseteq\supp(\ell)$ for each
  $\ell\in \mathit{Act}$. We often write $p\rel{}(\ell,p')$ in lieu of
  $(p,(\ell,p'))\in{\rel{}}$.

  Finally, the transition relation $\rel{}$ must satisfy \emph{alpha-conversion of
    residuals}, that is, if ${a\in \bn(\ell)}$, $b\#(\ell,p')$ and $p\rel{}(\ell,p')$ then
  also $p\rel{}(\peract{\tr{a}{b}}{\ell},\peract{\tr{a}{b}}{p'})$, or equivalently
  $p\rel{}\peract{\tr{a}{b}}{(\ell,p')}$.
\end{defi}

We will consider an NTS (without its associated binding-names function $\bn$) as a
particular case of a nominal residual transition system, which we introduce next.
\begin{defi}[Nominal residual transition system]%
  \label{def:nrts}
  A \emph{nominal residual transition system} (NRTS) is a triple $(\Sta,\Res,\rel{})$
  where $\Sta$ and $\Res$ are nominal sets, and where ${\rel{}}\subseteq S\times R$ is an
  equivariant binary transition relation. We say $\Sta$ is the set of \emph{states} and
  $\Res$ is the set of \emph{residuals}.
\end{defi}
The connection between NTSs and NRTSs will be studied in more detail in
Section~\ref{sec-rule-format-nrts}.

%%%%%%%%%%%%%%%%%%%%%%%%%%%%%%%%%%%%%%%%%%%%%%%%%%%%%%%%%%%%%%%%%%%%%%%%%%%
%% Terms
\section{Nominal Terms}%
\label{sec:terms}

This section is devoted to the notion of nominal terms, which are syntactic objects that
make use of the atom abstractions of Definition~\ref{def:atom-abstraction} and represent
terms up to alpha-equivalence. As a first step, we introduce raw terms, devoid of any
notion of alpha-equivalence. Our raw terms resemble those from the literature, mainly from~\cite{UPG04,FG07,CP07,Pit13}, but with some important differences. In particular, our
terms include both variables (\ie\ unknowns) and moderated terms (\ie\ explicit renamings
over raw terms), and we consider atom and abstraction sorts. (The raw terms of~\cite{Pit13} do not include moderated terms, and the ones in~\cite{UPG04,FG07} only
consider moderated variables where the delayed renaming is a permutation. In~\cite{CP07}
the authors consider neither atom nor abstraction sorts.) We also adopt the classic
presentation of free algebras and term algebras in~\cite{GTWW77,BS00} in a different way
from that in~\cite{CP07,Pit13}. The raw terms correspond to the standard notion of free
algebra over a signature generated by a set of variables. We then adapt the
$\Sigma$-structures of~\cite{CP07} to our sorting schema. Finally, the nominal terms are
the interpretations of the ground terms in the initial $\Sigma$-structure; we show that
they coincide with the nominal algebraic terms of~\cite{Pit13}.

\begin{defi}[Nominal signature and nominal sort]
  A \emph{nominal signature} (or simply a \emph{signature}) $\Sigma$ is a triple
  $(\Delta,\ASet,F)$ where $\Delta=\{\delta_1,\ldots,\delta_n\}$ is a finite set of
  \emph{base sorts}, $\ASet$ is a countable set of \emph{atom sorts}, and $F$ is a finite
  set of \emph{function symbols}. The \emph{nominal sorts} over $\Delta$ and $\ASet$ are
  given by the grammar
  \begin{displaymath}
    \sigma~::=~\delta\mid \asort \mid [\asort]\sigma\mid \sigma_1 \times \cdots \times
    \sigma_k,
  \end{displaymath}
  with $k\geq 0$, $\delta\in\Delta$ and $\asort\in\ASet$. The sort $[\asort]\sigma$ is the
  \emph{abstraction sort}. Symbol $\times$ denotes the \emph{product sort}, which is
  associative; $\sigma_1\times\cdots\times\sigma_k$ stands for the sort of the empty
  product when $k=0$, which we may write as $\One$. We write $\Sort$ for the set of
  nominal sorts. We arrange the function symbols in $F$ based on the sort of the data
  (base sort) that they produce. We write $f_{ij}\in F$ with $1\leq i\leq n$ and
  $1\leq j \leq m_i$ such that $f_{ij}$ has arity $\sigma_{ij}\to\delta_i$, where
  $\delta_i$ is a base sort.
\end{defi}

The theory of nominal sets extends to the case of many-sorted atoms (see Sec.~4.7 in~\cite{Pit13}). We assume that $\Atom$ contains a countably infinite collection of atoms
$a_\asort$, $b_\asort$, $c_\asort$,~\ldots\ for each atom sort $\asort$ such that the sets
of atoms $\Atom_\asort$ of each sort are mutually disjoint. We write
$\permso{\Atom}=\{\pi\in\perm{\Atom}\mid \forall \asort\in\ASet.\, \forall
a\in\Atom_\asort.\;\pi\,a\in\Atom_\asort\}$
for the subgroup of finite permutations that respect the sorting. The sorted nominal sets
are the $\permso{\Atom}$-sets whose elements are finitely supported. We also consider
renamings that respect the sorting, which we write
${(\Atom\to_\fs \Atom)}_s=\{\rho\in \Atom\to_\fs \Atom \mid \forall \alpha\in A.~\forall
a\in \Atom_\alpha.~ \rho\ a \in \Atom_\alpha\}$.
(Notice that every permutation in $\permso{\Atom}$ is a renaming that respects the
sorting.) In the sequel we may drop the $s$ subscript in $\permso{\Atom}$ and in
${(\Atom\to_\fs\Atom)}_s$, and omit the `sorted' epithet from `sorted nominal sets'.

We let $\Var$ be a set that contains a countably infinite collection of \emph{variable
  names} (variables for short) $x_\sigma$, $y_\sigma$, $z_\sigma$, \ldots\ for each sort
$\sigma$, such that the sets of variables $\Var_\sigma$ of each sort are mutually
disjoint. We also assume that $\Var$ is disjoint from $\Atom$.

\begin{defi}[Raw terms]%
  \label{def:raw-terms}
  Let $\Sigma=(\Delta,A,F)$ be a signature. The set of \emph{raw terms over signature
    $\Sigma$ and set of variables $\Var$} (\emph{raw terms} for short) is given by the
  grammar
  \begin{displaymath}
    \begin{array}{rcl}
      t_\sigma ::=&
                    x_\sigma\mid
                    a_\asort\mid
                    {(\susp{t_\sigma}{\rho})}_\sigma\mid
                    {([a_\asort]t_\sigma)}_{[\asort]\sigma}\mid
                    {(t_{\sigma_1},\ldots,t_{\sigma_k})}_{\sigma_1\times\cdots\times\sigma_k} \mid
                    {(f_{ij}(t_{\sigma_{ij}}))}_{\delta_i},
    \end{array}
  \end{displaymath}
  where term $x_\sigma$ is a \emph{variable} of sort $\sigma$, term $a_\alpha$ is an
  \emph{atom} of sort $\alpha$, term ${(\susp{t_\sigma}{\rho})}_\sigma$ is a \emph{moderated
    term} (\ie\ the explicit, or delayed, renaming $\rho$ over term $t_\sigma$), term
  ${([a_\asort]t_\sigma)}_{[\asort]\sigma}$ is the \emph{abstraction of atom $a_\asort$ in
    term $t_\sigma$}, term
  ${(t_{\sigma_1},\ldots,t_{\sigma_k})}_{\sigma_1\times\cdots\times\sigma_k}$ is the
  \emph{product of terms} $t_{\sigma_1}$, \ldots, $t_{\sigma_k}$, and term
  ${(f_{ij}(t_{\sigma_{ij}}))}_{\delta_i}$ is the \emph{datum of base sort $\delta_i$
    constructed from term $t_{\sigma_{ij}}$ and function symbol
    $f_{ij} : \sigma_{ij}\to\delta_i$}.  When they are clear from the context or
  immaterial, we leave the arities and sorts implicit and write $x$, $a$,
  $\susp{t}{\rho}$, $[a]t$, $(t_1,\ldots,t_k)$, $f(t)$, \etc
\end{defi}
Given a raw term $t$, the \emph{size of $t$} is the number of nodes of $t$'s abstract
syntax tree.

The raw terms are the inhabitants of the carrier of the free algebra over the set of
variables $\Var$ and over the $\Sort$-sorted conventional signature that consists of the
function symbols in $F$, together with a constant symbol for each atom $a_\alpha$, a unary
symbol that produces moderated terms for each renaming $\rho$ and each sort $\sigma$, a
unary symbol that produces abstractions for each atom $a_\alpha$ and sort $\sigma$, and a
$k$-ary symbol that produces a product of sort $\sigma_1\times\cdots\times\sigma_k$ for
each sequence of sorts $\sigma_1$, \ldots, $\sigma_k$. (See~\cite{GTWW77} for a classic
presentation of term algebras, initial algebra semantics and free algebras.)

We write $\rTerm{\Sigma}_\sigma$ for the set of raw terms of sort $\sigma$. A raw term $t$
is \emph{ground} iff no variables occur in $t$. We write $\gTerm{\Sigma}_\sigma$ for the
set of ground terms of sort $\sigma$. The sets of raw terms (resp.\ ground terms) of each
sort are mutually disjoint as terms carry sort information. Therefore we sometimes
identify the family ${(\rTerm{\Sigma}_\sigma)}_{\sigma\in\Sort}$ of $\Sort$-indexed raw
terms and the family ${(\gTerm{\Sigma}_\sigma)}_{\sigma\in\Sort}$ of $\Sort$-indexed ground
terms with their respective ranges $\bigcup_{\sigma\in\Sort}\rTerm{\Sigma}_\sigma$ and
$\bigcup_{\sigma\in\Sort}\gTerm{\Sigma}_\sigma$, which we abbreviate as $\rTerm{\Sigma}$
and $\gTerm{\Sigma}$ respectively.

The set $\rTerm{\Sigma}$ of raw terms is a nominal set, with the $\perm{\Atom}$-action and
the support of a raw term given by:
\begin{displaymath}
\begin{array}{cc}\hspace{-.3cm}
\begin{array}[t]{rcl}
\peract{\pi}{x}&=&x\\
\peract{\pi}{a}&=&\pi\,a\\
\peract{\pi}{(\susp{t}{\rho})}&=&\susp{(\peract{\pi}{t})}{\peract{\pi}{\rho}}\\
\peract{\pi}{[a]t}&=&[\pi\,a](\peract{\pi}{t})\\
\peract{\pi}{(t_1,\ldots,t_k)}&=&
(\peract{\pi}{t_1},\ldots,\peract{\pi}{t_k})\\
\peract{\pi}{(f(t))}&=&
f(\peract{\pi}{t}),\\[8pt]
\end{array}&
\begin{array}[t]{rcl}
\supp(x)&=&\emptyset\\
\supp(a)&=&\{a\}\\
\supp(\susp{t}{\rho})&=&\supp(t)\cup\supp(\rho)\\
\supp([a](t))&=&\{a\}\cup\supp(t)\\
\supp((t_1,\ldots,t_k))&=&\supp(t_1)\cup\ldots\cup\supp(t_k)\\
\supp(f(t))&=&\supp(t).
\end{array}
\end{array}
\end{displaymath}

It is straightforward to check that the permutation action for raw terms is
sort-preserving (remember that permutations are also sort-preserving). The set
$\gTerm{\Sigma}$ of ground terms is also a nominal set since it is closed with respect to
the $\perm{\Atom}$-action given above.

Below on the left we introduce the action of renaming for a raw term $t$, which replaces
each occurrence of a free atom $a$ in $t$ by $\ren{a}{\rho}$. On the right we present the
function $\fa:\rTerm{\Sigma}\to\pset_\omega(\Atom)$\label{fun:fa}, which delivers the set
of free atoms in a raw term:
\begin{displaymath}
  \begin{array}{cc}
    \begin{array}[t]{rcl}
      \ren{x}{\rho}&=&x\\
      \ren{a}{\rho}&=&\rho\,a\\
      \ren{(\susp{t}{\rho_1})}{\rho_2}&=&\susp{t}{\rho_1;\rho_2}\\
      \ren{([a]t)}{\rho}&=&[\rho\,a](\ren{t}{\rho})\\
      \ren{(t_1,\ldots,t_k)}{\rho}&=&(\ren{t_1}{\rho},\ldots,\ren{t_k}{\rho})\\
      \ren{(f(t))}{\rho}&=&f(\ren{t}{\rho}),\\[8pt]
    \end{array}&
    \begin{array}[t]{rcl}
      \fa(x)&=&\emptyset\\
      \fa(a)&=&\{a\}\\
      \fa(\susp{t}{\rho})&=&\fa(\ren{t}{\rho})\\
      \fa([a]t)&=&\fa(t)\setminus \{a\}\\
      \fa(t_1,\ldots,t_k)&=&\fa(t_1)\cup \ldots \cup \fa(t_k)\\
      \fa(f(t))&=&\fa(t).
    \end{array}
  \end{array}
\end{displaymath}
Notice that the set of free atoms in a raw term differs from the support of the term. For
instance, $\fa([a](a,b))=\{b\}$, but $\supp([a](a,b))=\{a,b\}$.

\begin{rem}\label{rem:ren-size}
  Let $t$ be a raw term and $\rho$ a renaming. Then the size of $\ren{t}{\rho}$ equals the
  size of $t$, which can be checked in a straightforward way by the definition
  above.\blockqed
\end{rem}

Observe that the action of renaming is equivariant.
\begin{lem}\label{lem:renaming-equivariant}
  Let $t$ be a term, $\rho$ be a renaming and $\pi$ be a permutation. Then,
  $\peract{\pi}{(\ren{t}{\rho})} = \ren{(\peract{\pi}{t})}{\peract{\pi}{\rho}}$.
\end{lem}

As expected, the free atoms of a raw term are contained in its support.
\begin{lem}\label{lem:fa-subset-supp}
  Let $t$ be a raw term. Then $\fa(t)\subseteq\supp(t)$.
\end{lem}
The proof of Lemmas~\ref{lem:renaming-equivariant} and~\ref{lem:fa-subset-supp} are in
Appendix~\ref{ap:terms}.

\begin{exa}[$\pi$-calculus]%
  \label{ex:pi-calculus}
  Consider a signature $\Sigma$ for the $\pi$-calculus~\cite{SW01,CMRG12} given by a
  single atom sort $\Chan$ of channel names, and base sorts $\Proc$ and $\Act$ for
  processes and actions respectively. The function symbols (adapted from~\cite{SW01}) are
  the following:
  \begin{displaymath}
    \begin{array}{lll}
      F=\{
      \begin{array}[t]{l}
	\nullPA:\One\to \Proc,\\
	\tauPA[]:\Proc\to \Proc,\\
	\inPA[]{}:(\Chan\times [\Chan]\Proc)\to \Proc,\\
	\outPA{}{}{}:(\Chan\times \Chan\times \Proc)\to
	\Proc,
      \end{array}
      \begin{array}[t]{l}
	\parPA[]{}:(\Proc\times \Proc)\to \Proc,\\
	\sumPA[]{}:(\Proc\times \Proc)\to \Proc,\\
	\repPA[]:\Proc\to \Proc,\\
	\newPA[]:[\Chan]\Proc\to \Proc,
      \end{array}
      \begin{array}[t]{l}
	\tauAA:\One\to\Act,\\
	\inAA[]{}:(\Chan\times \Chan)\to\Act,\\
	\outAA[]{}:(\Chan\times \Chan)\to\Act,\\
	\boutAA[]{}:(\Chan\times \Chan)\to\Act\quad\}.
      \end{array}
    \end{array}
  \end{displaymath}

  Recalling terminology from~\cite{SW01,CMRG12}, $\nullPA$ stands for inaction,
  $\tauPA[p]$ for the internal action after which process $p$ follows, $\inPA[a]{[b]p}$
  for the input at channel $a$ where the input name is bound to $b$ in the process $p$
  that follows, $\outPA[a]{b}{p}$ for the output of name $b$ through channel $a$ after
  which process $p$ follows, $\parPA[p]{q}$ for parallel composition, $\sumPA[p]{q}$ for
  nondeterministic choice, $\repPA[p]$ for parallel replication, and $\newPA[{[a]p}]$ for
  the restriction of channel $a$ in process $p$ ($a$ is private in $p$). Actions and
  processes belong to different sorts. We use $\tauAA$, $\outAA[a]{b}$, $\inAA[a]{b}$ and
  $\boutAA[a]{b}$ respectively for the internal action, the output action, the input
  action and the bound output action.

  The set of terms of the $\pi$-calculus corresponds to the subset of ground terms over
  $\Sigma$ of sort $\Proc$ and $\Act$ in which no moderated (sub-)terms occur. For % chktex 36
  instance, the process $(\nu b)(\overline{a}b.0)$ corresponds to the ground term
  $\newPA[{[b](\outPA[a]{b}{\nullPA})}]$, whose support is $\{a,b\}$. Both free and bound
  channel names (such as the $a$ and $b$ respectively in the example process) are
  represented by atoms. The set of ground terms also contains generalised processes and
  actions with moderated (sub-)terms $\susp{p}{\rho}$, which stand for a delayed renaming % chktex 36
  $\rho$ that ought to be applied to a term $p$, \eg\
  $\newPA[\susp{([b](\outPA[a]{b}{\nullPA}))}{\rho}]$.\blockqed
\end{exa}

Raw terms allow variables to occur in the place of any ground subterm. The variables
represent \emph{unknowns}, and should be mistaken with neither free nor bound channel
names. For instance, the raw term $\newPA[{[b](\outPA[a]{b}{x})}]$ represents a
$\pi$-calculus process $(\nu b)(\overline{a}b.P)$ where the $x$ is akin to the
meta-variable $P$, which stands for some unknown process. The process
$(\nu b)(\overline{a}b.P)$ unifies with $(\nu b)(\overline{a}b.0)$ by replacing $P$ with
$0$. In the nominal setting, the raw term $\newPA[{[b](\outPA[a]{b}{x})}]$ unifies with
ground term $\newPA[{[b](\outPA[a]{b}{\nullPA})}]$, by means of a \emph{substitution}
$\varphi$ such that $\varphi(x)=\nullPA$. Formally, substitutions are defined below.

\begin{defi}[Substitution]%
  \label{def:substitution}
  A \emph{substitution} $\varphi:\Var\to_\fs\rTerm{\Sigma}$ is a sort-preserving, finitely
  supported function from variables to raw terms. The \emph{domain $\dom(\varphi)$ of a
    substitution $\varphi$} is the set ${\{x\mid \varphi(x)\not=x\}}$. A substitution
  $\varphi$ is \emph{ground} iff $\varphi(x)\in\gTerm{\Sigma}$ for every variable
  $x\in\dom(\varphi)$.
\end{defi}

The set of substitutions is a nominal set. The \emph{extension to raw terms
  $\overline{\varphi}$ of substitution $\varphi$} is the unique homomorphism induced by
$\varphi$ from the free algebra $\rTerm{\Sigma}$ to itself, which coincides with the
function given by:
\begin{displaymath}
\begin{array}[t]{rcl}
\overline{\varphi}(x)&=&\varphi(x)\\
\overline{\varphi}(a)&=&a\\
\overline{\varphi}(\susp{t}{\rho})&=&\susp{\overline{\varphi}(t)}{\rho}\\
\overline{\varphi}([a]t)&=&[a](\overline{\varphi}(t))\\
\overline{\varphi}(t_1,\ldots,t_k)&=&(\overline{\varphi}(t_1),\ldots,
\overline{\varphi}(t_k))\\
\overline{\varphi}(f(t))&=&f(\overline{\varphi}(t)).
\end{array}
\end{displaymath}
Given substitutions $\varphi$ and $\gamma$ we write $\varphi\circ\gamma$ for their
composition, which is defined as follows: For every variable $x$,
$(\varphi\circ\gamma)(x)=\overline{\varphi}(t)$ where $\gamma(x)=t$. It is straightforward
to check that
$(\overline{\varphi\circ\gamma})(t) = \overline{\varphi}(\overline{\gamma}(t))$. We note
that our definition of substitution is different from those in both~\cite{UPG04,CP07},
where the authors consider delayed permutations instead of delayed renamings, and where
their substitution function performs the delayed permutations of the moderated terms
\emph{on-the-fly}.

\begin{lem}[Extension to raw terms is equivariant]%
  \label{lem:extension-equivariant}
  Let $\varphi$ be a substitution and $\pi$ a permutation. Then,
  $\peract{\pi}{\overline{\varphi}}= \overline{\peract{\pi}{\varphi}}$.
\end{lem}
\begin{proof}
  We prove $(\peract{\pi}{\overline{\varphi}})(t)= \overline{\peract{\pi}{\varphi}}(t)$ by
  induction on the structure of raw term $t$. By conjugation,
  \begin{displaymath}
    \begin{array}{l}
      (\peract{\pi}{\overline{\varphi}})(x) =
      \peract{\pi}{\overline{\varphi}(\peract{\pi^{-1}}{x})} =
      \peract{\pi}{\overline{\varphi}(x)} =
      \peract{\pi}{\varphi(x)}\\
      = \peract{\pi}{\varphi(\peract{\pi^{-1}}{x})} =
      (\peract{\pi}{\varphi})(x) = \overline{\peract{\pi}{\varphi}}(x)
    \end{array}
  \end{displaymath}
  and the lemma holds for the base case $t=x$. Similarly,
  \begin{displaymath}
    \begin{array}{l}
      (\peract{\pi}{\overline{\varphi}})(a) =
      \peract{\pi}{\overline{\varphi}(\peract{\pi^{-1}}{a})} =
      \peract{\pi}{(\peract{\pi^{-1}}{a})} = a =
      \overline{\peract{\pi}{\varphi}}(a)
    \end{array}
  \end{displaymath}
  and the lemma holds for the base case $t=a$. The rest of the cases are straightforward
  by induction.
\end{proof}

It is easy to check that the support of $\overline{\varphi}$ coincides with the support of
$\varphi$. By the above lemma, the set of extended substitutions is also a nominal set,
since it is closed with respect to the $\perm{\Atom}$-action. Hereafter we sometimes write
$\varphi(t)$, where $t$ is a raw term, instead of $\overline{\varphi}(t)$. We may also
write $\varphi^\pi$ instead of $\peract{\pi}{\overline{\varphi}}$ or
$\overline{\peract{\pi}{\varphi}}$ for short\label{pg:varphi-pi}.

The following result highlights the relation between substitution and the permutation
action.

\begin{lem}[Substitution and permutation action]%
  \label{lem:subst-perm}
  Let $\varphi$ be a substitution, $\pi$ a permutation and $t$ a raw term.  Then,
  $\peract{\pi}{\varphi(t)} = \varphi^\pi(\peract{\pi}{t})$.
\end{lem}
\begin{proof}
  By definition of $\varphi^\pi$, we have that
  $\varphi^\pi(\peract{\pi}{t}) =
  \peract{\pi}{\varphi(\peract{\pi^{-1}}{(\peract{\pi}{t})})} = \peract{\pi}{\varphi(t)}$
  and we are done.
\end{proof}

Our goal is to give meaning to ground terms in nominal sets. To this end, we need a
suitable class of algebraic structures that can be used to give an \emph{interpretation}
of those ground terms.

\begin{defi}[$\Sigma$-structure]%
  \label{def:sigma-structure}
  Let $\Sigma=(\Delta,A,F)$ be a signature. A \emph{$\Sigma$-structure $M$} consists of a
  nominal set $M[\![\sigma]\!]$ for each sort $\sigma$ defined as follows
  \begin{displaymath}
    \begin{array}{rcl}
      M[\![\alpha]\!]&=&\Atom_\alpha\\
      M[\![[\alpha]\sigma]\!]&=&[\Atom_\alpha](M[\![\sigma]\!])\\
      M[\![\sigma_1\times\cdots\times\sigma_k]\!]&=&
         M[\![\sigma_1]\!]\times\cdots\times M[\![\sigma_k]\!],
    \end{array}
  \end{displaymath}
  where the $M[\![\delta_i]\!]$ with $\delta_i\in\Delta$ are given, as well as an
  equivariant function $M[\![f_{ij}]\!] : M[\![\sigma_{ij}]\!]\to M[\![\delta_i]\!]$ for
  each symbol ${(f_{ij})}_{\sigma_{ij}\to\delta_i}\in F$.
\end{defi}

The notion of $\Sigma$-structure adapts that of $\Sigma$-structure in~\cite{CP07} to our
sorting convention with atom and abstraction sorts. The $\Sigma$-structures characterise a
range of interpretations of ground terms into elements of nominal sets, such that any sort
$\sigma$ gives rise to the expected nominal set, \ie\ atom sorts give rise to sets of
atoms, abstraction sorts give rise to sets of atom abstractions, and product sorts give
rise to finite products of nominal sets.

Next we define the \emph{interpretation of a ground term in a $\Sigma$-structure}, which
resembles the \emph{value of a term} in~\cite{CP07}.

\begin{defi}[Interpretation of ground terms in a $\Sigma$-structure]%
  \label{def:interpretation}
  Let $\Sigma$ be a signature and $M$ be a $\Sigma$-structure. The \emph{interpretation
    $M[\![p]\!]$ of a ground term $p$ in $M$} is given by:
  \begin{displaymath}
    \begin{array}[t]{rcl}
      M[\![a]\!]&=&a\\
      M[\![\susp{p}{\rho}]\!]&=&M[\![\ren{p}{\rho}]\!]\\
      M[\![[a]p]\!]&=& \langle a\rangle(M[\![p]\!])\\
      M[\![(p_1,\ldots,p_k)]\!]&=&
	(M[\![p_1]\!],\ldots,M[\![p_k]\!])\\
      M[\![f(p)]\!]&=&M[\![f]\!](M[\![p]\!]).
    \end{array}
  \end{displaymath}
\end{defi}

Notice that the moderated ground term $\susp{p}{\iota}$ is syntactically different from
the ground term $p$, although both terms have the same interpretation in any
$\Sigma$-structure since $\ren{p}{\iota}=p$.

The next lemma states that interpretation in a $\Sigma$-structure is equivariant and
highlights the relation between interpretation and moderated terms.

\begin{lem}[Interpretation and moderated terms]%
  \label{lem:interp-mod}
  Let $M$ be a $\Sigma$-structure. Interpretation in $M$ is equivariant, that is,
  $\peract{\pi}{M[\![p]\!]} = M[\![\peract{\pi}{p}]\!]$ for every ground term $p$ and
  permutation $\pi$.
\end{lem}
\begin{proof}
  Recall that the size of a ground term is the number of nodes of its abstract syntax
  tree. We proceed by induction on the size of $p$. The base case $p=a$ is trivial.

  If $p=\susp{q}{\rho}$, then
  $\peract{\pi}{M[\![\susp{q}{\rho}]\!]} = \peract{\pi}{M[\![\ren{q}{\rho}]\!]}$, which by
  the induction hypothesis is equal to
  \begin{displaymath}
    \begin{array}{l}
      M[\![\peract{\pi}{(\ren{q}{\rho})}]\!] =
      M[\![\ren{q}{\rho;\pi}]\!] =
      M[\![\ren{q}{\pi;\pi^{-1};\rho;\pi}]\!] =
      M[\![\ren{q}{\pi;\peract{\pi}{\rho}}]\!] =
      M[\![\ren{(\ren{q}{\pi})}{\peract{\pi}{\rho}}]\!]\\
      = M[\![\ren{(\peract{\pi}{q})}{\peract{\pi}{\rho}}]\!]
      = M[\![\susp{(\peract{\pi}{q})}{\peract{\pi}{\rho}}]\!]
      = M[\![\peract{\pi}{(\susp{q}{\rho})}]\!].
    \end{array}
  \end{displaymath}

  If $p=[a]q$, then
  $\peract{\pi}{M[\![[a]q]\!]} = \peract{\pi}{(\langle a\rangle(M[\![q]\!]))}  =
  \langle\peract{\pi}{a}\rangle(\peract{\pi}{(M[\![q]\!])})$,
  which, by the induction hypothesis, is equal to
  $ \langle\peract{\pi}{a}\rangle(M[\![\peract{\pi}{q}]\!]) =
  M[\![[\peract{\pi}{a}](\peract{\pi}{q})]\!] = M[\![\peract{\pi}{([a]q)}]\!]$.

  The remaining cases are straightforward by the induction hypothesis.
\end{proof}

Moreover, the interpretation $M[\![\susp{p}{\pi}]\!]$ of a suspension whose delayed
renaming is a permutation $\pi$, is equal to the permutation $\peract{\pi}{M[\![p]\!]}$.
Indeed, by Definition~\ref{def:interpretation} and Lemma~\ref{lem:interp-mod} we have
$M[\![\susp{p}{\pi}]\!] = M[\![\peract{\pi}{p}]\!] = \peract{\pi}{M[\![p]\!]}$

Finally, we introduce the $\Sigma$-structure $\mathit{NT}$, which formalises the set of
\emph{nominal terms}.

\begin{defi}[$\Sigma$-structure for nominal terms]%
  \label{def:NT}
  Let $\Sigma$ be a signature. The \emph{$\Sigma$-structure $\mathit{NT}$ for nominal
    terms} is given by the least tuple $(\NT{\delta_1},\ldots,\NT{\delta_n})$ satisfying
  \begin{displaymath}
    \NT{\delta_i} = \NT{\sigma_{i1}} + \cdots + \NT{\sigma_{im_i}}\quad
    \text{for each base sort $\delta_i\in\Delta$, and}
  \end{displaymath}
  $\NT{f_{ij}}=\inj_j:\NT{\sigma_{ij}}\to\NT{\delta_i}$, for each function symbol
  $f_{ij}\in F$.

  In the conditions above, the `less than or equal to' relation for tuples is pointwise
  set inclusion. The $\NT{f_{ij}}$ is the $j$th injection of the $i$th component in
  $(\NT{\delta_1},\ldots,\NT{\delta_n})$.
\end{defi}

Nominal terms represent alpha-equivalence classes of raw terms by using the atom
abstractions of Definition~\ref{def:atom-abstraction}.
\begin{defi}[Nominal terms]
  Let $\Sigma$ be a signature. The set $\nTerm{\Sigma}_\sigma$ of \emph{nominal terms over
    $\Sigma$ of sort $\sigma$} is the domain of interpretation of the ground terms of sort
  $\sigma$ in the $\Sigma$-structure $\mathit{NT}$, that is,
  $\nTerm{\Sigma}_\sigma=\NT{\sigma}$.
\end{defi}

We sometimes write $p$, $\ell$ instead of $\NT{p}$, $\NT{\ell}$ when it is clear from the
context that we are referring to the interpretation into nominal terms of ground terms $p$
and $\ell$.

\subsection*{Nominal Terms and Nominal Algebraic Datatypes}
We check that the nominal sets $\nTerm{\Sigma}_\sigma$ coincide (up to isomorphism) with
the nominal algebraic datatypes of Definition~8.9 in~\cite{Pit13}. We first illustrate the
nominal terms by means of the signature $\Sigma$ for the $\pi$-calculus in
Example~\ref{ex:pi-calculus}.
\begin{exa}
  The $\Sigma$-structure $\mathit{NT}$ is given by the least pair $(\NT{\Proc},\NT{\Act})$
  of nominal sets satisfying the following set equations
  \begin{displaymath}
    \begin{array}{rcl}
      \NT{\Proc} &=& \NT{\One} +  \NT{\Proc}
                     + \NT{\Chan\times[\Chan]\Proc}
                     + \NT{\Chan\times\Chan\times\Proc}\\
                 & & {} + \NT{\Proc\times\Proc} + \NT{\Proc\times\Proc} +  \NT{\Proc}
                     + \NT{[\Chan]\Proc}\\
                 &=& \{()\} + \NT{\Proc}
                     + (\Atom_{\Chan}\times[\Atom_{\Chan}](\NT{\Proc}))
                     + (\Atom_{\Chan}\times\Atom_{\Chan}\times \NT{\Proc})\\
                 & & {} + (\NT{\Proc}\times\NT{\Proc}) + (\NT{\Proc}\times\NT{\Proc})
                     +  \NT{\Proc} + [\Atom_{\Chan}](\NT{\Proc}),\\[8pt]
      \NT{\Act} &=& \NT{\One} + \NT{\Chan\times\Chan} + \NT{\Chan\times\Chan}
                    + \NT{\Chan\times\Chan}\\
                 &=& \{()\} + (\Atom_{\Chan}\times\Atom_{\Chan})
                     + (\Atom_{\Chan}\times\Atom_{\Chan})
                     + (\Atom_{\Chan}\times\Atom_{\Chan}),
    \end{array}
  \end{displaymath}
  together with an equivariant function for each function symbol in $F$ (we only show a
  few)
  \begin{displaymath}
    \begin{array}{rcl}
      \NT{\nullPA} = \inj_{1} &:& \{()\} \to \NT{\Proc}\\
      \NT{\tauPA[]} = \inj_2 &:& \NT{\Proc} \to \NT{\Proc}\\
      \NT{\outPA[]{}{}} = \inj_{4} &:& \Atom_{\Chan} \times [\Atom_{\Chan}](\NT{\Proc})
                                       \to \NT{\Proc}\\
      \NT{\newPA[]}  = \inj_8 &:& [\Atom_\Chan](\NT{\Proc}) \to \NT{\Proc}\\
      \NT{\tauAA} = \inj_{1} &:& \{()\} \to \NT{\Act}\\
      \NT{\boutAA[]{}} =\inj_{4} &:& \Atom_{\Chan} \times \Atom_{\Chan} \to \NT{\Act}.
    \end{array}
  \end{displaymath}
  For example, the process $(\nu b)(\overline{a}b.0)$ is encoded as the ground term
  $\newPA[{[b](\outPA[a]{b}{\nullPA})}]$, whose interpretation in $\mathit{NT}$ is
  $\inj_8(\langle b\rangle(\inj_4(a,b,\inj_1())))$.\blockqed
\end{exa}

\begin{rem}\label{rem:nominal-not-disjoint}
  Recall that the constructor $\inj_j$ for disjoint union has the polymorphic type
  \begin{displaymath}
    \inj_j : \forall (S_1 + \cdots + S_m).S_j \to S_1 + \cdots + S_m,
    \qquad\textup{where}\ j\leq m.
  \end{displaymath}
  Therefore, a nominal term may have `polymorphic sort' and the sets of nominal terms of
  each sort may not be mutually disjoint. For instance, both ground terms $\nullPA$ and
  $\tauAA$ have the same interpretation $\inj_1()$ in $\mathit{NT}$. However, each of the
  $\NT{\nullPA}$ and $\NT{\tauAA}$ live in different components of the carrier
  $(\NT{\Proc},\NT{\Act})$ of the $T$-algebra induced by the $\Sigma$-structure
  $\mathit{NT}$ and, by all means, the sort information is never lost. Here we are not
  concerned with this technical subtlety and, at any rate, we can always determine the
  `monomorphic sort' of a given nominal term by using implicit type parameters (within
  curly braces) that fix the set $S_1 + \cdots + S_n$ over which each constructor $\inj_j$
  is universally quantified, \ie\
  $\NT{\nullPA}=\inj_1\{\NT{\Proc}\}()$.\blockqed
\end{rem}

The nominal term with implicit type parameters that corresponds to process
$(\nu b)(\overline{a}b.0)$ is
$\inj_8\{\NT{\Proc}\}(\langle
b\rangle(\inj_4\{\NT{\Proc}\}(a,b,\inj_1\{\NT{\Proc}\}())))$.

The remainder of this section shows that the nominal terms are connected to the elements
of the nominal algebraic data types of Definition~8.9 in~\cite{Pit13}. We follow closely
the exposition on initial algebraic semantics for nominal algebraic data types in~\cite{Pit13}. The reader is advised to read Sections~8.3 and 8.4 of~\cite{Pit13}
alongside.

Let $\mathbf{Nom}^n=\mathbf{Nom}\times\ldots_{\text{$n$ times}}\ldots\times\mathbf{Nom}$
be the $n$-product category and let $T:\mathbf{Nom}^n\to\mathbf{Nom}^n$ be the nominal
algebraic functor induced by a signature $\Sigma$ (see Section~8.3 of~\cite{Pit13}), which
we describe next. Given an $n$-tuple $S = (S_1,\ldots,S_n)$ of nominal sets, each sort
$\sigma$ gives rise to a nominal set $[\![\sigma]\!]S$ defined by:
\begin{displaymath}
\begin{array}{rcl}
\lbrack\!\lbrack\delta_i\rbrack\!\rbrack S &=& S_i\\
\lbrack\!\lbrack\alpha\rbrack\!\rbrack S &=& \Atom_\alpha\\
\lbrack\!\lbrack[\alpha]\sigma\rbrack\!\rbrack S &=&
[\Atom_\alpha](\lbrack\!\lbrack\sigma\rbrack\!\rbrack S)\\
\lbrack\!\lbrack\sigma_1\times\cdots\times\sigma_k\rbrack\!\rbrack S &=&
\lbrack\!\lbrack\sigma_1\rbrack\!\rbrack S\times\cdots\times
\lbrack\!\lbrack\sigma_k\rbrack\!\rbrack S.
\end{array}
\end{displaymath}
Let the sorts $\sigma_{ij}$ be such that $f_{ij}:\sigma_{ij}\to\delta_i$ are the function
symbols of signature $\Sigma$. The nominal algebraic functor $T$ has components
$T_i:\mathbf{Nom}^n\to \mathbf{Nom}$ mapping each $S=(S_1,\ldots,S_m)\in \mathbf{Nom}^n$
to $T_i S=[\![\sigma_{i1}]\!]S + \cdots + [\![\sigma_{im_i}]\!]S$, and similarly for
$n$-tuples of equivariant functions.

A $\Sigma$-structure $M$ gives rise to a $T$-algebra whose carrier is the $n$-tuple of
nominal sets $S=(M[\![\delta_1]\!],\ldots,M[\![\delta_n]\!])$ and whose morphism is the
$n$-tuple of equivariant functions $F=(F_1,\ldots,F_n)$ where
$F_i(\inj_j s)=M[\![f_{ij}]\!](s)$ for each $s\in S_i$.

\begin{thm}%
  \label{thm:nonminal-terms-noiminal-algebraic-datatypes}
  The nominal sets $\nTerm{\Sigma}_\sigma$ coincide (up to isomorphism) with the nominal
  algebraic datatypes of Definition~8.9 in~\cite{Pit13}.
\end{thm}
\begin{proof}
  Let $D=(\NT{\delta_1},\ldots,\NT{\delta_n})$ together with morphism
  $I = (I_1,\ldots,I_n)$ such that $I_i(\inj_j s)=\NT{f_{ij}}(s)$ be the $T$-algebra
  induced by $\Sigma$-structure $\mathit{NT}$. It is routine to check that $I$ maps $T(D)$
  to $D$, where $T$ is the nominal algebraic functor induced by signature $\Sigma$, and
  that the morphism $I$ coincides with the identity. Since $D$ is the least tuple
  satisfying this condition, the tuple coincides with the least fixed point of functor
  $T$. By a well known result by Lambek~\cite{Lam68}, $(D,I)$ constitutes the initial
  $T$-algebra. The theorem follows by Theorem~8.15 in~\cite{Pit13}.
\end{proof}

%%%%%%%%%%%%%%%%%%%%%%%%%%%%%%%%%%%%%%%%%%%%%%%%%%%%%%%%%%%%%%%%%%%%%%%%%%%
%% Rule format for equivariance
\section{Specifications of NRTSs}%
\label{sec-specification-nrts}
The NRTSs of Definition~\ref{def:nrts} are meant to be a model of computation for
calculi with name-binding operators and state/residual presentation. In this
section we present syntactic specifications for NRTSs. We start by defining nominal
residual signatures.

\begin{defi}[Nominal residual signature]%
  \label{def:nominal-signature}
  A \emph{nominal residual signature} (a residual signature for short) is a quintuple
  $\Sigma=(\Delta,\ASet,\sigma,\varrho,F)$ such that $(\Delta,\ASet,F)$ is a nominal
  signature and $\sigma$ and $\varrho$ are distinguished nominal sorts over $\Delta$ and
  $\ASet$, which we call \emph{state sort} and \emph{residual sort} respectively. We say
  that $\nTerm{\Sigma}_\sigma$ is the set of \emph{states} and $\nTerm{\Sigma}_\varrho$ is
  the set of \emph{residuals}.

  Let $\mathcal{T}=(S,R,\rel{})$ be an NRTS and $\Sigma=(\Delta,A,\sigma,\varrho,F)$ be a
  residual signature. We say that $\mathcal{T}$ is an NRTS \emph{over signature} $\Sigma$
  iff the sets of states $S$ and residuals $R$ coincide with the sets of nominal terms of
  state sort $\nTerm{\Sigma}_\sigma$ and residual sort $\nTerm{\Sigma}_\varrho$ respectively.
\end{defi}

Our next goal is to introduce syntactic specifications of NRTSs, which we call nominal
residual transition system specifications adapting a terminology introduced by Groote and
Vaandrager~\cite{GV92}. To this end, we will make use of residual formulas and freshness
assertions over raw terms, which are defined below.

\begin{defi}[Residual formula and freshness assertion]
  A \emph{residual formula} (a \emph{formula} for short) over a residual signature
  $\Sigma$ is a pair $(s,r)$, where $s\in\rTerm{\Sigma}_\sigma$ and
  $r\in\rTerm{\Sigma}_\varrho$. We use the more suggestive $s\rel{}r$ in lieu of $(s,r)$. A
  formula $s\rel{}r$ is \emph{ground} iff $s$ and $r$ are ground terms.

  A \emph{freshness assertion} (an \emph{assertion} for short) over a signature $\Sigma$
  is a pair $(a,t)$ where $a\in\Atom$ and $t\in\rTerm{\Sigma}$. We will write $a\fra t$ in
  lieu of $(a,t)$.  An assertion is \emph{ground} iff $t$ is a ground term.
\end{defi}
\begin{rem}
  Formulas and assertions are raw syntactic objects, similar to raw terms, which will
  occur in the rules of the nominal residual transition system specifications to be
  defined, and whose purpose is to represent respectively transitions and freshness
  relations involving nominal terms. A formula $s\rel{}r$ (resp.\ an assertion $a\fra t$)
  unifies with a ground formula $\varphi(s)\rel{}\varphi(r)$ (resp.\ a ground assertion
  $a\fra \varphi(t)$), which in turn represents a transition
  $\NT{\varphi(s)}\rel{}\NT{\varphi(r)}$ (resp.\ a freshness relation
  $a\# \NT{\varphi(t)}$).  For the assertions, notice how the symbols $\fra$, $\#$ and
  $\NT{\ }$ interact. The ground assertion $a\fra [a]a$ represents the freshness relation
  $a\#\NT{[a]a}$, which is true. On the other hand, the freshness relation $a\#[a]a$ is
  false because $a\in\supp([a]a)$.\blockqed
\end{rem}

Permutation action, substitution and the function $\fa$ extend to residual formulas and
freshness assertions in the expected way, \ie\
\begin{displaymath}
  \begin{array}{rcl}
    \peract{\pi}{(s\rel{}r)}&=&\peract{\pi}{s}\rel{}\peract{\pi}{r}\\
    \peract{\pi}{(a\fra t)}&=&\peract{\pi}{a}\fra \peract{\pi}{t}\\[8pt]
    \varphi(s\rel{}r)&=&\varphi(s)\rel{}\varphi(r)\\
    \varphi(a\fra t)&=&a\fra\varphi(t)\\[8pt]
    \fa(r\rel{} s)&=&\fa(r) \cup \fa(s)\\
    \fa(a\fra t)&=&\{a\}\cup \fa(t).
  \end{array}
\end{displaymath}
Residual formulas and freshness assertions are elements of nominal sets. The support of a
residual formula (respectively a freshness assertion) is the union of the supports of the
raw terms in it. We write $\supp(t\rel{}t')$ and $\supp(a\fra t)$ for the supports of
residual formula $t\rel{}t'$ and of freshness assertion $a\fra t$ respectively. We write
$b\#(t\rel{} t')$ and $b\#(a\fra t)$ for the freshness relations that involve atom $b$ and
residual formula $t\rel{}t'$ and freshness assertion $a\fra t$ respectively.

\begin{defi}[Nominal residual transition system specification]%
  \label{def:NRTSS}
  Let $\Sigma$ be a residual signature $(\Delta,\ASet,\sigma,\varrho,F)$. A
  \emph{transition rule} \emph{over} $\Sigma$ (a \emph{rule}, for short) is of the form
  \begin{mathpar}
    \inferrule*[right={,}]
    {\{u_i\rel{}u'_i\mid i\in I\}\quad \{a_j\fra v_j\mid j\in J\}}
    {t\rel{}t'}
  \end{mathpar}
  abbreviated as $H,\nabla/t\rel{}t'$, where $H={\{u_i\rel{}u'_i\mid i\in I\}}$ is a
  finitely supported set of formulas over $\Sigma$ (we call $H$ the set of
  \emph{premisses}) and where $\nabla=\{a_j\fra v_j\mid j\in J\}$ is a finite set of
  assertions over $\Sigma$ (we call $\nabla$ the \emph{freshness environment}). We say
  formula $t\rel{}t'$ over $\Sigma$ is the \emph{conclusion}, where $t$ is the
  \emph{source} and $t'$ is the \emph{target}. A rule is an \emph{axiom} iff it has an
  empty set of premisses. Note that axioms might have a non-empty freshness environment.

  A \emph{nominal residual transition system specification over} $\Sigma$ (abbreviated to
  NRTSS) is a set of transition rules over $\Sigma$.
\end{defi}

Permutation action and substitution extend to rules in the expected way; they are applied
to each of the formulas and freshness assertions in the rule.

Notice that the rules of an NRTSS are elements of a nominal set. The support of a rule
$H,\nabla/t\rel{}t'$ is the union of the support of $H$, the support of $\nabla$ and the
support of $t\rel{}t'$. In the sequel we write $\supp(\textsc{Ru})$ for the support of
rule \textsc{Ru}, and $a\#\textsc{Ru}$ for a freshness relation involving atom $a$ and
rule \textsc{Ru}. Observe that the set $H$ of premisses of a rule may be infinite, but its
support must be finite. However, the freshness environment $\nabla$ must be finite in
order to make the simplification rules of Definition~\ref{def:simplification} to follow
terminating. These simplification rules will be used in Section~\ref{sec-rule-format-nrts}
to define the rule format in Definition~\ref{def:alpha-conv-format}.

Let $\R$ be an NRTSS\@. We say that the formula $s\rel{}r$ \emph{unifies} with rule \Ru\ in
$\R$ iff \Ru\ has conclusion $t\rel{}t'$ and $s\rel{}r$ is a substitution instance of
$t\rel{}t'$. If $s$ and $r$ are ground terms, we also say that transition
$\NT{s}\rel{}\NT{r}$ unifies with \Ru.

\begin{defi}
  Let $a\fra t$ be a freshness assertion and $\varphi$ a ground substitution. We say that
  $\varphi(a\fra t)$ \emph{holds} iff the freshness relation $a \# \NT{\varphi(t)}$
  holds.

  Let $\nabla = \{a_j\fra t_j \mid j\in J\}$ be a freshness environment. We say that
  $\varphi(\nabla)$ \emph{holds} iff the conjunction
  $ \bigwedge_{j\in J}(a_j\# \NT{\varphi(t_j)})$ holds.
\end{defi}

\begin{defi}[Proof tree]%
  \label{def:proof-tree}
  Let $\Sigma$ be a residual signature and $\R$ be an NRTSS over $\Sigma$. A \emph{proof
    tree} in $\R$ of a transition $\NT{s}\rel{}\NT{r}$ is an upwardly branching rooted
  tree without paths of infinite length whose nodes are labelled by transitions such that
  \begin{enumerate}[label={(\roman*)},ref={(\roman*)}]
  \item the root is labelled by $\NT{s}\rel{}\NT{r}$, and
  \item if $K=\{\NT{q_i}\rel{}\NT{q_i'}\mid i\in I\}$ is the set of labels of the nodes
    directly above a node with label ${\NT{p}\rel{}\NT{p'}}$, then there exist a rule
    \begin{mathpar}
      \inferrule*
      {\{u_i\rel{}u'_i\mid i\in I\} \qquad
        \{a_j\fra v_j\mid j\in J\}}
      {t\rel{}t'}
    \end{mathpar}
    in $\R$ and a ground substitution $\varphi$ such that $\varphi(t\rel{}t')=p\rel{}p'$,
    for each $i\in I$ $\varphi(u_i\rel{}u_i')=q_i\rel{}q_i'$, and
    $\varphi(\{a_j\fra t_j\mid j\in J\})$ holds.
  \end{enumerate}
  We say that $\NT{s}\rel{}\NT{r}$ is \emph{provable} in $\R$ iff it has a proof
  tree in $\R$. The transition relation specified by $\R$ consists of all the
  transitions that are provable in $\R$.
\end{defi}

The nodes of a proof tree are labelled by transitions, which contain nominal terms (\ie\
syntactic objects that use the atom abstractions of
Definition~\ref{def:atom-abstraction}). The use of nominal terms in a proof tree captures
the convention in typical nominal calculi of considering terms `up to alpha-equivalence'.

\begin{exa}
  Consider the residual signature with base sort $\mathsf{b}$, atom sort $\mathsf{a}$, two
  function symbols $f,g$ with arity $[\mathsf{a}]\mathsf{a}\to \mathsf{b}$ and state and
  residual sorts equal to $\mathsf{b}$. Let $\R$ be the NRTSS defined by the rules:
  \begin{mathpar}
    \inferrule*[right=Ax]
    { }{g(x) \rel{} g(x)}
    \and
    \inferrule*[right=Ru {\normalsize $,\quad\textup{where}\ a,b\in \Atom_\mathsf{a}$.}] % chktex 1
    {g([a]a) \rel{} g([b]b)\quad a\fra b}
    {f([a]a) \rel{} f([b]b)}
  \end{mathpar}
  The nominal term $\NT{f([a]a)}$ is equal to $\NT{f([b]b)}$, and $\NT{g([a]a)}$ is equal
  to $\NT{g([b]b)}$, so the transition $\NT{f([a]a)}\rel{} \NT{f([a]a)}$ is provable with
  the following proof tree, where rule \textsc{Ax} is instantiated using a ground
  substitution $\varphi$ such that $\varphi(x)=[a]a$:\footnote{Extending the existing
    convention to our notion of proof tree, we depict proof trees as trees of inference
    rules where the conclusion and premisses in each rule are replaced by \emph{the
      transitions denoted by} their substitution instances, and where the freshness
    assertions in each rule are replaced by \emph{the freshness relations denoted by}
    their substitution instances.}
  \begin{mathpar}
    \inferrule*[right=Ru{\normalsize .}]
    {
      \inferrule*[right=Ax]
      { }
      {\NT{g([a]a)}\rel{}\NT{g([b]b)}}
      \quad a\#b
    }
    {\NT{f([a]a)}\rel{}\NT{f([b]b)}}
  \end{mathpar}
  Intuitively, the freshness assertion $a\fra b$ in rule \textsc{Ru} is superfluous
  because it references atoms $a$ and $b$, which do not occur free in the rule (\ie\
  $a,b\not\in\fa(f([a]a)\rel{}f([b]b))$ and
  $a,b\not\in\fa(g([a]a)\rel{}g([b]b))$).\blockqed
\end{exa}

The fact that the nodes of a proof tree are labelled by nominal terms is the main
difference between our approach and previous work in nominal structural operational
semantics~\cite{CMRG12,ACGIMR}, nominal rewriting~\cite{UPG04,FG07} and nominal algebra~\cite{GM09}. In all these works, the `up-to-alpha-equivalence' transitions are explicitly
instrumented within the model of computation by adding inference rules that perform
alpha-conversion of raw terms to the specification system.

\section{Rule Formats for NRTSSs}%
\label{sec-rule-format-nrts}
This section defines two rule formats for NRTSSs that ensure that:
\begin{enumerate}[label={(\roman*)},ref={(\roman*)}]
\item an NRTSS induces an equivariant transition relation, and thus an NRTS in the sense
  of Definition~\ref{def:nrts};
\item an NRTSS induces a transition relation which, together with an equivariant function
  $\bn$, corresponds to an NTS of Definition~\ref{def:nts}~\cite{PBEGW15}. For the latter,
  we need to ensure that the induced transition relation is equivariant and satisfies
  \emph{alpha-conversion of residuals} (recall, if $p\rel{}(\ell,p')$ is provable in $\R$
  and $a$ is in the set of binding names of $\ell$, then for every atom $b$ that is fresh
  in $(\ell,p')$ the transition $p\rel{}\peract{\tr{a}{b}}{(\ell,p')}$ is also provable).
\end{enumerate}

\noindent
As a first step, we introduce a rule format ensuring equivariance of the induced
transition relation.

\begin{defi}[Equivariant format]%
  \label{def:equivariant-format}
  Let $\mathcal{R}$ be an NRTSS\@. $\mathcal{R}$ is in \emph{equivariant format} iff the
  rule $\peract{\tr{a}{b}}{\textsc{Ru}}$ is in $\mathcal{R}$, for every rule $\textsc{Ru}$
  in $\R$ and for each $a,b\in\Atom$.
\end{defi}

\begin{lem}%
  \label{lem:equivariant-format}
  Let $\mathcal{R}$ be an NRTSS in equivariant format. For every rule $\textsc{Ru}$ in
  $\mathcal{R}$ and for every permutation $\pi$, the rule $\peract{\pi}{\textsc{Ru}}$ is
  in $\mathcal{R}$.
\end{lem}
\begin{proof}
  The claim follows straightforwardly by Definition~\ref{def:equivariant-format}, since
  each permutation $\pi$ can be expressed as a composition of transpositions
  $\tr{a_1}{b_1}\circ\ldots\circ\tr{a_n}{b_n}$ with $n\geq 0$.
\end{proof}

\begin{thm}[Rule format for NRTSs]%
  \label{thm:rule-format-equivariance}
  Let $\R$ be an NRTSS\@. If $\R$ is in equivariant format then $\R$ induces an NRTS\@.
\end{thm}
\begin{proof}
  We prove that the transition relation induced by $\R$ is equivariant. That is, if
  $\NT{p}\rel{}\NT{p'}$ then $\peract{\pi}{\NT{p}}\rel{}\peract{\pi}{\NT{p'}}$ for every
  permutation $\pi$. We proceed by induction on the height of the proof tree of
  $\NT{p}\rel{}\NT{p'}$. Assume that the last rule used in this proof is
  \begin{mathpar}
    \inferrule*[right=Ru]
    {\{u_i\rel{}u'_i\mid i\in I\}
      \qquad\{a_j\fra v_j\mid j\in J\}}
    {t\rel{}t'}
  \end{mathpar}
  and that, for some ground substitution $\varphi$,
  \begin{enumerate}[label={(\roman*)},ref={(\roman*)}]
  \item the premisses $\NT{\varphi(u_i)}\rel{}\NT{\varphi(u'_i)}$ with $i\in I$ are
    provable in $\R$,
  \item the freshness relations $a_j\#\NT{\varphi(v_j)}$ with $j\in J$ hold, and
  \item $\varphi(t)\rel{}\varphi(t')=p\rel{}p'$.
  \end{enumerate}
  Since $\R$ is in equivariant format, by Lemma~\ref{lem:equivariant-format} $\R$ contains
  the rule
  \begin{mathpar}
    \inferrule*[right=Ru$_\pi$.]
    {\{\pi\cdot u_i\rel{}\pi\cdot u'_i\mid i\in I\}\qquad
      \{\peract{\pi}{a_j}\fra\pi\cdot v_j\mid j\in J\}}
    {\pi\cdot t\rel{}\pi\cdot t'}
  \end{mathpar}
  Our goal now is to show that the transition
  $\peract{\pi}{\NT{p}}\rel{}\peract{\pi}{\NT{p'}}$ is provable using rule
  $\textsc{Ru}_\pi$ and substitution $\varphi^\pi$ defined on
  page~\pageref{pg:varphi-pi}. Let $j\in J$. By Lemma~\ref{lem:subst-perm} we know that
  $\peract{\pi}{\varphi(v_j)}={\varphi^\pi(\pi\cdot v_j)}$. Moreover, since $\#$ is
  equivariant, by Lemma~\ref{lem:interp-mod}, the freshness relation
  $\peract{\pi}{a_j}\#{\NT{\varphi^\pi(\pi\cdot v_j)}}$ holds. Assume now that $i\in I$.
  We know that the premiss
  $\peract{\pi}{\NT{\varphi(u_i)}}\rel{}\peract{\pi}{\NT{\varphi(u'_i)}}$ is provable in
  $\R$ by the induction hypothesis ($I=\emptyset$ corresponds to the base case, \ie\ a
  rule without premisses). By Lemmas~\ref{lem:subst-perm} and~\ref{lem:interp-mod}, this
  premiss is equal to
  ${\NT{\varphi^\pi(\pi\cdot u_i)}} \rel{} \NT{\varphi^\pi(\pi\cdot u'_i)}$.  Therefore,
  the transition ${\peract{\pi}{\NT{p}}\rel{}\peract{\pi}{\NT{p'}}}$ is provable using
  rule \Ru$_\pi$ and substitution $\varphi^\pi$ because it is equal to
  $\NT{\varphi^\pi(\pi\cdot t)}\rel{}\NT{\varphi^\pi(\pi\cdot t')}$ by
  Lemmas~\ref{lem:subst-perm} and~\ref{lem:interp-mod}.
\end{proof}

\begin{rem}\label{rem:transposed-proof-tree}
  It is straightforward to check that the proof tree of transition
  $\NT{\peract{\tr{a}{b}}{p}}\rel{}\NT{\peract{\tr{a}{b}}{p'}}$ obtained in the proof
  above coincides with the proof tree of
  $\peract{\tr{a}{b}}{(\NT{p})}\rel{}\peract{\tr{a}{b}}{(\NT{p'})}$, where atoms $a$ and
  $b$ have been transposed. Both proof trees have the same height.\blockqed
\end{rem}

Before introducing a rule format ensuring alpha-conversion of residuals, we adapt to our
freshness environments the simplification rules and the entailment relation of
Definition~10 and Lemma~15 in~\cite{FG07}.

\begin{defi}[Simplification of freshness environments]%
  \label{def:simplification}
  Consider a signature $\Sigma$. The following rules, where $\nabla$ is a freshness
  environment over $\Sigma$, define \emph{simplification of freshness environments}:
  \begin{displaymath}
    \begin{array}[t]{rcl}
      \{a\fra b\}\cup \nabla&\Longrightarrow& \nabla\quad\text{if $a\neq b$}\\
      \{a\fra \susp{b}{\rho}\}\cup \nabla&\Longrightarrow&
         \{a\fra \rho\ b\}\cup \nabla\\
      \{a\fra \susp{(\susp{t}{\rho_1})}{\rho}\}\cup \nabla&\Longrightarrow&
         \{a\fra \susp{t}{\rho_1;\rho}\}\cup \nabla\\
      \{a\fra \susp{([b]t)}{\rho}\}\cup \nabla&\Longrightarrow&
         \{a \fra [\rho\ b](\susp{t}{\rho})\} \cup \nabla\\
      \{a\fra \susp{(t_1,\ldots,t_k)}{\rho}\}\cup\nabla&\Longrightarrow&
         \{a\fra \susp{t_1}{\rho},\ldots,a\fra \susp{t_k}{\rho}\}\cup\nabla\\
      \{a\fra \susp{(f(t))}{\rho}\}\cup\nabla&\Longrightarrow&
         \{a\fra \susp{t}{\rho}\}\cup\nabla\\
      \{a\fra [b]t\}\cup\nabla&\Longrightarrow&
         \left\{
         \begin{array}{ll}
           \{a\fra t\}\cup\nabla&\quad\textup{if}\ a\neq b\\
           \nabla&\quad\textup{otherwise}
         \end{array}\right.\\
      \{a\fra (t_1,\ldots,t_k)\}\cup\nabla&\Longrightarrow&
         \{a\fra t_i,\ldots,a\fra t_k\}\cup\nabla\\
      \{a\fra f(t)\}\cup\nabla&\Longrightarrow&\{a\fra t\}\cup\nabla. % chktex 36
    \end{array}
  \end{displaymath}

  The rules define a reduction relation on freshness environments. We write
  $\nabla\Longrightarrow\nabla'$ when $\nabla'$ is obtained from $\nabla$ by applying one
  simplification rule, and $\Longrightarrow^*$ for the reflexive and transitive closure of
  $\Longrightarrow$.
\end{defi}

\begin{lem}%
  \label{lem:unique-nf}
  The relation $\Longrightarrow$ is confluent and terminating.
\end{lem}

A freshness assertion is \emph{reduced} iff it is of the form $a\fra a$, $a \fra x$ or
$a \fra \susp{x}{\rho}$. We say that $a\fra a$ is \emph{inconsistent} and $a \fra x$ and
$a \fra \susp{x}{\rho}$ are \emph{consistent}. (Notice that assertions $a\fra x$ and
$a \fra \susp{x}{\iota}$ are syntactically different, although both represent the same
freshness relation.) An environment $\nabla$ is \emph{reduced} iff it consists only of
reduced assertions. An environment containing a freshness assertion that is not reduced
can always be simplified using one of the rules in
Definition~\ref{def:simplification}. Therefore, by Lemma~\ref{lem:unique-nf}, an
environment $\nabla$ reduces by $\Longrightarrow^*$ to a unique reduced environment, which
we call the \emph{normal form} of $\nabla$, written $\nf{\nabla}$. An environment $\nabla$
is \emph{inconsistent} iff $\nf{\nabla}$ contains some inconsistent assertion.

We write $\nf{\widetilde{\nabla}}$ for the environment obtained by replacing every
assertion $a\fra x$ in $\nf{\nabla}$ by the assertion $a\fra \susp{x}{\iota}$. Both
$\nf{\nabla}$ and $\nf{\widetilde{\nabla}}$ denote the same set of freshness
relations. Adding the identity renaming $\iota$ to variables that are not moderated
simplifies the definition of the entailment relation below.

\begin{lem}%
  \label{lem:entails}
  Let $\nabla$ be an environment over $\Sigma$ and let $\varphi$ be a ground
  substitution. Then $\varphi(\nabla)$ holds iff $\varphi(\nf{\nabla})$ holds. Moreover,
  $\varphi(\nf{\nabla})$ holds iff $\varphi(\nf{\widetilde{\nabla}})$ holds.
\end{lem}
The proof of Lemma~\ref{lem:entails} is in Appendix~\ref{ap:rule_formats}.

Notice that if $\nabla$ is inconsistent, then for every ground substitution $\varphi$ none
of $\varphi(\nabla)$, $\varphi(\nf{\nabla})$ and $\varphi(\nf{\widetilde{\nabla}})$ holds.

Our notion of entailment $\nabla\vdash\nabla'$ to be defined below represents that the
freshness relations in $\varphi(\nabla)$ imply the freshness relations in
$\varphi(\nabla')$. In the presence of assertions of the shape $a\fra \susp{x}{\rho}$,
checking that one environment entails another requires some care. Take the entailment
$\{a\fra \susp{x}{\rep{a}{b}}\}\vdash\{b\fra\susp{x}{\rep{b}{a}}\}$. We have
\begin{displaymath}
  \begin{array}{l}
    \peract{\tr{a}{b}}{\NT{\varphi(\susp{x}{\rep{a}{b}})}}
    = \peract{\tr{a}{b}}{\NT{\susp{\varphi(x)}{\rep{a}{b}}}}
    = \peract{\tr{a}{b}}{\NT{\ren{\varphi(x)}{\rep{a}{b}}}}\\
    = \NT{\peract{\tr{a}{b}}{(\ren{\varphi(x)}{\rep{a}{b}})}}
    = \NT{\ren{(\peract{\tr{a}{b}}{\varphi(x)})}{\peract{\tr{a}{b}}{\rep{a}{b}}}}\\
    = \NT{\ren{(\peract{\tr{a}{b}}{\varphi(x)})}{\tr{a}{b};\rep{a}{b};\tr{a}{b}}}
    = \NT{\ren{\varphi(x)}{\tr{a}{b};\tr{a}{b};\rep{a}{b};\tr{a}{b}}}\\
    = \NT{\ren{\varphi(x)}{\rep{a}{b};\tr{a}{b}}}
    = \NT{\ren{\varphi(x)}{\rep{b}{a}}}
    = \NT{\susp{\varphi(x)}{\rep{b}{a}}}
    = \NT{\varphi(\susp{x}{\rep{b}{a}})},
  \end{array}
\end{displaymath}
for every ground substitution $\varphi$. By equivariance of $\#$,
$a\#\NT{\varphi(\susp{x}{\rep{a}{b}})}$ holds iff $b\#\NT{\varphi(\susp{x}{\rep{b}{a}})}$
holds. The permutation $\tr{a}{b}$ mediates between the atoms $a$ and $b$ and between the
renamings $\rep{a}{b}$ and $\rep{b}{a}$. Definition~\ref{def:entails} below considers such
a mediating permutation.

\begin{defi}\label{def:entails}
  We say $\nabla$ \emph{entails} $\nabla'$ (written $\nabla \vdash \nabla'$) iff either
  $\nabla$ is inconsistent, or otherwise for every assertion $a_1\fra \susp{x}{\rho_1}$ in
  $\nf{\widetilde{\nabla'}}$ there exists a permutation $\pi$ and a freshness assertion
  $a_2\fra \susp{x}{\rho_2}$ in $\nf{\widetilde{\nabla}}$ such that $\pi\ a_1 = a_2$ and
  $\rho_1;\pi = \rho_2$.
\end{defi}

\begin{lem}%
  \label{lem:entails2}
  Let $\nabla$ and $\nabla'$ be environments over $\Sigma$ such that
  $\nabla\vdash \nabla'$. Then, for every ground substitution $\varphi$, if
  $\varphi(\nabla)$ holds then $\varphi(\nabla')$ holds.
\end{lem}
\begin{cor}
  In particular, if $\emptyset \vdash \nabla$ then $\varphi(\nabla)$ holds for every
  ground substitution $\varphi$.
\end{cor}
The proof of Lemma~\ref{lem:entails2} is in Appendix~\ref{ap:rule_formats}.

We are interested in NTSs~\cite{PBEGW15}, which consider signatures with base sorts $\Act$
(for actions) and $\Proc$ (for processes), with a single atom sort $\Chan$ and with source
and residual sorts $\Proc$ and $\Act\times\Proc$ respectively. We let
$\SigmaNTS$\label{pag:sigma-nts} be any such signature parametric on a set $F$ of
function symbols that we keep implicit.\label{pag:sigma_nts} We let
$\bn:\nTerm{\Sigma}_\Act\to\pset_\omega(\Atom_\Chan)$ be the binding-names function of a
given NTS\@. From now on we restrict our attention to the NTS of~\cite{PBEGW15} (without
predicates), and the definitions and results to come apply to NRTS/NRTSS over a signature
$\SigmaNTS$. We require that the rules of an NRTSS only contain ground actions $\ell$ and
therefore function $\bn$ is always defined over $\NT{\ell}$. (Recall that we write
$\bn(\ell)$ instead of $\bn(\NT{\ell})$ since it is clear in this context that the $\ell$
stands for a nominal term.) The rule format that we introduce in
Definition~\ref{def:alpha-conv-format} relies on identifying the rules that give rise to
transitions with actions $\ell$ such that $\bn(\ell)$ is non-empty, which are the
transitions that meet the conditions of the property of alpha-conversion of residuals. To
this end, we adapt the notion of strict stratification from~\cite{FV03,AFGI17}.

\begin{defi}[Partial strict stratification]%
  \label{def:partial-strict stratification}
  Let $\mathcal{R}$ be an NRTSS over a signature $\SigmaNTS$ and $\bn$ be a binding-names
  function. Let $S$ be a partial map from pairs of ground processes and actions to ordinal
  numbers. $S$ is a \emph{partial strict stratification of $\mathcal{R}$ with respect to
    $\bn$} iff
  \begin{enumerate}[label={(\roman*)},ref={(\roman*)}]
  \item $S(\varphi(t), \ell) \not= \bot$, for every rule in $\R$ with conclusion
    $t \rel{} (\ell, t')$ such that $\bn(\ell)$ is non-empty and for every ground
    substitution $\varphi$, and
  \item $S(\varphi(u_i),\ell_i)<S(\varphi(t),\ell)$ and $S(\varphi(u_i),\ell_i)\not=\bot$,
    for every rule $\Ru$ in $\R$ with conclusion $t\rel{}(\ell,t')$ such that
    $S(\varphi(t),\ell)\not=\bot$, for every premiss $u_i\rel{}(\ell_i,u'_i)$ of $\Ru$ and
    for every ground substitution $\varphi$.
  \end{enumerate}
  We say a pair $(p,\ell)$ of ground process and action \emph{has order} $S(p,\ell)$.
\end{defi}
The choice of $S$ determines which rules will be considered by the rule format for NRTSSs
of Definition~\ref{def:alpha-conv-format} below, which guarantees that the induced
transition relation satisfies alpha-conversion of residuals and, therefore, the associated
transition relation together with function $\bn$ are indeed an NTS\@.
We will intend the map $S$ to be such that the only rules whose source and label of the
conclusion have defined order are those that may take part in proof trees of transitions
with some binding atom in the action.

\begin{defi}[Alpha-conversion-of-residuals format]%
  \label{def:alpha-conv-format}
  Let $\mathcal{R}$ be an NRTSS over a signature $\SigmaNTS$, $\bn$ be a binding-names
  function and $S$ be a partial strict stratification of $\mathcal{R}$ with respect to
  $\bn$. Assume that all the actions occurring in the rules of $\mathcal{R}$ are
  ground. Let
  \begin{mathpar}
    \inferrule*[right=Ru] {\{u_i\rel{}(\ell_i,u'_i)\mid i\in I\} \qquad\nabla}
    {t\rel{}(\ell,t')}
  \end{mathpar}
  be a rule in $\mathcal{R}$. Let $D$ be the set of variables that occur in the source $t$
  of \textsc{Ru} but do not occur in the premisses $u_i\rel{}(\ell_i,u'_i)$ with $i\in I$,
  the environment $\nabla$ or the target $t'$ of the rule. The rule \textsc{Ru} is in
  \emph{alpha-conversion-of-residuals format with respect to $S$} (\emph{ACR format with
    respect to $S$} for short) iff for each ground substitution $\varphi$ such that
  $S(\varphi(t),\ell)\neq\bot$, there exists a ground substitution $\gamma$ such that
  $\dom(\gamma)\subseteq D$, and for every atom $a$ in the set
  $\{c\mid \nf{\{c\fra t\}}\not=\emptyset\}$ and for every atom $b\in\bn(\ell)$, the
  following hold:
  \begin{enumerate}[label={(\roman*)},ref={(\roman*)}]
  \item\label{it:one} $\{a\fra t'\} \cup \nabla \vdash \{a\fra u'_i \mid i\in I\}$,
  \item\label{it:two}
    $\{a\fra t'\} \cup \nabla \cup \{a\fra u_i \mid i\in I\} \vdash \{a\fra \gamma(t)\}$,
    and
  \item\label{it:three}
    $\nabla\cup\{b\fra u_i\mid i\in I \land b\in \bn(\ell_i)\} \vdash \{b\fra\gamma(t)\}$.
  \end{enumerate}
  An NRTSS $\mathcal{R}$, together with a binding-names function $\bn$, is in \emph{ACR
    format with respect to a partial strict stratification $S$} iff $\mathcal{R}$ is in
  equivariant format and all the rules in $\mathcal{R}$ are in ACR format with respect to
  $S$.
\end{defi}

Given a transition $p\rel{}(\ell,q)$ that unifies with the conclusion of \Ru, the rule
format ensures that any atom $a$ that is fresh in $(\ell, q)$ is also fresh in $p$, and
also that the binding atom $b$ is fresh in $p$.\label{pag:fresh-in-source} We have
obtained the constraints of the rule format by considering the variable flow in each node
of a proof tree and the freshness relations that we want to ensure. Constraints~(i) and
(ii) cover the case for the freshness relation $a\#p$ and Constraint~(iii) covers the case
for the freshness relation $b\#p$. The purpose of substitution $\gamma$ is to ignore the
variables that occur in the source of a rule but are dropped everywhere else in the rule.
Constraints~(i) and (ii) are not required for atoms $a$ that for sure are fresh in $p$,
and this explains why the $a$ in the rule format ranges over
$\{c\mid \nf{\{c\fra t\}}\not=\emptyset\}$. For example, take the instance of rule
\textsc{Res} in Figure~\ref{fig:early-pi-NTS} from Section~\ref{sec:early-pi-calculus}
with $\ell=\boutAA[a]{b}$. Condition (i)
\begin{displaymath}
 \{c\fra (\boutAA[a]{b},\newPA[{[c]y}]),
 c\fra \boutAA[a]{b}\} \vdash\{c\fra (\boutAA[a]{b},y)\}
\end{displaymath}
does not hold because $c\fra [c]y$ does not entail that $c\fra y$. However, $c$ is fresh
in $\NT{\newPA[{[c]p}]}$ even if it is not fresh in $\NT{p}$.
\begin{thm}[Rule format for NTSs]%
  \label{the:alpha-conversion}
  Let $\mathcal{R}$ be an NRTSS over a signature $\SigmaNTS$, $\bn$ be a binding-names
  function and $S$ be a partial strict stratification of $\mathcal{R}$ with respect to
  $\bn$.  If $\mathcal{R}$ is in ACR format with respect to $S$ then the NRTS induced by
  $\mathcal{R}$ and $\bn$ constitute an NTS---that is, the transition relation induced by
  $\mathcal{R}$ is equivariant and satisfies alpha-conversion of residuals.
\end{thm}

\begin{proof}[Sketch of the proof]
  Given a transition $\NT{\varphi(t)}\rel{}\NT{\varphi(\ell,t')}$, we first prove the
  freshness relations $a\#\NT{\varphi(\gamma(t))}$ and $b\#\NT{\varphi(\gamma(t))}$, for
  each $a \in \Atom\setminus \{c\in\supp(t)\mid {\nf{\{c\fra t\}}}=\emptyset\}$ and for
  every atom $b\in\bn(\ell)$. Both relations are proven by induction on
  $S(\varphi(\gamma(t)),\ell)$, and by analysing the variable flow in the rule unifying
  with $\varphi(t)\rel{}\varphi(\ell,t')$. For the first relation, we assume
  $a\#\NT{\varphi(t')}$, use Constraint~(i) to prove that $a\#\NT{\varphi(u'_i)}$ for each
  target $u'_i$ of a premiss, apply the induction hypothesis to obtain
  $a\#\NT{\varphi(\gamma(u_i))}$ for each source of a premiss $u_i$, and use
  Constraint~(ii) to conclude that $a\#\NT{\varphi(\gamma(t))}$. For the second relation,
  the induction hypothesis ensures that $b\#\NT{\varphi(\gamma(u_i))}$ for each source
  $u_i$ of a premiss having $b$ as a binding name, and we use Constraint~(iii) to conclude
  that $b\#\NT{\varphi(\gamma(t))}$. From these two freshness relations it is
  straightforward to prove that
  $\NT{\varphi(t)}\rel{}\peract{\tr{a}{b}}{\NT{\varphi((\ell,t'))}}$ and we are done.
\end{proof}

The full proof of Theorem~\ref{the:alpha-conversion} is in Appendix~\ref{ap:rule_formats}.

%%%%%%%%%%%%%%%%%%%%%%%%%%%%%%%%%%%%%%%%%%%%%%%%%%%%%%%%%%%%%%%%%%%%%%%%%%%
%% Example of application
\section{Example of Application of the ACR-Format to the \texorpdfstring{$\pi$}{pi}-Calculus}%
\label{sec:example-nts}

In this section we consider two different semantics of the $\pi$-calculus. These semantics
differ in the moment at which substitution is performed at input processes. In the
\emph{early} semantics, substitution is performed whenever a process makes an input
transition. To wit, an input process $\inPA[a]{[c]p}$ can perform a transition to a
process $\ren{p}{\rep{b}{c}}$ that is obtained from $p$ by renaming the channel name $c$
with a channel name $b$ received through channel $a$.

\sloppy
In the \emph{late} semantics, substitution is postponed to the moment when an input
process and an output process synchronise. For instance, a parallel composition
$\parPA[{\inPA[a]{[c]p}}]{\outPA[a]{b}{q}}$ can perform a transition to
$\parPA[{\ren{p}{\rep{b}{c}}}]{q}$ whose left component is obtained from $p$ by renaming
the channel name $c$ with a channel name $b$ received through channel $a$.

\fussy
\subsection{Early Semantics of the \texorpdfstring{$\pi$}{pi}-Calculus}%
\label{sec:early-pi-calculus}
Consider the NRTSS $\RE$ in Figure~\ref{fig:early-pi-NTS} for the early
semantics of the $\pi$-calculus~\cite{MPW92} over the residual signature $\SigmaNTS$ as
defined on page~\pageref{pag:sigma-nts} of Section~\ref{sec-rule-format-nrts}, where $F$
is the set of function symbols from Example~\ref{ex:pi-calculus}. Omitted rules
\textsc{EParR}, \textsc{EParResR}, \textsc{ECommR}, \textsc{ECloseR} and \textsc{SumR} are
respectively the symmetric version of rules \textsc{EParL}, \textsc{EParResL},
\textsc{ECommL}, \textsc{ECloseL} and \textsc{SumL}.

\begin{figure}[ht]
  \begin{mathpar}
    \inferrule*[right=EIn]
    { }
    {\inPA[a]{[b]x}\rel{}\resAF{\inAA[a]{c}}{\susp{x}{\rep{b}{c}}}}
    \\
    \inferrule*[right=Out]
    { }
    {\outPA[a]{b}{x}\rel{}\resAF{\outAA[a]{b}}{x}}
    \and
    \inferrule*[right=Tau]
    { }
    {\tauPA[x]\rel{}\resAF{\tau}{x}}
    \and
    % \inferrule*[right=EParL,
    %             left=$\ell{{}\not={}}\mathit{boutA}(a{,}\,b)$]
    % {x_1\rel{}\resAF{\ell}{y_1}}
    % {\parPA[x_1]{x_2} \rel{}\resAF{\ell}{(\parPA[y_1]{x_2})}}
    \inferrule*[right=EParL,
                left=$\ell\not\in\{\mathit{boutA}(a{,}\,b){,} \mid a{,}b\in\Atom_{\Chan}\}$] % chktex 36 chktex 1
    {x_1\rel{}\resAF{\ell}{y_1}}
    {\parPA[x_1]{x_2} \rel{}\resAF{\ell}{(\parPA[y_1]{x_2})}}
    \\
    \inferrule*[right=EParResL]
    {x_1 \rel{} (\boutAA[a]{b},y_1)\qquad b\fra x_2}
    {\parPA[x_1]{x_2} \rel{} (\boutAA[a]{b},(\parPA[y_1]{x_2}))}
    \\
    \inferrule*[right=ECommL]
    {x_1\rel{}\resAF{\outAA[a]{b}}{y_1}
      \qquad x_2\rel{}\resAF{\inAA[a]{b}}{y_2}}
    {\parPA[x_1]{x_2}\rel{}\resAF{\tauAA}{(\parPA[y_1]{y_2})}}
    \\
    \inferrule*[right=ECloseL]
    {x_1\rel{}\resAF{\boutAA[a]{b}}{y_1}
      \qquad x_2\rel{}\resAF{\inAA[a]{b}}{y_2} \qquad b\fra x_2}
    {\parPA[x_1]{x_2}\rel{}\resAF{\tauAA}{\newPA[{[b](\parPA[y_1]{y_2})}]}}
    \\
    \inferrule*[right=SumL]
    {x_1\rel{}\resAF{\ell}{y_1}}
    {\sumPA[x_1]{x_2} \rel{}\resAF{\ell}{y_1}}
    \and
    \inferrule*[right=Rep]
    {x\rel{}\resAF{\ell}{y}}
    {\repPA[x] \rel{}\resAF{\ell}{(\parPA[y]{\repPA[x]})}}
    \\
    \inferrule*[right=ERepComm]
    {x\rel{}\resAF{\outAA[a]{b}}{y_1}\qquad x\rel{}\resAF{\inAA[a]{b}}{y_2}}
    {\repPA[x]\rel{}\resAF{\tauAA}{\parPA[{\parPA[y_1]{y_2}}]{\repPA[x]}}}
    \\
    \inferrule*[right=ERepClose]
    {x\rel{}\resAF{\boutAA[a]{b}}{y_1}
      \qquad x\rel{}\resAF{\inAA[a]{b}}{y_2}\qquad b\fra x}
    {\repPA[x]\rel{}\resAF{\tauAA}{{\parPA[{\newPA[{[b](\parPA[y_1]{y_2})}]}]{\repPA[x]}}}}
    \\
    \inferrule*[right=Open]
    {x\rel{}\resAF{\outAA[a]{b}}{y}\qquad b\fra a}
    {\newPA[{[b]x}]\rel{}\resAF{\boutAA[a]{b}}{y}}
    \and
    \inferrule*[right=Res]
    {x\rel{}\resAF{\ell}{y}\qquad b\fra \ell}
    {\newPA[{[b]x}] \rel{}\resAF{\ell}{\newPA[{[b]y}]}}
  \end{mathpar}
  \begin{center}
    where $a,b,c\in \Atom_\Chan$ and $\ell$ is a ground action.
  \end{center}
  \caption{NRTSS $\RE$ for the early $\pi$-calculus.}%
  \label{fig:early-pi-NTS}
\end{figure}

In the rule \textsc{EIn}, the moderated term $\susp{x}{\rep{b}{c}}$ is used in order to
indicate that the renaming $\rep{b}{c}$ will be performed over the term substituted for
variable $x$.

The rule \textsc{ECloseL} specifies the interaction of a process like
$\NT{\newPA[{[b](\outPA[a]{b}{p})}]}$, which exports a private channel name $b$ through
channel $a$, composed in parallel with an input process such as $\NT{\inPA[a]{[c]q}}$ that
reads through channel $a$. The private name $b$ is exported and the resulting process
$\NT{\newPA[{[b](\parPA[p]{\peract{\tr{c}{b}}{q}})}]}$ is the parallel composition of
processes $p$ and $q$ where atom $b$ is restricted. For illustration, consider the raw
terms $t\equiv \newPA[{[b](\outPA[a]{b}{p})}]$ and $t'\equiv(\boutAA[a]{b},p)$. The
transition $\NT{t}\rel{}\NT{t'}$ is provable in $\RE$ by the following proof tree:
\begin{mathpar}
  \inferrule*[right=Open{\normalsize .}]
  {
    \inferrule*[right=Out]
    { }
    {\NT{\outPA[a]{b}{p}}\rel{}\NT{(\outAA[a]{b},p)}}\\
    b\#a
  }
  {\NT{\newPA[{[b](\outPA[a]{b}{p})}]}\rel{}\NT{(\boutAA[a]{b},p)}}
\end{mathpar}

Notice that the nodes of the proof tree above are labelled by transitions involving
nominal terms. Therefore, if we were to start with the raw term
$q\equiv\newPA[{[c](\outPA[a]{c}{p})}]$ where $c\# (a,p)$---which is alpha-equivalent to
$t$---then the transition $\NT{q}\rel{}\NT{t'}$ would have the same proof tree as above,
since $\NT{t}$ and $\NT{q}$ are the same nominal term.

We use the rule format of Definition~\ref{def:alpha-conv-format} to show that $\RE$,
together with the equivariant function $\bnE$ such that $\bnE(\boutAA[a]{b})=\{b\}$, and
$\bnE(\ell)=\emptyset$ otherwise, specifies an NTS\@. We consider the following partial
strict stratification
\begin{displaymath}
\begin{array}{rcl}
S(\outPA[a]{b}{p},\outAA[a]{b})&=&0\\
S(\parPA[p]{q},\ell)&=&1+\max\{S(p,\ell),S(q,\ell)\}\\
S(\sumPA[p]{q},\ell)&=&1+\max\{S(p,\ell),S(q,\ell)\}\\
S(\repPA[p],\ell)&=&1+S(p,\ell)\\
S(\newPA[{[c]p}],\ell)&=&1+S(p,\ell) \quad \textup{if}~c \#\ell\\
S(\newPA[{[b]p}],\boutAA[a]{b})&=&1+S(p,\outAA[a]{b})\\
S(p,\ell')&=&\bot\quad \text{otherwise}
\end{array}
\end{displaymath}
where $a,b\in\Atom_{\Chan}$ and
$\ell\in\{\boutAA[a]{b},\outAA[a]{b}\mid a,b\in\Atom_{\Chan}\}$. Operators $\max$ and $+$
above are extended with $\bot$ in the following way:
\begin{displaymath}
  \begin{array}{rcl}
    \max(\{s_1,\ldots,s_n\} \cup \{\bot\}) &=& \max\{s_1,\ldots,s_n\}\\
    \max\{\bot\} &=& \bot\\
    \bot + s &=& \bot\\
    s + \bot &=& \bot.
  \end{array}
\end{displaymath}

We check that $\RE$, together with the binding-names function $\bnE$, is in ACR format
with respect to $S$ as follows. First of all, notice that, from the definition of $S$, we
have that $S(p,\tauAA)=S(p,\inAA[a]{b})=\bot$, for each $p$ and $a,b \in\Atom_{\Chan}$.
Observe that $S$ meets Definition~\ref{def:partial-strict stratification}(i) because a % chktex 36
formula with either action $\tauAA$ or $\inAA[a]{b}$ does not take part in any proof tree
that proves a transition whose action has binding names. Therefore, the only rules in
$\RE$ whose sources and actions unify with pairs of processes and actions that have
defined order are \textsc{Out}, \textsc{Open} and \textsc{EParResL}, the instance of rule
\textsc{EParL} where $\ell=\outAA[a]{b}$, and the instances of rules \textsc{SumL},
\textsc{Rep} and \textsc{Res} where $\ell\in\{\boutAA[a]{b},\outAA[a]{b}\}$ (and the
corresponding instances of the symmetric versions \textsc{EParResR}, \textsc{EParR} and
\textsc{SumR}, which are omitted in the excerpt and will not be checked). Observe that $S$
meets Definition~\ref{def:partial-strict stratification}(ii) because for each rule whose % chktex 36
conclusion has either action $\boutAA[a]{b}$ or $\outAA[a]{b}$, the order of the ground
transition that unifies with its conclusion is always bigger than the order of the ground
transitions that unify with its premisses.

For rule \textsc{Out}, we have an empty set of premisses and the set $D$ of variables that
are in $\supp(\outPA[a]{b}{x})$ but are not in $\supp(\outAA[a]{b},x)$ is empty. Therefore
we can do away with substitution~$\gamma$. Every atom $c$ is such that
$\nf{\{c\fra \outPA[a]{b}{x}\}}\not=\emptyset$, and the set $\bnE(\outAA[a]{b})$ is
empty. We only need to check that for every atom $c$, the obligation
$\{c\fra (\outAA[a]{b},x)\}\vdash\{c\fra \outPA[a]{b}{x}\}$ holds. For atoms
$c\in\supp(\outAA[a]{b},x)$ this obligation vacuously holds, and therefore it suffices to
pick an atom $c$ fresh in the rule and check that
$\{c\fra (\outAA[a]{b},x)\}\vdash\{c\fra \outPA[a]{b}{x}\}$, which simplifies to
$\{c\fra\susp{x}{\iota}\} \vdash\{c\fra\susp{x}{\iota}\}$. The permutation $\iota$
witnesses that this entailment trivially holds as in Definition~\ref{def:entails}(i). % chktex 36

For rule \textsc{Open} the set $D$ is empty and every atom $c\#b$ is such that
$\nf{\{c\fra \newPA[{[b]x}]\}}\not=\emptyset$. It suffices to pick atom $c$ fresh in the
rule (and therefore different from $b$) and check that
\begin{displaymath}
\begin{array}{c}
\{c\fra (\boutAA[a]{b},y),b\fra a\}
\vdash\{c\fra (\boutAA[a]{b},y)\}\qquad\text{and}\\
\{c\fra (\boutAA[a]{b},y), b\fra a, c\fra x\}
\vdash\{c\fra \newPA[{[b]x}]\}
\qquad\text{and}\\
\{b\fra x, b \fra a\}\vdash\{b\fra \newPA[{[b]x}]\},
\end{array}
\end{displaymath}
which holds because $b\fra \newPA[{[b]x}]$ reduces to the empty set.

For rule \textsc{EParResL} we have premiss $x_1\rel{}(\boutAA[a]{b},y_1)$ and the set $D$
is empty. Every atom $c$ is such that $\nf{\{c\fra \parPA[x_1]{x_2}\}}\not=\emptyset$ and
the set $\bnE(\boutAA[a]{b})$ contains atom $b$. We check that
\begin{displaymath}
\begin{array}{c}
\{c\fra (\boutAA[a]{b},\parPA[y_1]{x_2}),b\fra x_2\}
\vdash\{c\fra (\boutAA[a]{b}, y_1)\}\qquad\text{and}\\
\{c\fra (\boutAA[a]{b},\parPA[y_1]{x_2}),b\fra x_2,
c\fra x_1\}\vdash\{c\fra \parPA[x_1]{x_2}\}\qquad\text{and}\\
\{b\fra x_1,b\fra x_2\}\vdash\{b\fra \parPA[x_1]{x_2}\}.
\end{array}
\end{displaymath}
Atom $c$ is either fresh in the rule, or otherwise $c=a$ or $c=b$. In all three cases,
checking the obligations above is straightforward.

Consider the instance of rule \textsc{EParL} where $\ell=\outAA[a]{b}$. That rule instance
has premiss $x_1\rel{}(\outAA[a]{b},y_1)$ and the set $D$ is empty. Every atom $c$ is such
that $\nf{\{c\fra \parPA[x_1]{x_2}\}}\not=\emptyset$ and the set $\bnE(\outAA[a]{b})$ is
empty. We consider the three cases over $c$ as before and check that
\begin{displaymath}
\begin{array}{c}
\{c\fra (\outAA[a]{b},\parPA[y_1]{x_2})\}
\vdash\{c\fra (\outAA[a]{b}, y_1)\}\qquad\text{and}\\
\{c\fra (\outAA[a]{b},\parPA[y_1]{x_2}),
c\fra x_1\}\vdash\{c\fra \parPA[x_1]{x_2}\},
\end{array}
\end{displaymath}
which is straightforward.

Consider now the instance of rule \textsc{SumL} where $\ell=\boutAA[a]{b}$. We have
premiss $x_1\rel{}(\boutAA[a]{b},y_1)$ and the set $D$ contains only $x_2$. We pick
$\gamma$ such that $\gamma(x_2)=\nullPA$. Every atom $c$ is such that
$\nf{\{c\fra \sumPA[x_1]{x_2}\}}\not=\emptyset$ and the set $\bnE(\boutAA[a]{b})$ contains
only atom $b$. Again, we check that
\begin{displaymath}
\begin{array}{c}
\{c\fra (\boutAA[a]{b},y_1)\}
\vdash\{c\fra (\boutAA[a]{b},y_1)\}\qquad\text{and}\\
\{c\fra (\boutAA[a]{b},y_1), c\fra x_1\}\vdash\{c\fra \gamma(\sumPA[x_1]{x_2})\}
\qquad\text{and}\\
\{b\fra x_1\}\vdash\{b\fra \gamma(\sumPA[x_1]{x_2})\},
\end{array}
\end{displaymath}
which holds since $\gamma(\sumPA[x_1]{x_2})=\sumPA[x_1]{\nullPA}$ and $b\fra \nullPA$
reduces to the empty set.

The instance of rule \textsc{SumL}, where $\ell=\outAA[a]{b}$, has premiss
$x_1\rel{}(\outAA[a]{b},y_1)$, and the set $D$ and the substitution $\gamma$ are the same
as for the previous instance of \textsc{SumL}. Every atom $c$ is such that
$\nf{\{c\fra \sumPA[x_1]{x_2}\}}\not=\emptyset$ and the set $\bnE(\outAA[a]{b})$ is
empty. We check that
\begin{displaymath}
  \begin{array}{c}
    \{c\fra (\outAA[a]{b},y_1)\} \vdash\{c\fra (\outAA[a]{b},y_1)\}\qquad\text{and}\\
    \{c\fra (\outAA[a]{b},y_1), c\fra x_1\}\vdash\{c\fra \gamma(\sumPA[x_1]{x_2})\},
  \end{array}
\end{displaymath}
which hold as before.

For the instance of rule \textsc{Rep}, where $\ell=\boutAA[a]{b}$, the set $D$ is empty
and every atom $c$ is such that $\nf{\{c\fra \repPA[x]\}}\not=\emptyset$. We need to check
that
\begin{displaymath}
\begin{array}{c}
\{c\fra (\boutAA[a]{b},\parPA[y]{\repPA[x]})\}
\vdash\{c\fra (\boutAA[a]{b},y)\}\qquad\text{and}\\
\{c\fra (\boutAA[a]{b},\parPA[y]{\repPA[x]}), c\fra x\}
\vdash\{c\fra \repPA[x]\}
\qquad\text{and}\\
\{b\fra x\}\vdash\{b\fra \repPA[x]\},
\end{array}
\end{displaymath}
which is straightforward.

For the instance of rule \textsc{Rep}, where $\ell=\outAA[a]{b}$, the set $D$ is empty and
every atom $c$ is such that $\nf{\{c\fra \repPA[x]\}}\not=\emptyset$. It suffices to check
that
\begin{displaymath}
\begin{array}{c}
\{c\fra (\outAA[a]{b},\parPA[y]{\repPA[x]})\}
\vdash\{c\fra (\outAA[a]{b},y)\}\qquad\text{and}\\
\{c\fra (\outAA[a]{b},\parPA[y]{\repPA[x]}), c\fra x\}
\vdash\{c\fra \repPA[x]\},
\end{array}
\end{displaymath}
which is straightforward.

For the instance of the rule \textsc{Res}, where $\ell=\boutAA[a]{b}$, the set $D$ is
empty and every atom $d\#c$ is such that $\nf{\{d\fra \newPA[{[c]x}]\}}\not=\emptyset$. We
check that
\begin{displaymath}
\begin{array}{c}
\{d\fra (\boutAA[a]{b},\newPA[{[c]y}]), c\fra \boutAA[a]{b}\}
\vdash\{d\fra (\boutAA[a]{b},y)\}\qquad\text{and}\\
\{d\fra (\boutAA[a]{b},\newPA[{[c]y}]), c\fra \boutAA[a]{b}, d\fra x\}
\vdash\{d\fra \newPA[{[c]x}]\}
\qquad\text{and}\\
\{b\fra x, c\fra \boutAA[a]{b}\}\vdash\{b\fra \newPA[{[c]x}]\}.
\end{array}
\end{displaymath}
Atom $d$ is either fresh in the rule, or otherwise $d=a$ or $d=b$. In all three cases,
checking the obligations above is straightforward. For instance, in the second and third
obligations, $d\fra x$ and $b\fra x$ entail $d\fra \newPA[{[c]x}]$ and
$b\fra \newPA[{[c]x}]$ respectively.

For the instance of the rule \textsc{Res} where $\ell=\outAA[a]{b}$ the set $D$ is empty
and every atom $d\# c$ is such that $\nf{\{d\fra \newPA[{[c]x}]\}}\not=\emptyset$. We
consider the three cases over $d$ as before and check that
\begin{displaymath}
\begin{array}{c}
  \{d\fra (\outAA[a]{b},\newPA[{[c]y}]), c\fra \outAA[a]{b}\}
  \vdash\{d\fra (\outAA[a]{b},y)\}\qquad\text{and}\\
  \{d\fra (\outAA[a]{b},\newPA[{[c]y}]), c\fra \outAA[a]{b}, d\fra x\}
  \vdash\{d\fra \newPA[{[c]x}]\},
\end{array}
\end{displaymath}
which holds because $d\fra x$ entails $d\fra \newPA[{[c]x}]$.

Atoms $a$, $b$ and $c$ in $\RE$ range over $\Atom_\Chan$, and thus $\RE$ is in equivariant
format. Therefore $\RE$ is in ACR format with respect to $S$. By
Theorem~\ref{the:alpha-conversion} the NRTS induced by $\RE$, together with function
$\bnE$, constitute an NTS of Definition~\ref{def:nts}.

\subsection{Late Semantics of the \texorpdfstring{$\pi$}{pi}-Calculus}%
\label{sec:late-pi-calculus}
The NRTSS $\RL$ over the residual signature $\SigmaNTS$ models the late semantics of the
$\pi$-calculus~\cite{MPW92} in our target semantic model, which is an NTS\@. $\RL$ consists
of the rules in Figure~\ref{fig:late-pi-NTS} together with rules
\textsc{Out},\textsc{Tau}, \textsc{SumL}, \textsc{Rep}, \textsc{Open} and \textsc{Res}
from Figure~\ref{fig:early-pi-NTS} in Section~\ref{sec:early-pi-calculus}, and the omitted
symmetric versions \textsc{LParR}, \textsc{LParResR}, \textsc{LCommR}, \textsc{LCloseR}
and \textsc{SumR}.

$\RL$ is an NRTSS over signature $\SigmaNTS$, where the free-input actions are replaced by
\emph{bound-input actions} (page~159 of~\cite{SW01}), which we write $\binAA[a]{b}$. We
let the binding-names function $\bnL$ be such that the binding name of both the
bound-output action $\boutAA[a]{b}$ and the bound-input action $\binAA[a]{b}$ be $b$, that
is, $\bnL(\boutAA[a]{b}) = \bnL(\binAA[a]{b}) = \{b\}$ and $\bnL(\ell) = \emptyset$
otherwise.

\begin{figure}[ht]
  \begin{mathpar}
    \inferrule*[right=LIn]
    {b\fra a}
    {\inPA[a]{[b]x}\rel{}\resAF{\binAA[a]{b}}{x}}
    \\
    \inferrule*[right=LParL,
                left=$\ell\not\in\{\mathit{boutA}(a{,}\,b){,}\,\mathit{binA}(a{,}\,b) \mid a{,}b\in\Atom_{\Chan}\}$] % chktex 36 chktex 1
    {x_1\rel{}\resAF{\ell}{y_1}}
    {\parPA[x_1]{x_2} \rel{}\resAF{\ell}{(\parPA[y_1]{x_2})}}
    \\
    \inferrule*[right=LParResL,
                left=$\ell\in\{\mathit{boutA}(a{,}\,b){,}\,\mathit{binA}(a{,}\,b) \mid a{,}b\in\Atom_{\Chan}\}$] % chktex 36 chktex 1
    {x_1 \rel{} (\ell,y_1)\qquad b\fra x_2}
    {\parPA[x_1]{x_2} \rel{} (\ell,(\parPA[y_1]{x_2}))}
    \\
    \inferrule*[right=LCommL]
    {x_1\rel{}\resAF{\outAA[a]{b}}{y_1}
      \qquad x_2\rel{}\resAF{\binAA[a]{c}}{y_2}}
    {\parPA[x_1]{x_2}\rel{}\resAF{\tauAA}{(\parPA[y_1]{\susp{y_2}{\rep{b}{c}}})}}
    \\
    \inferrule*[right=LCloseL]
    {x_1\rel{}\resAF{\boutAA[a]{b}}{y_1}
      \qquad x_2\rel{}\resAF{\binAA[a]{b}}{y_2}}
    {\parPA[x_1]{x_2}\rel{}\resAF{\tauAA}{\newPA[{[b](\parPA[y_1]{y_2})}]}}
    \\
    \inferrule*[right=LRepComm]
    {x\rel{}\resAF{\outAA[a]{b}}{y_1}\qquad x\rel{}\resAF{\binAA[a]{c}}{y_2}}
    {\repPA[x]\rel{}\resAF{\tauAA}
      {\parPA[{\parPA[y_1]{\susp{y_2}{\rep{b}{c}}}}]{\repPA[x]}}}
    \\
    \inferrule*[right=LRepClose]
    {x\rel{}\resAF{\boutAA[a]{b}}{y_1}
      \qquad x\rel{}\resAF{\binAA[a]{b}}{y_2}}
    {\repPA[x]\rel{}\resAF{\tauAA}{{\parPA[{\newPA[{[b](\parPA[y_1]{y_2})}]}]{\repPA[x]}}}}
  \end{mathpar}
  \begin{center}
    where $a,b,c\in \Atom_\Chan$ and $\ell$ is a ground action.
  \end{center}
  \caption{NRTSS $\RL$ for the late $\pi$-calculus.}%
  \label{fig:late-pi-NTS}
\end{figure}

In rule \textsc{LIn}, the binding input action $\binAA[a]{b}$ binds atom $b$ in the term
substituted for variable $x$ on the right side of the residual. In rules \textsc{LCommL}
and \textsc{LRepComm}, the moderated term $\susp{y_2}{\rep{b}{c}}$ is used in order to
indicate that the renaming $\rep{b}{c}$ will be performed over the term substituted for
variable $y_2$.

\begin{rem}\label{rem:prevent-capture-in}
  In order to represent the binding input action of the late $\pi$-calculus in an NTS,
  rule \textsc{LIn} ensures that the binding atom $b$ is different from the communication
  channel $a$ by requiring $b\fra a$. This is similar to the requirement $b\fra a$ in rule
  \textsc{Open}. As a result, the obtained semantics minimally differs from the original
  one in~\cite{San96}. Consider the original late $\pi$-calculus and take the transitions
  $a(b).(\peract{\tr{a}{b}}{p})\lowrel{a(b)}\peract{\tr{a}{b}}{p}$ where $b$ is either $a$ or
  fresh in $p$. Our rule \textsc{LIn} prevents the transition
  $\NT{\inPA{a}{[a]p}}\rel{}\NT{(\binAA[a]{a},p)}$, and our semantics fails to faithfully
  represent the above-mentioned transition in the original late $\pi$-calculus when
  $b=a$. By alpha-conversion of residuals, if the state $\NT{\inPA[a]{[a]p}}$ has
  derivative $\NT{(\binAA[a]{a},p)}$, then the same state has to have all the derivatives
  $\{\NT{(\binAA[c]{c},\peract{\tr{a}{c}}{p})}\mid c\#(\binAA[a]{a},p)\}$, but these
  derivatives do not represent valid transitions in the original late $\pi$-calculus.

  However, the discrepancy between the original and our semantics has very limited
  consequences, since the binding name of an input process vanishes when communication is
  performed. Our semantics allows for the transition
\begin{equation}\label{eq:communication-binding-names}
  \begin{array}{l}
    \NT{\parPA[{\outPA[a]{a}{\nullPA}}]{\inPA[a]{[b](\outPA[c]{b}{\nullPA})}}} \rel{} \\
    \NT{\resAF{\tauAA}{\parPA[{\nullPA}]{\ren{(\outPA[c]{b}{\nullPA})}{\rep{a}{b}}}}}
    = \NT{\resAF{\tauAA}{\parPA[{\nullPA}]{\outPA[c]{a}{\nullPA}}}},
  \end{array}
\end{equation}
where the name $a$ is transmitted over the channel with the same name and where
$b\#(a,c)$. The transition in (\ref{eq:communication-binding-names}) faithfully represents
$(\overline{a}a.0\parallel a(b).\overline{c}b.0)\lowrel{\tau} (0\parallel\overline{c}a.0)$ in
the original late $\pi$-calculus.  By the nominal interpretations of terms, the process
$\NT{\inPA[a]{[b](\outPA[c]{b}{\nullPA})}}$ with binding atom $b$ is equal to the process
$\NT{\inPA[a]{[a](\outPA[c]{a}{\nullPA})}}$ with binding atom $a$, and thus the transition
in (\ref{eq:communication-binding-names}) also represents
$(\overline{a}a.0\parallel a(a).\overline{c}a.0)\lowrel{\tau} (0\parallel \overline{c}a.0)$
in the original late $\pi$-calculus.\blockqed
\end{rem}

As we did in Section~\ref{sec:early-pi-calculus}, we use the rule format of
Definition~\ref{def:alpha-conv-format} to show that $\RL$, together with equivariant
function $\bnL$ specifies an NTS\@. We consider the following partial strict stratification
\begin{displaymath}
\begin{array}{rcl}
S(\outPA[a]{b}{p},\outAA[a]{b})&=&0\\
S(\inPA[a]{[b]p},\binAA[a]{b})&=&0\\
S(\parPA[p]{q},\ell)&=&1+\max\{S(p,\ell),S(q,\ell)\}\\
S(\sumPA[p]{q},\ell)&=&1+\max\{S(p,\ell),S(q,\ell)\}\\
S(\repPA[p],\ell)&=&1+S(p,\ell)\\
S(\newPA[{[c]p}],\ell)&=&1+S(p,\ell) \quad \textup{if}~c \#\ell\\
S(\newPA[{[b]p}],\boutAA[a]{b})&=&1+S(p,\outAA[a]{b})\\
S(p,\ell')&=&\bot\quad \text{otherwise},
\end{array}
\end{displaymath}
where $\ell \in \{\boutAA[a]{b},\outAA[a]{b},\binAA[a]{b}\mid a{,}b\in\Atom_{\Chan}\}$.

Notice that the differences between the $S$ above and the partial strict stratification
from Section~\ref{sec:early-pi-calculus} are the inclusion of the second clause above,
which defines an order for the pair of input process and bound-input action, and the
addition of the bound-input action to the set over which the $\ell$ above ranges.

We check that $\RL$, together with the binding-names function $\bnL$, is in ACR format
with respect to $S$ as follows. First of all, the definition of $S$ yields that
$S(p,\tauAA)=\bot$, for each $p$. Observe that $S$ meets
Definition~\ref{def:partial-strict stratification}(i) because a formula with action % chktex 36
$\tauAA$ does not take part in any proof tree that proves a transition whose action has
binding names. Therefore, the only rules in $\RL$ whose sources and actions unify with
pairs of processes and actions that have defined order are \textsc{LIn}, \textsc{Out},
\textsc{Open} and \textsc{LParResL}, the instance of rule \textsc{LParL} where
$\ell=\outAA[a]{b}$, and the instances of rules \textsc{SumL}, \textsc{Rep} and
\textsc{Res} where $\ell\in\{\boutAA[a]{b},\outAA[a]{b},\binAA[a]{b}\}$ (and the
corresponding instances of the symmetric versions \textsc{LParResR}, \textsc{LParR} and
\textsc{SumR}, which are omitted in the excerpt and will not be checked). Observe that $S$
meets Definition~\ref{def:partial-strict stratification}(ii) because for each rule whose % chktex 36
conclusion has any of the actions $\boutAA[a]{b}$, $\outAA[a]{b}$ or $\binAA[a]{b}$, the
order of the transition that unifies with its conclusion is always bigger than the order
of the transitions that unify with its premisses.

We have already checked the ACR-format for rules \textsc{Out} and \textsc{Open} in
Section~\ref{sec:early-pi-calculus}. We have also checked the ACR-format for the instances
of rules \textsc{SumL}, \textsc{Rep} and \textsc{Res} where
$\ell\in\{\boutAA[a]{b},\outAA[a]{b}\}$, and we will not check the ACR-format for the
instances where $\ell=\binAA[a]{b}$ because the checking proceeds exactly as in the case
where $\ell=\boutAA[a]{b}$. We will limit ourselves to checking that rule \textsc{LIn} is
in the ACR-format with respect to $S$, as the checking for the other rules are similar to
those presented earlier.

For rule \textsc{LIn} we have an empty set of premisses, the set $D$ is empty, and every
atom $c\# b$ is such that ${\nf{\{b\fra \inPA[a]{[b]x}\}}\not=\emptyset}$. We check that
\begin{displaymath}
  \begin{array}{c}
    \{c\fra (\binAA[a]{b},x)\}\vdash\{c\fra \inPA[a]{[b]x}\}\qquad\textup{and}
    \qquad\{b\fra a\}\vdash\{b\fra \inPA[a]{[b]x}\}.
  \end{array}
\end{displaymath}
Let us consider the obligation on the left first. If $c=a$, that obligation vacuously
holds since its left-hand-side is inconsistent. If $c\not=a$, the obligation simplifies to
\begin{displaymath}
  \{c\fra\susp{x}{\iota}\}\vdash\{c\fra\susp{x}{\iota}\},
\end{displaymath}
which holds
straightforwardly. Checking the obligation on the right is also straightforward.

Atoms $a$, $b$ and $c$ in $\RL$ range over $\Atom_\Chan$, and thus $\RL$ is in equivariant
format. Therefore $\RL$ is in ACR format with respect to $S$. By
Theorem~\ref{the:alpha-conversion} the NRTS induced by $\RL$, together with function
$\bnL$, constitute an NTS of Definition~\ref{def:nts}.

\section{NTSs with Residuals of Abstraction Sort}%
\label{sec:nts-abstraction-sort}
In this section we explore alternative specifications of the NTSs à la Parrow in which we
allow for the use of residuals of abstraction sort. Intuitively, by the requirement of
alpha-conversion of residuals, the NTSs à la Parrow treat the actions with binding manes
as binding operators. In the systems with residuals of abstraction sorts, we let the
binding name in an action to be the binding atom of the residual in which the action
occurs. Our aim is to provide translations between the systems with and without residuals
of abstraction sort, and to give conditions under which the translations are inverse to
each other.

We have already defined the signature $\SigmaNTS$ on page~\pageref{pag:sigma-nts}, which
is parametric on a set $F$ of function symbols that we keep implicit. For the alternative
specifications with residuals of abstraction sort, we consider signatures with base and
residual sorts $\Proc$ and $[\Chan](\Act\times\Proc)$, and we let $\SigmaNTSAbs$ be any
such signature parametric on the set $F$ of function symbols.\label{pag:sigma-nts-abs}

We let $\Tr$ and $\TrAbs$ range over NRTSs over signatures $\SigmaNTS$ and $\SigmaNTSAbs$,
respectively, where we write $\trel$ and $\trelAbs$ for the transition relations of $\Tr$
and $\TrAbs$ respectively. We let $\bn$ range over equivariant functions that deliver the
binding names in an action. A tuple $(\Tr,\bn)$ where $\trel$ enjoys alpha-conversion of
residuals constitutes an NTS\@. In what follows, we assume that $|\bn(\ell)|\leq 1$ for
every action $\ell\in\Act$.\footnote{It is straightforward to generalise the results in
  this section to the case where $|\bn(\ell)|\leq n$ by iterating $n$ abstractions in the
  residuals, \ie, by adopting a sort $[\Chan_1]\ldots[\Chan_n](\Act\times\Proc)$ for the
  residuals and fixing function $\bn$ so that it returns an ordered list
  $(a_{\Chan_1},\ldots,a_{\Chan_n})$ of names instead of a set. We omit this
  generalisation here in order not to clutter notation.}

The translation from an NTS $(\Tr,\bn)$ to an NRTS $\TrAbs$ is given in the definition
below.

\begin{defi}%
  \label{def:nts-to-nrts}
  Let $(\Tr,\bn)$ be an NTS with equivariant transition relation $\trel$ and with
  equivariant function $\bn$ where $|\bn(\ell)|\leq 1$ such that $\trel$ enjoys
  alpha-conversion of residuals. %
  The NTS $(\Tr,\bn)$ \emph{translates to} an NRTS $\TrAbs$ with transition relation
  $\trelAbs$, which is the least relation satisfying that for all $p$, $\ell$ and $p'$,
  \begin{displaymath}
    p \trel (\ell, p') \implies p \trelAbs [a](\ell, p'),
  \end{displaymath}
  where either $a\#(\ell,p')$ and $\bn(\ell) = \emptyset$, or $\bn(\ell)=\{a\}$.

  We write $\TransAbs$ for the translation function, \ie,
  $\TrAbs = \TransAbs[\Tr,\bn]$.
\end{defi}

We prove that the transition relation $\trelAbs$ obtained by
Definition~\ref{def:nts-to-nrts} is equivariant and thus the translation produces an NRTS\@.

\begin{lem}[Equivariance of $\trelAbs$]%
  \label{lem:equivariance-R}
  The relation $\trelAbs$ obtained by Definition~\ref{def:nts-to-nrts} is
  equivariant. More formally, $p\trelAbs\resA{a}{\ell}{p'}$ implies $\pi\cdot
  p\trelAbs\resA{\pi\cdot a}{\pi\cdot \ell}{\pi\cdot p'}$ for every
  permutation $\pi$.
\end{lem}
\begin{proof}
  Let $a$ be an atom, $p$ and $p'$ be processes and $\ell$ be an action. We assume that
  ${p\trelAbs\resA{a}{\ell}{p'}}$, which has been obtained from $p\trel\resAF{\ell}{p'}$
  by Definition~\ref{def:nts-to-nrts}. Now we prove that
  ${\peract{\pi}{p}\trelAbs\peract{\pi}{\resA{a}{\ell}{p'}}}$ for every permutation
  $\pi$. We distinguish the following cases.
  \begin{description}
  \item[Case $\bn(\ell)=\emptyset$] Then $a\#\resAF{\ell}{p'}$. Since $\#$ is equivariant,
    we have that $(\pi\cdot a)\#\resAF{\pi\cdot\ell}{\pi\cdot p'}$. Since $\trel$ is
    equivariant, it follows that $\pi\cdot p\trel\resAF{\pi\cdot\ell}{\pi\cdot p'}$. By
    the translation function, and since $(\pi\cdot a)\#\resAF{\pi\cdot\ell}{\pi\cdot p'}$,
    we have that $\pi\cdot p\trelAbs\resA{\pi\cdot a}{\pi\cdot \ell}{\pi\cdot p'}$. Since
    $\bn$ is equivariant, $\bn(\peract{\pi}{\ell})$ is empty and we are done.
  \item[Case $\bn(\ell)=\{a\}$] Since $\bn$ is equivariant,
    $\bn(\pi\cdot \ell)=\{\pi\cdot a\}$. Since $\trel$ is equivariant,
    $\pi\cdot p\trel\resAF{\pi\cdot\ell}{\pi\cdot p'}$. By the translation function, and
    since $\bn(\pi\cdot \ell)=\{\pi\cdot a\}$, we have that
    $\pi\cdot p\trelAbs\resA{\pi\cdot a}{\pi\cdot \ell}{\pi\cdot p'}$.\qedhere
  \end{description}
\end{proof}
\begin{rem}
  Note that, as expected, the fact the transition relation $\rel{}$ in an NTS enjoys
  alpha-conversion of residuals does not play a role in the proof of the above
  result.\blockqed
\end{rem}

The translation from an NRTS $\TrAbs$ into an NTS $(\Tr,\bn)$ is given in the definition
below.
\begin{defi}%
  \label{def:nrts-to-nts}
  Let $\TrAbs$ be an NRTS with equivariant transition relation $\trelAbs$. %
  The NRTS $\TrAbs$ \emph{translates to} an NTS $(\Tr,\bn)$ with transition relation
  $\trel$, which is the least relation satisfying that for all $p$, $a$, $\ell$ and $p'$,
  \begin{displaymath}
    p \trelAbs [a](\ell, p') \implies
    p \trel (\peract{\tr{b}{a}}{\ell}, \peract{\tr{b}{a}}{p'})
  \end{displaymath}
  for $b = a$ and for each $b \#(\ell,p')$, and with binding-names function
  \begin{displaymath}
    \bn(\ell)=\{a\mid p \trelAbs \resA{a}{\ell}{p'}\land
    a\in\supp(\ell)\}.
  \end{displaymath}
  We write $\Trans$ for the translation function, \ie, $(\Tr,\bn) = \Trans[\TrAbs]$.
\end{defi}

Notice that if $a\in\supp(\ell,p')$, then $\Trans$ maps transition
$p\trelAbs [a](\ell, p')$ into every transition in
$\{p \trel (\peract{\tr{b}{a}}{\ell}, \peract{\tr{b}{a}}{p'}) \mid b = a \lor b
\#(\ell,p')\}$,
which encompasses the alpha-equivalence class of the target $[a](\ell, p')$.

We prove that the $\trel$ and $\bn$ obtained by
Definition~\ref{def:nrts-to-nts} are equivariant, and that $\trel$ enjoys
alpha-conversion of residuals. Thus, the translation is sound.

\begin{lem}[Equivariance of $\trel$]%
  \label{lem:equivariance-N}
  The relation $\trel$ obtained by Definition~\ref{def:nrts-to-nts} is
  equivariant. More formally, $p\trel\resAF{\ell}{p'}$ implies $\pi\cdot
  p\trel\resAF{\pi\cdot \ell}{\pi\cdot p'}$ for every permutation $\pi$.
\end{lem}
\begin{proof}
  Let $a$ be an atom, $p$ and $p'$ be processes and $\ell$ be an action.  We assume
  $p\trel\resAF{\ell}{p'}$ which is one of the transitions in the set
  $\{p\trel\resAF{\tr{b}{a}\cdot\ell}{\tr{b}{a}\cdot p'}\mid b=a\lor
  b\#\resAF{\ell}{p'}\}$
  obtained from $p\trelAbs\resA{a}{\ell}{p'}$ by Definition~\ref{def:nrts-to-nts}. Now we
  prove that for every permutation $\pi$, if
  $p\trel\resAF{\tr{b}{a}\cdot\ell}{\tr{b}{a}\cdot p'}$ where $b=a$ or
  $b\#\resAF{\ell}{p'}$, then
  $\pi\cdot p\trel\resAF{\pi\cdot\tr{b}{a}\cdot\ell}{\pi\cdot\tr{b}{a}\cdot p'}$. Since
  $\trelAbs$ is equivariant, we have that $p\trelAbs\resA{a}{\ell}{p'}$ implies
  $\pi\cdot p\trelAbs\pi\cdot\resA{a}{\ell}{p'}$. By Definition~\ref{def:nrts-to-nts}, for
  every atom $c=\peract{\pi}{a}$ or $c\#\resAF{\pi\cdot\ell}{\pi\cdot p'}$ we have that
  $\pi\cdot p\trel \resAF{\tr{c}{(\pi\cdot a)}\cdot\pi\cdot\ell} {\tr{c}{(\pi\cdot
      a)}\cdot\pi\cdot p'}$.
  Therefore it suffices to find such an atom $c$ that entails that
  $\tr{c}{(\pi\cdot a)}
  \cdot\pi\cdot\resAF{\ell}{p'}=\pi\cdot\tr{b}{a}\cdot\resAF{\ell}{p'}$.
  The latter equation holds by choosing $c=\pi\cdot b$. If $b = a$, then the
  transpositions $\tr{c}{(\pi\cdot a)}$ and $\tr{b}{a}$ are equal to $\iota$ and the
  equation above trivially follows. Otherwise, $b\#\resAF{\ell}{p'}$ and the equation
  above follows since $\pi\cdot b\#\pi\cdot\resAF{\ell}{p'}$ by equivariance of $\#$, and
  since
  $\tr{(\pi\cdot a)}{(\pi\cdot b)}\cdot\pi\cdot\resAF{\ell}{p'} =
  \peract{(\peract{\pi}{\tr{a}{b}})}{(\peract{\pi}{\resAF{\ell}{p'}})} =
  \pi\cdot\tr{a}{b}\cdot\resAF{\ell}{p'}$
  by equivariance of the permutation action.\qedhere
\end{proof}

\begin{lem}[Equivariance of $\bn$]
  The function $\bn$ obtained by Definition~\ref{def:nrts-to-nts} is equivariant. More
  formally, for every permutation $\pi$ and every action $\ell$ we have that
  $\bn(\pi\cdot \ell)=\pi\cdot\bn(\ell)$.
\end{lem}
\begin{proof}
  By calculating
  \begin{displaymath}
    \begin{array}[b]{rcl}
      \bn(\pi\cdot\ell)&=&\{a\mid p \trelAbs \resA{a}{\pi\cdot\ell}{p'} \in \TrAbs
                           \land a\in\supp(\pi\cdot\ell)\}\\
      &=&\myby{considering $a=\pi\cdot b$, $p=\pi\cdot q$ and $p'=\pi\cdot q'$}\\
      &&\{\pi\cdot b\mid\pi\cdot q \trelAbs
         \resA{\pi\cdot b}{\pi\cdot \ell}{\pi\cdot q'} \in \TrAbs
      \land \pi\cdot b\in\supp(\pi\cdot \ell)\}\\
      &=&\myby{equivariance of $\trelAbs$ and $\supp$}\\
      &&\pi\cdot \{b\mid q\trelAbs \resA{b}{\ell}{q'} \in \TrAbs \land b\in\supp(\ell)\}\\
      &=&\myby{Definition~\ref{def:nrts-to-nts}}\\
      &&\pi\cdot\bn(\ell).
      \tag*{\qedhere}
    \end{array}
  \end{displaymath}
\end{proof}

\begin{lem}[Alpha-conversion of residuals]
  Given $\trel$ and $\bn$ obtained by Definition~\ref{def:nts-to-nrts}, if
  $p\trel\resAF{\ell}{p'}$, $a\in\bn(\ell)$ and $b\#\resAF{\ell}{p'}$,
  then $p\trel\resAF{\tr{b}{a}\cdot\ell}{\tr{b}{a}\cdot p'}$.
\end{lem}
\begin{proof}
  Transition $p\trel\resAF{\ell}{p'}$ stems from a transition
  $p\trelAbs\resA{a}{\ell}{p'}$ and, since $a\in\bn(\ell)$, by
  Definition~\ref{def:nrts-to-nts} we know that $a\in\supp(\ell)$. Thus, $a$ is not fresh
  in $\resAF{\ell}{p'}$ and since $b\#\resAF{\ell}{p'}$ we have that $b\not=a$. Transition
  $p\trel\resAF{\tr{b}{a}\cdot\ell}{\tr{b}{a}\cdot p'}$ follows by
  Definition~\ref{def:nrts-to-nts} and we are done.\qedhere
\end{proof}

Although both the translations in Definitions~\ref{def:nts-to-nrts} and~\ref{def:nrts-to-nts} are sound, they are not the inverse of each other. Consider an NRTS
$\TrAbs$ that contains a transition $p\trelAbs[a](\ell,p')$ where $a\in\supp(p')$ and
$a\#\ell$. (Atom $a$ is abstracted over $p'$ but is fresh in $\ell$.)  By
Definition~\ref{def:nrts-to-nts}, $\TrAbs$ translates to an NTS
$(\Tr,\bn) = \Trans[\TrAbs]$ that contains a transition
$p\trel(\ell,\peract{\tr{b}{a}}{p'})$ for each atom $b=a\lor b\#(\ell,p')$ and where
$a\not\in\bn(\ell)$. (The exported name $a$ does not occur as a binding name of $\ell$.)
Taking the translation in Definition~\ref{def:nts-to-nrts} back, we obtain an NRTS
$\TrAbs '=\TransAbs[\Tr,\bn]$ that has a distinct transition
$p\rel{}[c](\ell,\peract{\tr{b}{a}}{p'})$ for each atom $b=a\lor b\#(\ell,p')$ and where
$c\#(\ell,\peract{\tr{b}{a}}{p'})$, but which does not contain the original transition
$p\rel{}[a](\ell,p')$---equal to any of its alpha-equivalent representations
$p\rel{}[b](\ell,\peract{\tr{b}{a}}{p'})$ with $b\#(\ell,p')$---because
$b\in\supp(\peract{\tr{b}{a}}{p'})$ and thus $c\not=b$. (The original transition with name
$a$ abstracted in the residual's body cannot be obtained back.) Therefore
$\TrAbs ' \not=\TrAbs$.

However, given an NTS $(\Tr,\bn)$, translating it to an NRTS with atom-abstractions in the
residuals and then back, delivers the same NTS $(\Tr,\bn)$. The following lemma states
this fact.
\begin{thm}%
  \label{thm-composition-identity}
  Let $(\Tr,\bn)$ be an NTS such that $|\bn(\ell)|\leq 1$ for every action $\ell$. Then,
  $\Trans[\TransAbs[\Tr,\bn]]=(\Tr,\bn)$.
\end{thm}
\begin{proof}
  Let $(\Tr',\bn')=\Trans[\TransAbs[\Tr,\bn]]$. We prove that $\Tr'=\Tr$ and that
  $\bn'=\bn$. Let $p\trel (\ell,p')$ be a transition in $\Tr$. It suffices to prove that
  \begin{itemize}
  \item transition $p\trel (\ell,p')$ maps through the composition of $\TransAbs$ and
    $\Trans$ to a set of transitions that contains itself, and such that every other
    transition in the set is already in $\Tr$, and
  \item $a\in\bn(\ell)$ iff $a\in\bn'(\ell)$.
  \end{itemize}

\noindent
  We consider the following cases.
  \begin{description}
  \item[Case $\bn(\ell)=\emptyset$] By Definition~\ref{def:nts-to-nrts}, transition
    $p\trel (\ell,p')$ maps to transition $p\trelAbs [a](\ell,p')$ with $a\#(\ell,p')$ in
    $\TransAbs[\Tr,\bn]$. By Definition~\ref{def:nrts-to-nts}, transition
    $p\trelAbs [a](\ell,p')$ maps to transition $p\trel (\ell,p')$ in $\Tr'$ because
    $a\#(\ell, p')$. Furthermore, by Definition~\ref{def:nrts-to-nts},
    $a\not\in\bn'(\ell)$ because $a\#\ell$.
  \item[Case $\bn(\ell)=\{a\}$] By Definition~\ref{def:nts-to-nrts}, transition
    $p\trel (\ell,p')$ maps to transition $p\trelAbs [a](\ell,p')$ in
    $\TransAbs[\Tr,\bn]$, where $a$ is abstracted in the residual $(\ell,p')$. By
    Definition~\ref{def:nrts-to-nts}, transition $p\trelAbs [a](\ell,p')$ maps to the set
    $T=\{p\trel (\peract{\tr{a}{b}}{\ell},\peract{\tr{a}{b}}{p'})\mid b = a \lor b\#(\ell,
    p')\}$
    in $\Tr'$. The set $T$ contains the original transition $p\trel (\ell,p')$ and, by
    alpha-conversion of residuals, every other transition in $T$ is in $\Tr$. Furthermore,
    by Definition~\ref{def:nrts-to-nts}, $a\in\bn'(\ell)$ because
    $a\in\supp(\ell)$.\qedhere
  \end{description}
\end{proof}

\noindent
In order to prove that the composition of the translations in the inverse order is the
identity---\ie, $\TransAbs[\Trans[\TrAbs]] = \TrAbs$---it suffices to prevent that the
abstracted atom in a residual occurs in the process but not in the action.

\begin{thm}%
    \label{thm-composition-identity-abstraction}
  Let $\TrAbs$ be an NRTS such that for every transition $p\trelAbs{}[a](\ell,p')$ in
  $\TrAbs$, $a\#\ell$ implies that $a\#p'$. Then, $\TransAbs[\Trans[\TrAbs]]=\TrAbs$.
\end{thm}
\begin{proof}
  Let $\TrAbs'=\TransAbs[\Trans[\TrAbs]]$. We prove that $\TrAbs'=\TrAbs$. Let
  $p\trelAbs [a](\ell,p')$ be a transition in $\TrAbs$. It suffices to prove that
  $p\trelAbs [a](\ell,p')$ maps through the composition of $\Trans$ and $\TransAbs$ to
  itself.

  We consider the following cases.
  \begin{description}
  \item[Case $a\#\ell$] By assumption, $a\#p'$. By Definition~\ref{def:nrts-to-nts},
    transition $p\trelAbs [a](\ell,p')$ maps to transition $p \trel (\ell, p')$ in $\Tr$,
    and $a\not\in\bn(\ell)$, where $\Trans[\TrAbs]=(\Tr,\bn)$. By
    Definition~\ref{def:nts-to-nrts}, since $a\#\ell$, transition $p \trel (\ell, p')$
    maps to transition $p\trelAbs [a](\ell,p')$ in $\TrAbs'$.
  \item[Case $a\in\supp(\ell)$] By Definition~\ref{def:nrts-to-nts}, transition
    $p\trelAbs [a](\ell,p')$ maps to every transition in the set
    $T=\{p\trel\resAF{\peract{\tr{b}{a}}{\ell}}{\peract{\tr{b}{a}}{p'}}\mid b=a\lor
    b\#\resAF{\ell}{p'}\}$
    in $\Tr$, and $a\in\bn(\ell)$, where $\Trans[\TrAbs]=(\Tr,\bn)$. For each $b$ such
    that $b=a$ or $b\#(\ell,p')$, by Definition~\ref{def:nrts-to-nts}, transition
    $p\trel\resAF{\peract{\tr{b}{a}}{\ell}}{\peract{\tr{b}{a}}{p'}}$ maps to transition
    $p\trelAbs [b](\peract{\tr{b}{a}}{\ell},\peract{\tr{b}{a}}{p'})$ in $\TrAbs'$ because
    $b\in\bn(\peract{\tr{b}{a}}{\ell})$. By definition of atom-abstraction,
    $[b](\peract{\tr{b}{a}}{\ell},\peract{\tr{b}{a}}{p'}) = [a](\ell,p')$ for every $b$
    such that $b=a$ or $b\#(\ell,p')$. Therefore, every transition in $T$ maps
    to $p\trelAbs [a](\ell,p')$ in $\TrAbs'$ and we are done.\qedhere
  \end{description}
\end{proof}

\noindent
Below we introduce a rule format for NRTSSs over signature $\SigmaNTSAbs$ that ensures
that the composition $\TransAbs \circ \Trans$ is the identity over the associated NRTS\@. To
this end we adapt the notion of partial strict stratification from
Definition~\ref{def:partial-strict stratification}.

\begin{defi}[Partial strict stratification with atom-abstractions]%
  \label{def:partial-stratification-abs}
  Let $\mathcal{R}^{[\Chan]}$ be an NRTSS over a signature $\SigmaNTSAbs$. Let
  $S^{[\Chan]}$ be a partial map from ground nominal terms of sort
  $\Proc\times[\Chan]\Act$ to ordinal numbers. $S^{[\Chan]}$ is a \emph{partial strict
    stratification with atom-abstractions} of $\mathcal{R}^{[\Chan]}$ iff
  \begin{enumerate}[label={(\roman*)},ref={(\roman*)}]
  \item $S^{[\Chan]}(\varphi(t), [a]\ell) \not= \bot$, for every rule in $\R^{[\Chan]}$
    with conclusion $t \rel{} [a](\ell, t')$ such that $a\#\ell$ and for every ground
    substitution $\varphi$, and
  \item $S^{[\Chan]}(\varphi(u_i),[a_i]\ell_i)<S^{[\Chan]}(\varphi(t),[a]\ell)$ and
    $S^{[\Chan]}(\varphi(u_i),[a_i]\ell_i)\not=\bot$, for every rule $\Ru$ in
    $\R^{[\Chan]}$ with conclusion $t\rel{}[a](\ell,t')$ such that
    $S^{[\Chan]}(\varphi(t),[a]\ell)\not=\bot$, for every premiss
    $u_i\rel{}[a_i](\ell_i,u'_i)$ of $\Ru$ such that $a_i\# \ell_i$ and for every ground
    substitution $\varphi$.
  \end{enumerate}
  We say a ground nominal term $(p,[a]\ell)$ of sort $\Proc\times[\Chan]\Act$ \emph{has
    order} $S^{[\Chan]}(p,[a]\ell)$.
\end{defi}

The choice of $S^{[\Chan]}$ determines which rules will be considered by the rule format
for NRTSSs defined below, which guarantees that for every transition
$p\trelAbs [a](\ell,p')$ in the induced transition relation, $a\#\ell$ implies $a\#p'$. We
will intend the map $S^{[\Chan]}$ to be such that the only rules whose source, abstracted
atom and label of the conclusion have defined order are those that may take part in proof
trees of transitions where the abstracted atom in its residual is fresh in its action.

\begin{defi}[Binding-actions format]%
  \label{def:binding-actions-format}
  Let $\mathcal{R}^{[\Chan]}$ be an NRTSS over a signature $\SigmaNTSAbs$ and
  $S^{[\Chan]}$ be a partial strict stratification with atom-abstractions of
  $\mathcal{R}^{[\Chan]}$. Assume that all the actions occurring in the rules of
  $\mathcal{R}^{[\Chan]}$ are ground. Let
  \begin{mathpar}
    \inferrule*[right=Ru] {\{u_i\rel{}[a_i](\ell_i,u'_i)\mid i\in I\} \qquad\nabla}
    {t\rel{}[a](\ell,t')}
  \end{mathpar}
  be a rule in $\mathcal{R}^{[\Chan]}$. The rule \textsc{Ru} is in \emph{binding-actions
    format with respect to $S^{[\Chan]}$} (\emph{BA format with respect to $S^{[\Chan]}$}
  for short) iff either $a\in\supp(\ell)$, or otherwise the following holds:
  \begin{displaymath}
    \nabla \cup \{a_i\fra u'_i \mid i\in I \land a_i\# \ell_i\}
    \vdash \{a \fra t'\}.
  \end{displaymath}
  An NRTSS $\mathcal{R}^{[\Chan]}$ is in \emph{BA format with respect to} $S^{[\Chan]}$
  iff all the rules in $\mathcal{R}^{[\Chan]}$ are in BA format with respect to
  $S^{[\Chan]}$.
\end{defi}

\begin{thm}%
  \label{thm:binding-actions-format}
  Let $\mathcal{R}^{[\Chan]}$ be an NRTSS over a signature $\SigmaNTSAbs$ and
  $S^{[\Chan]}$ be a partial strict stratification with atom-abstractions of
  $\mathcal{R}^{[\Chan]}$. Assume that $\mathcal{R}^{[\Chan]}$ is in BA format with
  respect to $S^{[\Chan]}$ and let $\TrAbs$ be the NRTS induced by
  $\mathcal{R}^{[\Chan]}$. Then, for every transition $p\trelAbs{}[a](\ell,p')$ in
  $\TrAbs$, $a\#\ell$ implies that $a\# p'$.
\end{thm}
\begin{proof}
  Let $\NT{p}\trelAbs{}\NT{[a](\ell,p')}$ be provable in $\mathcal{R}^{[\Chan]}$ and
  assume that the last rule used in the proof of $\NT{p}\trelAbs{}\NT{[a](\ell,p')}$ is
  \begin{mathpar}
    \inferrule*[right=Ru]
    {\{u_i\rel{}[a_i](\ell_i,u'_i)\mid i\in I\}
      \qquad\{a_j\fra v_j\mid j\in J\}}
    {t\rel{}[a](\ell,t')}
  \end{mathpar}
  where $I$ and $J$ are disjoint. Therefore, for some ground substitution $\varphi$,
  \begin{itemize}
  \item $\NT{p} = \NT{\varphi(t)}$ and $\NT{[a](\ell,p')} = \NT{[a](\ell,\varphi(t'))}$,
  \item the premisses $\NT{\varphi(u_i)} \rel{} \NT{[a_i](\ell_i,\varphi(u'_i))}$ with
    $i\in I$ are provable in $\R$, and
  \item the freshness relations $a_j\#\NT{\varphi(v_j)}$ with $j\in J$ hold.
  \end{itemize}
  Recall that the actions $\ell$ and $\ell_i$ where $i\in I$ are ground, and thus
  $\varphi$ is not applied to the actions in the items above. Since
  $\NT{[a](\ell,\varphi(t'))}=\NT{[a](\ell,p')}$, by Definition~\ref{def:interpretation}
  and Remark~\ref{rem:equality-bodies}, $\NT{\varphi(t')}=\NT{p'}$.

  We need to prove that $a\#\ell$ implies $a\#\NT{p'}$. If $a\in\supp(\ell)$ then $a$ is
  not fresh in $\ell$ and we are done. Otherwise, we know that $a\#\ell$ and we show that
  $a\#\NT{\varphi(t')}$. Observe that $S^{[\Chan]}(\varphi(t),[a]\ell)$ is defined because
  $a\#\ell$. We proceed by induction on $S^{[\Chan]}(\varphi(t),[a]\ell)$.

  Since rule \Ru\ is in BA format with respect to $S^{[\Chan]}$,
  \begin{displaymath}
    \{a_j\fra v_j\mid j\in J\} \cup \{a_i\fra u'_i \mid i\in I \land a_i\# \ell_i\}
    \vdash \{a \fra t'\}.
  \end{displaymath}
  We use Lemma~\ref{lem:entails} to obtain the implication
  \begin{enumerate}[label={(\arabic*)},ref={(\arabic*)}]
  \item\label{eq:bc-format}
    $\bigwedge_{j\in J} (a_j\# \NT{\varphi(v_j)}) \land \bigwedge_{i\in I \land
      a_i\#\ell_i} (a_i\# \NT{\varphi(u'_i)}) \implies a \# \NT{\varphi(t')}$.
  \end{enumerate}
  By the existence of the proof tree, all the $a_j\# \NT{\varphi(v_j)}$ with $j\in J$
  hold, and it suffices to prove
  $\bigwedge_{i\in I \land a_i\#\ell_i} (a_i\# \NT{\varphi(u'_i)})$. The base case is when
  $S^{[\Chan]}(\varphi(t),[a]\ell)$ is minimal. By
  Definition~\ref{def:partial-stratification-abs} the rule \textsc{Ru} has no premisses
  and the set $I$ is empty, which makes
  $\bigwedge_{i\in I \land a_i\#\ell_i} (a_i\# \NT{\varphi(u'_i)})$ trivially true and we
  are done. Now we assume that $S^{[\Chan]}(\varphi(t),[a]\ell)$ is not minimal. Condition
  (ii) in Definition~\ref{def:partial-stratification-abs} ensures that
  $S^{[\Chan]}(\varphi(t),[a]\ell)\not=\bot$ and
  $S^{[\Chan]}(\varphi(u_i),[a_i]\ell_i) < S^{[\Chan]}(\varphi(t),[a]\ell)$ for every
  $i\in I$ such that $a_i\#\ell_i$. Thus, we can apply the induction hypothesis to obtain
  $a_i\# \NT{\varphi(u'_i)}$ for every $i \in I$ such that $a_i\#\ell_i$ and the theorem
  holds.\end{proof}

\section{Example of Application of the BA-Format to the \texorpdfstring{$\pi$}{pi}-Calculus}%
\label{sec:application-BA-format}

In this section we introduce the NRTSSs with residuals of abstraction sort $\REAbs$ and
$\RLAbs$, which respectively define our versions of the early and the late semantics of
the $\pi$-calculus. For each of these semantics, we aim at showing that the induced NTRSs
with and without residuals of abstraction sort represent the same model of computation, in
the sense that $\TrEAbs = \TransAbs(\TrE,\bnE)$ and $(\TrE,\bnE) = \Trans(\TrEAbs)$ (and
respectively for $(\TrL,\bnL)$ and $\TrLAbs$). Since we have already checked that both
$\RE$ and $\RL$ are in ACR-format in Section~\ref{sec:example-nts}, in order to establish
that the models of computation are the same we need to check that
\begin{enumerate}[label={(\roman*)},ref={(\roman*)}]
\item both $\REAbs$ and $\RLAbs$ are in BA-format, and
\item $\TrEAbs = \TransAbs(\TrE,\bnE)$, where $\TrEAbs$ is induced by $\REAbs$ (and
  respectively for $\TrLAbs$ and $(\TrL,\bnL)$).
\end{enumerate}
The translations between these systems with and without residuals of abstraction sort are
inverse to each other, and thus the two-way correspondence holds.

\subsection{Early Semantics of the \texorpdfstring{$\pi$}{pi}-Calculus}
Consider the NRTSS $\REAbs$ in Figure~\ref{fig:early-pi-atom-abstraction} for our version
of the early semantics $\pi$-calculus~\cite{MPW92} over the residual signature
$\SigmaNTSAbs$ as defined on page~\pageref{pag:sigma-nts-abs} in
Section~\ref{def:nts-to-nrts}, where $F$ is the set of function symbols from
Example~\ref{ex:pi-calculus}. Omitted rules \textsc{AParR}, \textsc{AECommR},
\textsc{AECloseR} and \textsc{ASumR} are, respectively, the symmetric version of rules
\textsc{AParL}, \textsc{AECommL}, \textsc{AECloseL} and \textsc{ASumL}.

\begin{figure}[ht]
\begin{mathpar}
  \inferrule*[right=AEIn]
  {d\fra (a,c,\susp{x}{\rep{b}{c}})}
  {\inPA[a]{[b]x}\rel{}\resA{d}{\inAA[a]{c}}{\susp{x}{\rep{b}{c}}}}
  \and
  \inferrule*[right=AOut]
  {c\fra (a,b,x)}
  {\outPA[a]{b}{x}\rel{}\resA{c}{\outAA[a]{b}}{x}}
  \and
  \inferrule*[right=ATau]
  {a\fra x}
  {\tauPA[x]\rel{}\resA{a}{\tauAA}{x}}
  \and
  \inferrule*[right=AParL]
  {x_1\rel{}\resA{a}{\ell}{y_1}\qquad a\fra x_2}
  {\parPA[x_1]{x_2} \rel{}\resA{a}{\ell}{(\parPA[y_1]{x_2})}}
  \\
  \inferrule*[right=AECommL]
  {x_1\rel{}\resA{c}{\outAA[a]{b}}{y_1}
    \qquad x_2\rel{}\resA{c}{\inAA[a]{b}}{y_2}}
  {\parPA[x_1]{x_2}\rel{}\resA{c}{\tauAA}{(\parPA[y_1]{y_2})}}
  \\
  \inferrule*[right=AECloseL]
  {x_1\rel{}\resA{b}{\boutAA[a]{b}}{y_1}
    \qquad x_2\rel{}\resA{c}{\inAA[a]{b}}{y_2} \qquad b\fra x_2 \qquad c\fra y_1}
  {\parPA[x_1]{x_2}\rel{}\resA{c}{\tauAA}{\newPA{([b](\parPA[y_1]{y_2}))}}}
  \\
  \inferrule*[right=ASumL]
  {x_1\rel{}\resA{a}{\ell}{y_1}}
  {\sumPA[x_1]{x_2} \rel{}\resA{a}{\ell}{y_1}}
  \and
  \inferrule*[right=ARep]
  {x\rel{}\resA{a}{\ell}{y}\qquad a\fra x}
  {\repPA[x] \rel{}\resA{a}{\ell}{(\parPA[y]{\repPA[x]})}}
  \and
  \inferrule*[right=AERepComm]
  {x\rel{}\resA{c}{\outAA[a]{b}}{y_1}\qquad x\rel{}\resA{c}{\inAA[a]{b}}{y_2}
    \qquad c\fra x}
  {\repPA[x]\rel{}\resA{c}{\tauAA}{\parPA[{\parPA[y_1]{y_2}}]{\repPA[x]}}}
  \\
  \inferrule*[right=AERepClose]
  {x\rel{}\resA{b}{\boutAA[a]{b}}{y_1}
    \qquad x\rel{}\resA{c}{\inAA[a]{b}}{y_2}\qquad b\fra x \qquad c\fra (x, y_1)}
  {\repPA[x]\rel{}\resA{c}{\tauAA}{\parPA[{\newPA{([b](\parPA[y_1]{y_2}))}}]{\repPA[x]}}}
  \\
  \inferrule*[right=AOpen]
  {x\rel{}\resA{c}{\outAA[a]{b}}{y}\qquad b\fra a}
  {\newPA{([b]x)}\rel{}\resA{b}{\boutAA[a]{b}}{y}}
  \and
  \inferrule*[right=ARes]
  {x\rel{}\resA{a}{\ell}{y}\qquad b\fra \ell}
  {\newPA{([b]x)} \rel{}\resA{a}{\ell}{\newPA{([b]y)}}}
\end{mathpar}
\begin{center}
  where $a,b,c,d\in \Atom_{\Chan}$ and $\ell$ is a ground action.
\end{center}
\caption{NRTSS for the early $\pi$-calculus with atom-abstractions in the residuals.}%
\label{fig:early-pi-atom-abstraction}
\end{figure} %

We use the rule format of Definition~\ref{def:binding-actions-format} to show that for
every transition $p\trelAbs [a](\ell,p')$ in $\TrEAbs$, $a\#\ell$ implies $a\#p'$. We
consider the following partial strict stratification with atom abstractions
\begin{displaymath}
\begin{array}{rcl}
S^{[\Chan]}(\inPA[a]{[b]p},[d]\inAA[a]{c})&=&0\\
S^{[\Chan]}(\outPA[a]{b}{p},[c]\outAA[a]{b})&=&0\\
S^{[\Chan]}(\tauPA[p],[a]\tauAA)&=&0\\
S^{[\Chan]}(\parPA[p]{q},[a]\ell)&=& 1+\max\{
{\setlength\arraycolsep{0pt}\begin{array}[t]{l}
  S^{[\Chan]}(p,[a]\ell),S^{[\Chan]}(q,[a]\ell),\\
  S^{[\Chan]}(p,[c](\outAA[a]{b})),\\
  S^{[\Chan]}(q,[c](\inPA[a]{b}))\}\\
\end{array}}\\
S^{[\Chan]}(\sumPA[p]{q},[a]\ell)&=&
1+\max\{S^{[\Chan]}(p,[a]\ell),S^{[\Chan]}(q,[a]\ell)\}\\
S^{[\Chan]}(\repPA[p],[a]\ell)&=& 1+\max\{
{\setlength\arraycolsep{0pt}\begin{array}[t]{l}
  S^{[\Chan]}(p,[a]\ell),\\
  S^{[\Chan]}(p,[c](\outAA[a]{b})),\\
  S^{[\Chan]}(p,[c](\inAA[a]{b}))\}
\end{array}}\\
S^{[\Chan]}(\newPA[{[b]p}],[a]\ell)&=&1+S^{[\Chan]}(p,[a]\ell)\\
S^{[\Chan]}(p,t)&=&\bot\quad \text{otherwise},
\end{array}
\end{displaymath}
where $a,b\in\Atom_{\Chan}$ and
$\ell\in\{\inAA[a]{b},\outAA[a]{b},\tauAA\mid a,b\in\Atom_{\Chan}\}$.

We check that $\REAbs$ is in BA-format with respect to $S^{[\Chan]}$ as follows. Consider
a transition $p\trelAbs [b](\ell, p')$. The abstracted atom $b$ is in the support of
$\ell$ iff $\ell \in \{\boutAA[a]{b} \mid a\in \Atom_{\Chan}\}$. From the last clause in
the definition of $S^{[\Chan]}$ above, we have that
$S^{[\Chan]}(p,[b]\boutAA[a]{b}) = \bot$, for each $p$ and $a,b\in\Atom_{\Chan}$. Observe
that $S^{[\Chan]}$ meets Definition~\ref{def:partial-stratification-abs}(i) because a % chktex 36
formula with a residual $[b](\boutAA[a]{b},p')$ does not take part in any proof tree that
proves a transition with a residual $[b](\ell,p'')$ such that $b\#\ell$. Therefore, the
only rules in $\REAbs$ whose sources, abstracted atoms, and actions have defined order are
\textsc{AEIn}, \textsc{AOut}, \textsc{ATau}, \textsc{AERepComm}, \textsc{AERepClose},
rules \textsc{AECommL}, \textsc{AECloseL} and their symmetric versions, and the instance
of rules \textsc{AParL}, \textsc{ASumL}, \textsc{ARep}, \textsc{ARes} where $t=[a]\ell$
and $a\#\ell$ (and the corresponding instance of the symmetric versions \textsc{AParR} and
\textsc{ASumR}). We will not check the BA-format for the symmetric versions of the
rules. Observe that $S^{[\Chan]}$ meets
Definition~\ref{def:partial-stratification-abs}(ii) because for each rule whose conclusion % chktex 36
has a residual $[b](\ell,p')$ such that $b\#\ell$, the order of the transition that
unifies with its conclusion is always bigger than the order of the transitions that unify
with those premisses $u_i\rel{}[a_i](\ell_i,p'_i)$ with $a_i\#\ell_i$.

The condition of the rule format is trivial to check in all these rules. We show some of
them for illustration.

For rule \textsc{AEIn}, we need to check that
$\{d\fra (a,c,\susp{x}{\rep{b}{c}})\} \vdash \{d\fra \susp{x}{\rep{b}{c}}\}$, which
trivially holds.

For rule \textsc{AERepComm}, we need to check that
\begin{displaymath}
  \{c\fra y_1\} \cup \{c\fra y_2\} \cup \{c\fra x\} \vdash
  \{c\fra \parPA[{\parPA[y_1]{y_2}}]{\repPA[x]}\},
\end{displaymath}
which trivially holds.

For rule \textsc{AParL}, it suffices to consider the instance where $a\#\ell$, and we need
to check that $\{a\fra x_2\} \cup \{a\fra y_1\} \vdash \{a\fra \parPA[y_1]{x_2}\}$, which
trivially holds.

Atoms $a$, $b$, $c$, and $d$ in $\REAbs$ range over $\Atom_{\Chan}$, and thus $\REAbs$ is
in equivariant format. Since $\REAbs$ is in the BA-format with respect to $S^{[\Chan]}$,
by Theorems~\ref{thm:binding-actions-format} and~\ref{thm-composition-identity-abstraction}, $\TransAbs[\Trans[\TrEAbs]]=\TrEAbs$. Since
$\RE$ is in ACR-format and by Theorem~\ref{thm-composition-identity},
$\Trans[\TransAbs[\TrE,\bnE]]=(\TrE,\bnE)$. Thus, in order to show that $\TrE$ and
$\TrEAbs$ represent the same model of computation, it suffices to check that
$\TransAbs(\TrE,\bnE)=\TrEAbs$.

\begin{lem}%
  \label{lem:tre-translates-treabs}
  Let $\TrE$ be the NRTS induced by $\RE$ in Figure~\ref{fig:early-pi-NTS} of
  Section~\ref{sec:early-pi-calculus}, $\bnE$ be its associated binding-names function,
  and $\TrEAbs$ be the NRTS induced by $\REAbs$. Then, $\TransAbs(\TrE,\bnE)=\TrEAbs$.
\end{lem}

The proof of Lemma~\ref{lem:tre-translates-treabs} is in
Appendix~\ref{ap:example-application-BA-format}. By Lemma~\ref{lem:tre-translates-treabs},
the NTS $(\TrE,\bnE)$ and the NRTS $\TrEAbs$ represent the same model of computation.

\subsection{Late Semantics of the \texorpdfstring{$\pi$}{pi}-Calculus}%
\label{sec:late-semantics-residual-atom-abstraction}

Consider the NRTSS $\RLAbs$ that consists of the rules in
Figure~\ref{fig:late-pi-atom-abstraction} together with rules \textsc{AOut},
\textsc{ATau}, \textsc{AParL}, \textsc{ASumL}, \textsc{ARep}, \textsc{AOpen} and
\textsc{ARes} in Figure~\ref{fig:early-pi-atom-abstraction} of
Section~\ref{lem:tre-translates-treabs}, and the symmetric versions \textsc{AParR},
\textsc{ALCommR}, \textsc{ALCloseR} and \textsc{ASumR}.

As we did in Section~\ref{sec:late-pi-calculus}, we replace the free-input actions by
bound-input actions, written $\binAA[a]{b}$.

\begin{figure}[ht]
\begin{mathpar}
  \inferrule*[right=ALIn]
  {b\fra a}
  {\inPA[a]{[b]x}\rel{}\resA{b}{\binAA[a]{b}}{x}}
  \and
  \inferrule*[right=ALCommL]
  {x_1\rel{}\resA{d}{\outAA[a]{b}}{y_1}
    \qquad x_2\rel{}\resA{c}{\binAA[a]{c}}{y_2} \qquad d\fra \susp{y_2}{\rep{b}{c}}}
  {\parPA[x_1]{x_2}\rel{}\resA{d}{\tauAA}{(\parPA[y_1]{\susp{y_2}{\rep{b}{c}}})}}
  \\
  \inferrule*[right=ALCloseL]
  {x_1\rel{}\resA{b}{\boutAA[a]{b}}{y_1}
    \qquad x_2\rel{}\resA{b}{\binAA[a]{b}}{y_2} \qquad c\fra [b](y_1,y_2)}
  {\parPA[x_1]{x_2}\rel{}\resA{c}{\tauAA}{\newPA{([b](\parPA[y_1]{y_2}))}}}
  \\
  \inferrule*[right=ALRepComm]
  {x\rel{}\resA{d}{\outAA[a]{b}}{y_1}\qquad x\rel{}\resA{c}{\binAA[a]{c}}{y_2}
    \qquad d\fra (x,\susp{y_2}{\rep{b}{c}})}
  {\repPA[x]\rel{}
    \resA{d}{\tauAA}{\parPA[{\parPA[y_1]{\susp{y_2}{\rep{b}{c}}}}]{\repPA[x]}}}
  \\
  \inferrule*[right=ALRepClose]
  {x\rel{}\resA{b}{\boutAA[a]{b}}{y_1}
    \qquad x\rel{}\resA{b}{\binAA[a]{b}}{y_2}\qquad c\fra (x, [b](y_1,y_2))}
  {\repPA[x]\rel{}\resA{c}{\tauAA}{\parPA[{\newPA{([b](\parPA[y_1]{y_2}))}}]{\repPA[x]}}}
\end{mathpar}
\begin{center}
  where $a,b,c,d\in \Atom_{\Chan}$. %and $\ell$ is a ground action.
\end{center}
\caption{NRTSS for the late $\pi$-calculus with atom-abstractions in the residuals.}%
\label{fig:late-pi-atom-abstraction}
\end{figure} %

In contrast with $\RL$ in Section~\ref{sec:late-pi-calculus}, rules \textsc{ALCommL} and
\textsc{ALRepComm} in Figure~\ref{fig:late-pi-atom-abstraction} do not use moderated terms
because the communication involving bound-input actions does not require renaming of
channel names, since the channel through which communication takes place is abstracted in
the residual of the input process.

\begin{rem}
  Similar to rule \textsc{LIn} in Section~\ref{sec:late-pi-calculus}, and as commented in
  Remark~\ref{rem:prevent-capture-in}, rule \textsc{ALIn} ensures that the binding atom
  $b$ is different from the communication channel. Our semantics allows for the transition
  \begin{displaymath}
    \begin{array}{l}
      \NT{\parPA[{\outPA[a]{a}{\nullPA}}]{\inPA[a]{[b](\outPA[c]{b}{\nullPA})}}} \rel{} \\
      \NT{\resA{b}{\tauAA}{\parPA[{\nullPA}]{\ren{(\outPA[c]{b}{\nullPA})}{\rep{a}{b}}}}}
      \hspace{-0.4mm}=\hspace{-0.4mm} \NT{\resA{b}{\tauAA}{\parPA[{\nullPA}]{\outPA[c]{a}{\nullPA}}}},
    \end{array}
  \end{displaymath}
  where $b\#(a,c)$, which models both
  $(\overline{a}a.0\parallel a(b).\overline{c}b.0)\lowrel{\tau} (0\parallel\overline{c}a.0)$
  and
  $(\overline{a}a.0\parallel a(a).\overline{c}a.0)\lowrel{\tau} (0\parallel \overline{c}a.0)$
  in the original late $\pi$-calculus.\blockqed
\end{rem}

We use the rule format of Definition~\ref{def:binding-actions-format} to show that for
every transition $p\trelAbs [a](\ell,p')$ in $\TrLAbs$, $a\#\ell$ implies $a\#p'$. We
consider the following partial strict stratification with atom abstractions
\begin{displaymath}
\begin{array}{rcl}
S^{[\Chan]}(\outPA[a]{b}{p},[c]\outAA[a]{b})&=&0\\
S^{[\Chan]}(\tauPA[p],[a]\tauAA)&=&0\\
S^{[\Chan]}(\parPA[p]{q},[a]\ell)&=& 1+\max\{
{\setlength\arraycolsep{0pt}\begin{array}[t]{l}
  S^{[\Chan]}(p,[a]\ell),S^{[\Chan]}(q,[a]\ell),\\
  S^{[\Chan]}(p,[c](\outAA[a]{b})),\\
  S^{[\Chan]}(q,[c](\inPA[a]{b}))\}\\
\end{array}}\\
S^{[\Chan]}(\sumPA[p]{q},[a]\ell)&=&
1+\max\{S^{[\Chan]}(p,[a]\ell),S^{[\Chan]}(q,[a]\ell)\}\\
S^{[\Chan]}(\repPA[p],[a]\ell)&=&1+\max\{
{\setlength\arraycolsep{0pt}\begin{array}[t]{l}
  S^{[\Chan]}(p,[a]\ell),\\
  S^{[\Chan]}(p,[c](\outAA[a]{b})),\\
  S^{[\Chan]}(p,[c](\inAA[a]{b}))\}
\end{array}}\\
S^{[\Chan]}(\newPA[{[b]p}],[a]\ell)&=&1+S^{[\Chan]}(p,[a]\ell)\\
S^{[\Chan]}(p,t)&=&\bot\quad \text{otherwise},
\end{array}
\end{displaymath}
where $a,b\in\Atom_{\Chan}$ and
$\ell\in\{\outAA[a]{b},\tauAA\mid a,b\in\Atom_{\Chan}\}$.

We check that $\RLAbs$ is in BA-format with respect to $S^{[\Chan]}$ as follows. Consider
a transition $p\trelAbs [b](\ell, p')$. The abstracted atom $b$ is in the support of
$\ell$ iff $\ell \in \{\boutAA[a]{b}, \binAA[a]{b} \mid a\in \Atom_{\Chan}\}$. By the
definition of $S^{[\Chan]}$, we have that
\begin{displaymath}
  S^{[\Chan]}(p,[b]\boutAA[a]{b}) = S^{[\Chan]}(p,[b]\binAA[a]{b}) = \bot,
\end{displaymath}
for each $p$ and $a,b\in\Atom_{\Chan}$. Observe that $S^{[\Chan]}$ meets
Definition~\ref{def:partial-stratification-abs}(i) because a formula with either a % chktex 36
residual $[b](\boutAA[a]{b},p')$ or $[b](\binAA[a]{b},p')$ does not take part in any proof
tree that proves a transition with a residual $[b](\ell,p'')$ such that
$b\#\ell$. Therefore, the only rules in $\RLAbs$ whose sources, abstracted atoms, and
actions have defined order are \textsc{AOut}, \textsc{ATau}, \textsc{ALRepComm},
\textsc{ALRepClose}, rules \textsc{ALCommL}, \textsc{ALCloseL} and their symmetric
versions, and the instance of rules \textsc{AParL}, \textsc{ASumL}, \textsc{ARep},
\textsc{ARes} where $t=[a]\ell$ and $a\#\ell$ (and the corresponding instance of the
symmetric versions \textsc{AParR} and \textsc{ASumR}). Observe that $S^{[\Chan]}$ meets
Definition~\ref{def:partial-stratification-abs}(ii) because for each rule whose conclusion % chktex 36
has a residual $[b](\ell,p')$ such that $b\#\ell$, the order of the transition that
unifies with its conclusion is always bigger than the order of the transitions that unify
with those premisses $u_i\rel{}[a_i](\ell_i,p'_i)$ with $a_i\#\ell_i$.

We will limit ourselves to checking that rules \textsc{ALCommL} and \textsc{ALRepComm} are
in BA-format with respect to $S^{[\Chan]}$, since the checks for the other rules are
similar to those presented earlier.

For rule \textsc{ALCommL}, we need to check that
\begin{displaymath}
  \{d\fra \susp{y_2}{\rep{b}{c}}\} \cup \{d\fra y_1\} \vdash
  \{d\fra \parPA[y_1]{\susp{y_2}{\rep{b}{c}}}\},
\end{displaymath}
which trivially holds.

For rule \textsc{ALRepComm}, we need to check that
\begin{displaymath}
  \{d\fra (x,\susp{y_2}{\rep{b}{c}})\} \cup \{d\fra y_1\} \vdash
  \{d\fra \parPA[{\parPA[y_1]{\susp{y_2}{\rep{b}{c}}}}]{\repPA[x]}\},
\end{displaymath}
which trivially holds.

Atoms $a$, $b$, $c$ and $d$ in $\RLAbs$ range over $\Atom_{\Chan}$, and thus $\RLAbs$ is
in equivariant format. Since $\RLAbs$ is in the BA-format with respect to $S^{[\Chan]}$,
by Theorems~\ref{thm:binding-actions-format} and~\ref{thm-composition-identity-abstraction}, $\TransAbs[\Trans[\TrLAbs]]=\TrLAbs$. Since
$\RL$ is in ACR-format and by Theorem~\ref{thm-composition-identity},
$\Trans[\TransAbs[\TrL,\bnL]]=(\TrL,\bnL)$. Thus, in order to show that $\TrL$ and
$\TrLAbs$ represent the same model of computation, it suffices to check that
$\TransAbs(\TrL,\bnL)=\TrLAbs$.

\begin{lem}%
  \label{lem:trl-translates-trlabs}
  Let $\TrL$ be the NRTS induced by $\RL$ in Figure~\ref{fig:late-pi-NTS} of
  Section~\ref{sec:late-pi-calculus}, $\bnL$ be its associated binding-names function,
  and $\TrLAbs$ be the NRTS induced by $\RLAbs$. Then, $\TransAbs(\TrL,\bnL)=\TrLAbs$.
\end{lem}

The proof of Lemma~\ref{lem:trl-translates-trlabs} is in
Appendix~\ref{ap:example-application-BA-format}. By Lemma~\ref{lem:trl-translates-trlabs},
the NTS $(\TrL,\bnL)$ and the NRTS $\TrLAbs$ represent the same model of computation.

\section{Related and Future Work}%
\label{sec:conclusions}

The work we have presented in this paper stems from the Nominal SOS (NoSOS)
framework~\cite{CMRG12} and from earlier proposals for nominal logic
in~\cite{UPG04,CP07,GM09}. It is by no means the only approach studied so far in the
literature that aims at a uniform treatment of binders and names in programming and
specification languages. Other existing approaches that accommodate variables and binders
within the SOS framework are those proposed by Fokkink and Verhoef in~\cite{FV98}, by
Middelburg in~\cite{Mid01,Mid03}, by Bernstein in~\cite{Ber98}, by Ziegler, Miller and
Palamidessi in~\cite{ZMP06} and by Fiore and Staton in~\cite{FS09} (originally, by Fiore
and Turi in~\cite{FT01}). The aim of all of the above-mentioned frameworks is to establish
sufficient syntactic conditions guaranteeing the validity of a semantic result (congruence
in the case of~\cite{Ber98,Mid01,ZMP06,FS09} and conservativity in the case
of~\cite{FV98,Mid03}). In addition, Gabbay and Mathijssen present a nominal axiomatisation
of the $\lambda$-calculus in~\cite{GM10}. None of these approaches addresses equivariance
nor the property of alpha-conversion of residuals in~\cite{PBEGW15}. The proposal that is
closest to ours is the one in~\cite{FS09}. In that paper, Fiore and Staton presented a
GSOS-like rule format for name-passing process calculi, where operational specifications
corresponds to theories in nominal logic, and show that a natural notion of bisimilarity
is preserved by operations specified in that format.

Nominal techniques have been implemented also in programming languages. This is the case
of FreshML~\cite{SPG03} where Shinwell, Pitts and Gabbay extend ML with constructs for
defining and working with data involving binding operations. In particular, FreshML adds
the keyword \emph{fresh} to ML in order to generate a fresh new name in an expression
inside the code.

In~\cite{GH08}, Gabbay and Hofmann study the category $\mathbf{Ren}$, which is a
generalisation of $\mathbf{Nom}$ in which the renamings play the role of the permutations
in $\mathbf{Nom}$. In this paper we use the action of renaming to define the
interpretation of moderated terms into nominal terms, which are akin to the nominal
algebraic datatypes of Pitts~\cite{Pit13} and live in $\mathbf{Nom}$.

Our moderated terms $\susp{t}{\rho}$ are reminiscent of the terms in calculi with explicit
substitutions in the tradition of the $\lambda\sigma$ calculus of Abadi
{\etal}~\cite{ACCL91}. In fact, a replacement $\rep{a}{b}$ is an instance of an explicit
substitution that substitutes atom $b$ with atom $a$, and both the renamings and the
explicit substitutions are composable, as the equality
\begin{displaymath}
  \NT{\susp{\susp{t}{\rho_1}}{\rho_2}} = \NT{\ren{(\susp{t}{\rho_1})}{\rho_2}} =
  \NT{\susp{t}{\rho_1;\rho_2}}
\end{displaymath}
in this paper and the equality $a[s][t] = a[s\circ t]$ in~\cite{ACCL91} witness. However,
renamings and explicit substitutions are very different objects: a renaming is a semantic
object in $\Atom \to_\fs\Atom$ that maps atoms to atoms, while an explicit substitution is
a syntactic object that specifies how to substitute placeholders with arbitrary terms.

In~\cite{MT05}, Miller and Tiu use an approach to higher-order abstract syntax that is
called \emph{$\lambda$-tree syntax}, which allows one to encode both the static and
dynamic structure of abstractions. Their logic $\mathit{FO}\lambda^{\Delta\nabla}$ uses
the \emph{new quantifier} $\nabla a.\phi$, whose meaning is that atom $a$ is fresh in the
formula $\phi$ that lies within the scope of the quantifier. The logic
$\mathit{FO}\lambda^{\Delta\nabla}$ is equipped with a sequent calculus that deals with
the issues concerning name-binding operations. This sequent calculus uses renamings as a
primitive operation.

In the NTSs of Parrow {\etal}~\cite{PBEGW15}, scope opening is modelled by the property of
alpha-conversion of residuals. We have explored an alternative in which scope opening is
encoded by a \emph{residual abstraction} of sort $[\Chan](\Act\times\Proc)$. Similarly,
Parrow has recently proposed an alternative definition of his nominal transition systems
in which scope opening is represented as an alpha-equivalence condition encoded by
explicit name abstraction~\cite{Par18}. We have developed mutual, one-to-one translations
between the NTSs and the NRTSs with residual abstractions. The generality of our NRTSs
also allows for neat specifications of our versions of the early and the late semantics of
the $\pi$-calculus.

Our current proposal aims at following closely the spirit of the seminal work on nominal
techniques by Gabbay, Pitts and their co-workers, and paves the way for the development of
results on rule formats akin to those presented in the aforementioned references.  Amongst
those, we consider the development of a congruence format for the notion of bisimilarity
presented in~\cite[Def.~2]{PBEGW15} to be of particular interest. The logical
characterisation of bisimilarity given in~\cite{PBEGW15} opens the intriguing possibility
of employing the divide-and-congruence approach from~\cite{FGW06} to obtain an elegant
congruence format and a compositional proof system for the logic.

We also plan to lift the congruence formats guaranteeing various \emph{bounded
  nondeterminism} properties (including determinism) to the setting of NRTSS~\cite{ABIMR12, AFGI17, FV03}. In order to increase the applicability of those results it
would also be useful to extend the results in this paper to a setting with state
predicates. Such predicates are an important component in the theory and application of
NTSs to some advanced calculi that include them, \eg, active substitutions and fusions.

Developing rule formats for SOS is always the result of a trade-off between ease of
application and generality. Our rule format for alpha-conversion of residuals in
Definition~\ref{def:alpha-conv-format} is no exception and might be generalised in various
ways. Together with substitution $\gamma$ in conditions (ii) and (iii), a substitution
$\gamma_i$ could be used in condition (i) for each premiss, in order to discard variables
that are used in the target of the premisses but are dropped in the target of the
rule. Moreover, the restrictions on atom $a$ in conditions (i) and (ii) could be relaxed
by considering a subset of premisses in the conditions.

Finally, we are developing rule formats for properties other than alpha-conversion of
residuals. One such rule format ensures a property for NRTSs to the effect that, in each
transition, the support of a state is a subset of the support of its derivative. Another
such format would ensure the converse property. That is, in each transition, the support
of the derivative is a subset of the support of the state. In~\cite{Par18}, Parrow
considers properties analogous to the previous one in the setting of NTSs.

\subsection*{Acknowledgements}
We are grateful to Joachim Parrow for inspiring email discussions about specifications of
NTSs with residuals of abstraction sort~\cite{Par18}. We also extend our gratitude to
Andrew Pitts for his comments on an early draft of the paper, pointers to the literature
and clarifications on nominal sets. Finally, we thank the anonymous reviewers of the
CONCUR 2017 and LMCS submissions for their truly invaluable comments which led to
substantial improvements to the paper.

%%%%%%%%%%%%%%%%%%%%%%%%%%%%%%%%%%%%%%%%%%%%%%%%%%%%%%%%%%%%%%%%%%%%%%%%%%%
%% References
\bibliographystyle{alpha}% the recommended bibstyle
\bibliography{rfnpc-lmcs}

\newpage
\appendix

\section{Preliminaries}%
\label{ap:preliminaries}

\begin{proof}[Proof of Proposition~\ref{prop:support-rho}]
  First, we show that the set $A=\{a,(\rho\ a)\mid\rho\ a\not=a\}$ supports $\rho$. This
  requires us to show that for all permutation $\pi$ that leaves each element in $A$
  invariant, and every atom $b$, $\ren{b}{\pi^{-1};\rho;\pi} = \ren{b}{\rho}$. We
  distinguish the following cases. If $b\in A$ the result holds since
  $\ren{b}{\pi^{-1};\rho;\pi} = \ren{b}{\rho;\pi} = \ren{(\rho\ b)}{\pi}= \rho\
  c=\ren{b}{\rho}$.
  If $b\notin A$ we have,
  $\ren{b}{\pi^{-1};\rho;\pi} = \ren{(\peract{\pi^{-1}}{b})}{\rho;\pi}$. Now, by
  definition of the set $A$ and the permutation $\pi$, it must be the case that
  $(\peract{\pi^{-1}}{b})\notin A$. Otherwise, we would have that there exists $c\in A$
  such that $\peract{\pi}{c}=b\neq c$, which results in a contradiction. Hence,
  $\ren{(\peract{\pi^{-1}}{b})}{\rho;\pi}=\ren{(\peract{\pi^{-1}}{b})}{\pi}=b =
  \ren{b}{\rho}$ since $b\not\in A$, and we are done.

  Finally, we prove that $A$ is the smallest set supporting $\rho$. Assume towards a
  contradiction that there exists $A'\subset A$ that supports $\rho$. Without loss of
  generality, assume that there exists an atom $a\in A$ which is not in $A'$. Let $\pi$ be
  the permutation that leaves each element in $A'$ invariant and such that
  $\peract{\pi}{a}= b$ and $\peract{\pi}{b}= a$ for some $b\notin A$. We have,
  \begin{displaymath}
    \ren{a}{\pi^{-1};\rho;\pi} = \ren{(\peract{\pi^{-1}}{a})}{\rho;\pi} =
    \ren{b}{\rho;\pi} =
    \ren{(\rho\ b)}{\pi} = \peract{\pi}{b}=a,
  \end{displaymath}
  which results in a contradiction because $\ren{a}{\rho}=\rho\ a\neq a$ since $a\in A$.
  Therefore, the set $A$ is the smallest set that supports $\rho$, which ends the proof.
\end{proof}

\section{Nominal Terms}%
\label{ap:terms}

\begin{proof}[Proof of Lemma~\ref{lem:renaming-equivariant}]
  By induction on the size of $t$. If $t=a$, then
  \begin{displaymath}
    \peract{\pi}{(\ren{a}{\rho})}
    = \ren{a}{\rho;\pi} = \ren{a}{\pi;\pi^{-1};\rho;\pi}
    = \ren{(\pi\ a)}{\pi^{-1};\rho;\pi}
    = \ren{(\pi\ a)}{\peract{\pi}{\rho}}.
  \end{displaymath}

  \noindent
  If $t=\susp{t'}{\rho'}$, then
  \begin{displaymath}
    \begin{array}{l}
      \peract{\pi}{(\ren{(\susp{t'}{\rho'})}{\rho})}
      = \peract{\pi}{(\susp{t'}{\rho';\rho})}
      = \susp{(\peract{\pi}{t'})}{\peract{\pi}{(\rho';\rho)}}\\
      = \susp{(\peract{\pi}{t'})}{\pi^{-1};\rho';\rho;\pi}
      = \susp{(\peract{\pi}{t'})}{\pi^{-1};\rho';\pi;\pi^{-1};\rho;\pi}\\
      = \susp{(\peract{\pi}{t'})}{\peract{\pi}{\rho'};\peract{\pi}{\rho}}
      = \ren{(\susp{(\peract{\pi}{t'})}{\peract{\pi}{\rho'}})}{\peract{\pi}{\rho}}
      = \ren{(\peract{\pi}{(\susp{t'}{\rho'})})}{\peract{\pi}{\rho}}.
    \end{array}
  \end{displaymath}

  \noindent
  If $t=[a]t'$, then
  \begin{displaymath}
    \begin{array}{l}
      \peract{\pi}{(\ren{([a]t')}{\rho})}
      = \peract{\pi}{([\rho\ a](\ren{t'}{\rho}))}
      = [\peract{\pi}{(\rho\ a)}](\peract{\pi}{(\ren{t'}{\rho})})
      = [\peract{\pi}{(\ren{a}{\rho})}](\peract{\pi}{(\ren{t'}{\rho})}).
    \end{array}
  \end{displaymath}
  By the induction hypothesis,
  \begin{displaymath}
    \begin{array}{l}
      [\ren{(\pi\ a)}{\peract{\pi}{\rho}}](\ren{(\peract{\pi}{t'})}{\peract{\pi}{\rho}})
      = [(\peract{\pi}{\rho})(\pi\ a)](\ren{(\peract{\pi}{t'})}{\peract{\pi}{\rho}})\\
      = \ren{([\pi\ a](\peract{\pi}{t'}))}{\peract{\pi}{\rho}}
      = \ren{(\peract{\pi}{([a]t')})}{\peract{\pi}{\rho}}.
    \end{array}
  \end{displaymath}
  The remaining cases are straightforward by the induction hypothesis.
\end{proof}

We lift the action of renaming $\ren{A}{\rho}$ to sets of atoms $A$ in the obvious
way. Let $t$ be a raw term. Lemma~11.1 in~\cite{GH08} states that the support of
$\ren{t}{\rho}$ is a subset of $\ren{(\supp(t))}{\rho}$.

\begin{lem}\label{lem:support-moderated-term}
  Let $A$ be a set of atoms. Then, $\ren{A}{\rho} \subseteq A\cup \supp(\rho)$.
\end{lem}
\begin{proof}
  Consider atom $a\in\ren{A}{\rho}$. If $a\in A$, then the result trivially
  follows. Otherwise, $a\in\ren{A}{\rho}\setminus A$. We claim that $a\in \supp(\rho)$.
  Indeed, $a = \rho\ b$ for some $b\in A$. Since $a$ is not in $A$, it follows that
  $a \neq b$, and therefore $a \in \supp(\rho)$ by Proposition~\ref{prop:support-rho}.
\end{proof}
% \begin{proof}
%   By induction on the size of $A$. If $A$ is empty, then the lemma trivially
%   holds. Otherwise, without loss of generality let $A=\{a\}\cup A'$. We have that
%   $\ren{A}{\rho} = \{\rho\ a\}\cup\ren{A'}{\rho}$. By the induction hypothesis we know
%   that $\ren{A'}{\rho}\subseteq A'\cup\supp(\rho)\subseteq A\cup\supp(\rho)$. If
%   $\rho\ a \in A$, then the results holds trivially. If $\rho\ a\not\in A$,
%   $\rho\ a\in\supp(\rho)$ by Proposition~\ref{prop:support-rho}, and the result holds.
% \end{proof}

\begin{proof}[Proof of Lemma~\ref{lem:fa-subset-supp}]
  By induction on the size of $t$. The only non-trivial case is $t={\susp{t'}{\rho}}$. By
  definition, $\fa(\susp{t'}{\rho}) = \fa(\ren{t'}{\rho})$. By the induction hypothesis,
  $\fa(\ren{t'}{\rho})\subseteq \supp(\ren{t'}{\rho})$. By
  Lemma~\ref{lem:support-moderated-term},
  $\ren{(\supp(t'))}{\rho} \subseteq \supp(t')\cup\supp(\rho)= \supp(\susp{t'}{\rho})$ and
  the claim follows by Lemma~11.1 in~\cite{GH08}.
\end{proof}

\section{Rule Formats for NRTSSs}%
\label{ap:rule_formats}

The proofs of some of the lemmas to come use induction on the size of a freshness
environment. We let the \emph{size of a freshness environment $\nabla$} be the sum of the
sizes of the raw terms in its assertions.

\begin{proof}[Proof of Lemma~\ref{lem:unique-nf}]
  The proof goes along the same lines as the proof of Lemma~11 in~\cite{FG07}. Since the
  simplification rules do not overlap, there are no critical pairs and confluence holds
  trivially. Each simplification rule decreases the size of some assertion in the
  environment, except for the rule
  $\{a \fra \susp{([b]t)}{\rho}\}\cup\nabla \Longrightarrow \{a \fra [\rho\
  b](\susp{t}{\rho})\}\cup\nabla$. However, the environments that pattern-match with that
  rule simplify as follows
  \begin{displaymath}
    \begin{array}{l}
      \{a \fra \susp{([b]t)}{\rho}\}\cup\nabla \Longrightarrow
      \{a \fra [\rho\ b](\susp{t}{\rho})\}\cup\nabla \Longrightarrow
      \left\{
      \begin{array}{ll}
        \{a\fra \susp{t}{\rho}\} \cup \nabla &\quad \textup{if $a \not= \rho\ b$}\\
        \nabla &\quad \textup{otherwise}
      \end{array}
      \right.
    \end{array}
  \end{displaymath}
  and thus the assertion $a \fra \susp{([b]t)}{\rho}$ either decreases its size or
  vanishes after the two consecutive simplification steps above. Since the reduction
  relation is confluent, the environments of the shape above can always be reduced in this
  fashion. Therefore the reduction relation is terminating.
\end{proof}

\begin{proof}[Proof of Lemma~\ref{lem:entails}]
  We first prove that $\varphi(\nabla)$ holds iff $\varphi(\nf{\nabla})$ holds. We proceed
  by induction on the size of $\nabla$. If $\nabla = \nf{\nabla}$ then the result follows
  trivially. Without loss of generality we let $\nabla = \{a \fra t\}\cup \nabla'$ and
  consider the cases where some simplification rule is applicable. If the assertion
  $a \fra t$ vanishes after the simplification step, then $t$ is either an atom $b\#a$, or
  $t=[a]t'$, and in both cases $a\#\NT{\varphi(t)}$ for every ground substitution
  $\varphi$. If the assertion $a \fra t$ simplifies to a set of assertions
  $\{a \fra t_i \mid i\in I\}$, then we show that $a\#\NT{\varphi(t)}$ iff
  $\bigwedge_{i\in I}(a\#\NT{\varphi(t_i)})$, for every ground substitution $\varphi$.

  For illustration, we provide the proof for the cases
  \begin{displaymath}
    \nabla = \{a \fra \susp{(\susp{t'}{\rho_1})}{\rho}\}\cup \nabla'\ \textup{and}\
    \nabla = \{a \fra \susp{([b]t')}{\rho}\}\cup \nabla'.
  \end{displaymath}
  The rest of the cases are straightforward by the induction hypothesis.

  If $\nabla = \{a \fra \susp{(\susp{t'}{\rho_1})}{\rho}\}\cup \nabla'$, then
    \begin{displaymath}
      \begin{array}{l}
        \NT{\varphi(\susp{(\susp{t'}{\rho_1})}{\rho})} =
        \NT{\susp{(\susp{\varphi(t')}{\rho_1})}{\rho}} =
        \NT{\ren{(\susp{\varphi(t')}{\rho_1})}{\rho}} \\
        = \NT{\susp{\varphi(t')}{\rho_1;\rho}}
        = \NT{\varphi(\susp{t'}{\rho_1;\rho})},
      \end{array}
  \end{displaymath}
  and therefore $a\#\NT{\varphi(\susp{(\susp{t'}{\rho_1})}{\rho})}$ iff
  $a\#\NT{\varphi(\susp{t'}{\rho_1;\rho})}$, and the lemma follows by the induction
  hypothesis since the assertion $a \fra\susp{(\susp{t'}{\rho_1})}{\rho}$ simplifies to
  $a \fra \susp{t'}{\rho_1;\rho}$.

  If $\nabla = \{a \fra \susp{([b]t')}{\rho}\}\cup \nabla'$, then we consider the following
  cases. If $a = \rho\ b$ then the environment $\nabla$ simplifies to $\nabla'$ in two
  steps and the lemma follows by the induction hypothesis. If $a\neq \rho\ b$, then
  \begin{displaymath}
    \begin{array}{l}
      \NT{\varphi(\susp{([b]t')}{\rho})} =
      \NT{\susp{([b]\varphi(t'))}{\rho}} =
      \NT{\ren{([b]\varphi(t'))}{\rho}}\\
      =\NT{[\rho\ b](\ren{\varphi(t')}{\rho})}
      = \langle\rho\ b\rangle(\NT{\ren{\varphi(t')}{\rho}})
      = \langle\rho\ b\rangle(\NT{\susp{\varphi(t')}{\rho}})\\
      = \NT{[\rho\ b](\susp{\varphi(t')}{\rho})} =
      \NT{\varphi([\rho\ b](\susp{t'}{\rho}))},
    \end{array}
  \end{displaymath}
  and therefore $a\# \NT{\varphi(\susp{([b]t')}{\rho})}$ iff
  $a\#\NT{\varphi([\rho\ b](\susp{t'}{\rho}))}$, and the lemma follows by the induction
  hypothesis since the assertion $a \fra \susp{([b]t')}{\rho}$ simplifies to
  $a \fra [\rho\ b](\susp{t'}{\rho})$.

  Now we show that $\varphi(\nf{\nabla})$ holds iff $\varphi(\nf{\widetilde{\nabla}})$
  holds. If $\nf{\widetilde{\nabla}}$ contains assertion $a\fra \susp{x}{\iota}$, then
  $\nf{\nabla}$ contains either $a\fra x$ or $a\fra \susp{x}{\iota}$, or both. The result
  follows trivially since
  $\NT{\varphi(\susp{x}{\iota})} = \NT{\susp{(\varphi(x))}{\iota}} =
  \ren{\NT{\varphi(x)}}{\iota} = \NT{\varphi(x)}$.
\end{proof}

\begin{rem}\label{rem:inconsistent}
  Notice that if $\nabla$ in the lemma above is inconsistent, then the lemma follows
  trivially since no substitution $\varphi$ exists such that $\varphi(\nabla)$ holds. This
  is so because $\nf{\nabla}$ contains some freshness assertion of the form $a\fra a$, and
  neither the conjunction of the freshness relations denoted by $\varphi(\nabla)$ holds,
  nor the conjunction of the ones denoted by $\varphi(\nf{\nabla})$
  does.\blockqed
\end{rem}

\begin{proof}[Proof of Lemma~\ref{lem:entails2}]
  Assume $\varphi(\nabla)$ holds. By Lemma~\ref{lem:entails}, $\nf{\widetilde{\nabla}}$
  holds. Without loss of generality we assume
  $\nf{\widetilde{\nabla'}} = \{a_i\fra b_i\mid i\in I\}\cup\{a_j\fra
  \susp{x_j}{\rho_j}\mid j\in J\}$. We prove that the conjunction
  \begin{displaymath}
    \bigwedge_{i\in I}(a_i\# b_i) \land
    \bigwedge_{j\in J}(a_j\# \NT{\susp{(\varphi(x_j))}{\rho_j}})
  \end{displaymath}
  holds. Since $\nabla\vdash \nabla'$, each of the assertions $a_i\fra b_i$ is contained
  in $\nf{\widetilde{\nabla}}$ and $\bigwedge_{i\in I}(a_i\# b_i)$ holds by
  assumptions. For each assertion $a_j\fra \susp{x_j}{\rho_j}$, we know that there exist a
  permutation $\pi$ and an assertion $b_j\fra \susp{x_j}{\rho'_j}$ in
  $\nf{\widetilde{\nabla}}$ such that $\pi\ a_j=b_j$ and $\rho_j;\pi=\rho'_j$. We have
  \begin{displaymath}
    \begin{array}{l}
      \peract{\pi^{-1}}{\NT{\susp{(\varphi(x_j))}{\rho'_j}}} =
      \peract{\pi^{-1}}{(\ren{\NT{\varphi(x_j)}}{\rho'_j})} =
      \ren{\NT{\varphi(x_j)}}{\rho'_j;\pi^{-1}}\\
      = \ren{\NT{\varphi(x_j)}}{\rho_j;\pi;\pi^{-1}}
      = \ren{\NT{\varphi(x_j)}}{\rho_j}
      = \NT{\susp{\varphi(x_j)}{\rho_j}},
    \end{array}
  \end{displaymath}
  and thus by equivariance of the freshness relation
  \[
    b_j\#\NT{\susp{\varphi(x_j)}{\rho'_j}} \quad \text{iff} \quad a_j\# \NT{\susp{(\varphi(x_j))}{\rho_j}}.
  \]
  Therefore, the conjunction
  \begin{displaymath}
    \bigwedge_{j\in J}(a_j\# \NT{\susp{(\varphi(x_j))}{\rho_j}})
  \end{displaymath}
  holds and we are done.
\end{proof}

\begin{rem}\label{rem:inconsistent2}
  Notice that if $\nabla$ in the lemma above is inconsistent, then the lemma follows
  trivially since no substitution $\varphi$ exists such that the antecedent
  $\varphi(\nabla)$ of the implication holds.\blockqed
\end{rem}

\begin{lem}\label{lem:gamma}
  Let $\mathcal{R}$ be an NRTSS and \textsc{Ru} be a rule
  \begin{mathpar}
    \inferrule*[right=Ru]
    {\{u_i\rel{}(\ell_i,u'_i)\mid i\in I\}
      \qquad\{a_j\fra v_j\mid j\in J\}}
    {t\rel{}(\ell,t')}
  \end{mathpar}
  in $\mathcal{R}$. Let $D$ be the set of variables that occur in the source $t$ of
  \textsc{Ru} but do not occur in the premisses $u_i\rel{}(\ell_i,u'_i)$ with $i\in I$,
  the environment $\nabla$ or the target $t'$ of the rule. For every
  $\gamma: D\to\gTerm{\SigmaNTS}$ and every substitution $\varphi$, a proof tree for
  transition $\NT{\varphi(\gamma(t))} \rel{} \NT{\varphi(\ell,t')}$ that uses \Ru\ as last
  rule exists iff a proof tree for transition
  $\NT{\varphi(t)} \rel{} \NT{\varphi(\ell,t')}$ that uses \Ru\ as last rule exists.
\end{lem}
\begin{proof}
  Since the domain $D$ of $\gamma$ contains variables neither in the premisses nor in the
  target of rule \textsc{Ru}, $\NT{\varphi(\gamma(t))} \rel{} \NT{\varphi(\ell,t')}$ is
  equal to $\NT{\varphi(\gamma(t))} \rel{} \NT{\varphi(\gamma(\ell,t'))}$. Consider a
  proof tree of transition $\NT{\varphi(\gamma(t))} \rel{} \NT{\varphi(\gamma(\ell,t'))}$
  that uses \Ru\ as last rule, if it exists.
  Since none of the variables occurring in the premisses and in the environment are in the
  domain of $\gamma$, $\varphi(w)=\varphi(\gamma(w))$ for each $w$ in
  $\{u_i,u'_i \mid i\in I\} \cup \{v_j \mid j \in J\}$. Hence, the sub-trees that prove
  the premisses $\NT{\varphi(\gamma(u_i))} \rel{} \NT{\varphi(\gamma(\ell_i,u'_i))}$, with
  $i\in I$, also prove transitions $\NT{\varphi(u_i)} \rel{} \NT{\varphi(\ell_i,u'_i)}$;
  and also all $a_j\# \varphi(\gamma(v_j))$ and $a_j\# \varphi(v_j)$ with $j\in J$ hold.
  Therefore, in the case they exist, the proof trees for transitions
  $\NT{\varphi(\gamma(t))} \rel{} \NT{\varphi(\ell,t')}$ and
  $\NT{\varphi(t)} \rel{} \NT{\varphi(\ell,t')}$ that use \Ru\ as last rule share the same
  sub-trees for the premisses, the freshness assertions hold, and the only difference
  between the proof trees is the root node.
\end{proof}

\begin{proof}[Proof of Theorem~\ref{the:alpha-conversion}]
  Since $\mathcal{R}$ is in equivariant format, $\mathcal{R}$ induces an NRTS with an
  equivariant transition relation (Theorem~\ref{thm:rule-format-equivariance}). We prove
  that this transition relation also enjoys alpha-conversion of residuals.  That is, if
  $\NT{p}\rel{}\NT{(\ell,p')}$ is provable in $\mathcal{R}$ and $b\in\bn(\ell)$, then
  $\NT{p}\rel{}\peract{\tr{a}{b}}{\NT{(\ell,p')}}$ for every atom $a$ that is fresh in
  $\NT{(\ell,p')}$.

  Assume that the last rule used in the proof of $\NT{p}\rel{}\NT{(\ell,p')}$ is
  \begin{mathpar}
    \inferrule*[right=Ru]
    {\{u_i\rel{}(\ell_i,u'_i)\mid i\in I\}
      \qquad\{a_j\fra v_j\mid j\in J\}}
    {t\rel{}(\ell,t')}
  \end{mathpar}
  and therefore that for some ground substitution $\varphi$
  \begin{itemize}
  \item $\NT{p} = \NT{\varphi(t)}$ and $\NT{(\ell,p')} = \NT{\varphi(\ell,t')}$,
  \item the premisses $\NT{\varphi(u_i)} \rel{} \NT{\varphi(\ell_i,u'_i)}$ with $i\in I$
    are provable in $\R$, and
  \item the freshness relations $a_j\#\NT{\varphi(v_j)}$ with $j\in J$ hold.
  \end{itemize}

  \noindent
  Observe that $S(\varphi(t),\ell)$ is defined because $\bn(\ell)$ is non-empty. Thus, as
  rule \Ru\ is in ACR format, there is a ground substitution $\gamma$ whose domain is
  contained in the set of variables $D$ occurring in $t$ but nowhere else in the rule,
  meeting conditions~\ref{it:one}-\ref{it:three} in Definition~\ref{def:alpha-conv-format}
  for each atom $a$ in the set
  $\Atom\setminus \{c\in\supp(t)\mid \nf{\{c\fra t\}}=\emptyset\}$ and each atom $b$ in
  $\bn(\ell)$.

  Let us fix any $b\in\bn(\ell)$ and any atom $a$ that is fresh in $\NT{(\ell,p')}$, we
  first show that transition
  $\NT{\varphi(t)}\rel{}\peract{\tr{a}{b}}{\NT{\varphi(\ell,t')}}$ is provable under the
  assumption that $a\#\NT{\varphi(\gamma(t))}$ and $b\#\NT{\varphi(\gamma(t))}$.  We will
  then show that those assumptions hold.

  By Lemma~\ref{lem:gamma} we know that a proof tree of
  $\NT{\varphi(\gamma(t))}\rel{}\NT{\varphi(\ell, t')}$ that uses \Ru\ as last rule
  exists, and since $\mathcal{R}$ is in equivariant format, a proof tree of
  $\peract{\tr{a}{b}}{\NT{\varphi(\gamma(t))}}
  \rel{}\peract{\tr{a}{b}}{\NT{\varphi(\ell,t')}}$
  that uses $\peract{\tr{a}{b}}{\Ru}$ as last rule exists.  By our assumptions, we have
  that $\peract{\tr{a}{b}}{\NT{\varphi(\gamma(t))}}=\NT{\varphi(\gamma(t))}$ and therefore
  $\NT{\varphi(\gamma(t))}\rel{}\peract{\tr{a}{b}}{\NT{\varphi(\ell,t')}}$.  Again by
  Lemma~\ref{lem:gamma}, a proof tree of
  $\NT{\varphi(t)}\rel{}\peract{\tr{a}{b}}{\NT{\varphi(\ell,t')}}$ that uses
  $\peract{\tr{a}{b}}{\Ru}$ as last rule exists, and the theorem holds.

  In the remainder we prove the assumptions $a\#\NT{\varphi(\gamma(t))}$ and
  $b\#\NT{\varphi(\gamma(t))}$.

  We prove first $a\# \NT{\varphi(\gamma(t))}$. We distinguish two cases:
  \begin{itemize}
  \item If $a\in\{c\in\supp(t)\mid \nf{\{c\fra t\}}=\emptyset\}$ then $\vdash \{a\fra t\}$
    and by Lemmas~\ref{lem:entails} and~\ref{lem:entails2},
    $a\#\NT{(\overline{\varphi\circ\gamma})(t)}=\NT{\varphi(\gamma(t))}$ holds.
  \item Otherwise, since $\mathcal{R}$ is in ACR format with respect to $S$,
    \begin{itemize}
    \item $\{a\fra t'\} \cup \nabla \vdash \{a\fra u'_i \mid i\in I\}$ and
    \item
      $\{a\fra t'\} \cup \nabla \cup \{a\fra u_i \mid i\in I\} \vdash \{a\fra
      \gamma(t)\}$.
      \end{itemize}
      We use Lemmas~\ref{lem:entails} and~\ref{lem:entails2} to obtain the implications
      \begin{enumerate}[label={(\arabic*)},ref={(\arabic*)}]
      \item\label{eq:first}
        $(\; a\#\NT{\varphi(t')}\ \land \bigwedge_{j\in J}(a_j\# \NT{\varphi(v_j)})\;)
        \Longrightarrow \bigwedge_{i\in I}(a\#\NT{\varphi(u'_i)})$ and
      \item\label{eq:second}
        $(\; a\#\NT{\varphi(t')}\ \land \bigwedge_{j\in J}(a_j\# \NT{\varphi(v_j)})\ \land
        \bigwedge_{i\in I}(a\#\NT{\varphi(u_i)}) \;) \Longrightarrow
        a\#\NT{\varphi(\gamma(t))}$.
     \end{enumerate}
     Since the set $D$ does not contain any variable occurring in $t'$ it follows that
     $\varphi(\gamma(t'))=\varphi(t')$.
     Now we prove the statement
     $a \# \NT{\varphi(t')} \Longrightarrow a \# \NT{\varphi(\gamma(t))}$ by induction on
     $S(\varphi(\gamma(t)),\ell)$ (this suffices to show the claim since
     $a \# \NT{\varphi(t')}$ holds by assumption).
  The base case is when $S(\varphi(\gamma(t)),\ell)$ is minimal.  By
  Definition~\ref{def:partial-strict stratification} the rule \Ru\ has no premisses and
  the set $I$ is empty, which makes $\bigwedge_{i\in I}(a\#\NT{\varphi(u_i)})$ trivially
  true and what we were proving holds by~\ref{eq:second}.
  Now assume that $S(\varphi(\gamma(t)),\ell)$ is not minimal.  Since all
  $a_j\# \NT{\varphi(v_j)}$ with $j\in J$ hold, all $a\#\NT{\varphi(u'_i)}$ with $i\in I$
  hold by~\ref{eq:first}.  Condition (ii) in Definition~\ref{def:partial-strict
    stratification} ensures that $S(\varphi(u_i),\ell_i)\not=\bot$ and
  $S(\varphi(\gamma(u_i)),\ell_i) < S(\varphi(\gamma(t)),\ell)$.  Thus, we can apply the
  induction hypothesis to obtain the implications
  $a\#\NT{\varphi(u'_i)}\implies a\#\NT{\varphi(\gamma(u_i))}$, with $i\in I$.  For each
  $i\in I$, since the variables that occur in $u_i$ are not in $\dom(\gamma)$, we have
  that $a\#\NT{\varphi(u'_i)}\implies a\#\NT{\varphi(u_i)}$.  And now by~\ref{eq:second}
  we know that $a\#\NT{\varphi(\gamma(t))}$ which is what was to be shown.
  \end{itemize}

  \noindent
  To finish the proof we prove the statement $b \# \NT{\varphi(\gamma(t))}$ by induction
  on $S(\varphi(\gamma(t)),\ell)$.
  Since $\mathcal{R}$ is in ACR format with respect to $S$ we have that
  $\nabla\cup\{b\fra u_i\mid i\in I \land b\in \bn(\ell_i)\} \vdash \{b\fra \gamma(t)\}$.
  We use Lemmas~\ref{lem:entails} and~\ref{lem:entails2} to obtain the implication
  \begin{enumerate}[label={(\arabic*)},ref={(\arabic*)}]
    \setcounter{enumi}{2}
  \item\label{eq:third}
    $(\;\bigwedge_{j\in J}(a_j\#\NT{\varphi(v_j)})\ \land \bigwedge_{i\in I \land
      b\in\bn(\ell_i)} (b\#\NT{\varphi(u_i)})\;) \Longrightarrow
    b\#\NT{\varphi(\gamma(t))}$.
  \end{enumerate}
  The base case for the induction is when $S(\varphi(\gamma(t)),\ell)$ is minimal.  By
  Definition~\ref{def:partial-strict stratification} the rule \Ru\ has no premisses and
  the set $I$ is empty, so that $\{i\mid i\in I\land b\in\bn(\ell_i)\}$ is empty as well,
  in which case $\bigwedge_{i\in I \land b\in\bn(\ell_i)} (b\#\NT{\varphi(u_i)})$ is
  trivially true and $b \# \NT{\varphi(\gamma(t))}$ holds.
  Now assume that $S(\varphi(\gamma(t)),\ell)$ is not minimal.  Condition (ii) in
  Definition~\ref{def:partial-strict stratification} ensures that
  $S(\varphi(u_i),\ell_i)\not=\bot$ and
  $S(\varphi(\gamma(u_i)),\ell_i) < S(\varphi(\gamma(t)),\ell)$ for every $i\in I$. Thus,
  we can apply the induction hypothesis to obtain $b\#\NT{\varphi(\gamma(u_i))}$ for each
  $i\in I$ such that $b\in\bn(\ell_i)$.  For each $i\in I$, since the set $D$ does not
  contain any variable occurring in $u_i$ we know that
  $\varphi(\gamma(u_i))=\varphi(u_i)$.  In particular this holds for $i\in I$ such that
  $b\in\bn(\ell_i)$.  By~\ref{eq:third} we know that $b\# \NT{\varphi(\gamma(t))}$ and we
  are done.
\end{proof}

\section{Example of Application of the BA-Format to the \texorpdfstring{$\pi$}{pi}-Calculus}%
\label{ap:example-application-BA-format}

\begin{proof}[Proof of Lemma~\ref{lem:tre-translates-treabs}]
  For every transition $p\trel (\ell, p')$, we have to show that $p\trel (\ell, p')$ has a
  proof tree in $\RE$ iff $p\trelAbs [a](\ell, p')$ has a proof tree in $\REAbs$, where
  either $\bnE(\ell)=\{a\}$, or $\bnE(\ell)=\emptyset$ and $a\#(\ell,p')$. We proceed by
  induction on the height of the proof tree of $p\trel (\ell, p')$. We prove the ``if''
  direction first.

  The base case is when $p\trel (\ell, p')$ is provable by any of the axioms \textsc{EIn},
  \textsc{Out} or \textsc{Tau}. In all of these cases, $\bnE(\ell)=\emptyset$. The
  transition $p\trelAbs [a](\ell, p')$ is provable by axioms \textsc{AEIn}, \textsc{AOut}
  or \textsc{ATau} respectively, where we let $a\#(\ell,p')$.

  For the inductive step, we distinguish the following sub-cases depending on the last rule
  used in the proof of $p\trel (\ell, p')$:
  \begin{itemize}
  \item The last rule used is \textsc{EParL} or \textsc{EParR}. Without loss of
    generality, we assume that the last rule used is \textsc{EParL} and thus
    $p=\parPA[p_1]{p_2}$ and $p'=\parPA[p_1']{p_2}$ where $p_1\trel(\ell,p_1')$ is
    provable in $\RE$ where $\ell\not\in\{\boutAA[a]{b}\mid a,b\in\Atom_{\Chan}\}$ and
    therefore $\bnE(\ell)=\emptyset$. By the induction hypothesis,
    $p_1\trelAbs[a](\ell,p_1')$ is provable in $\REAbs$, where $a\#(\ell,p_1')$. Without
    loss of generality, we let $a\#p_2$. The transition
    $p\trelAbs[a](\ell,\parPA[p_1']{p_2})$ is provable by rule \textsc{AParL}.
  \item The last rule used is \textsc{EParResL} or \textsc{EParResR}. Without loss of
    generality, we assume that the last rule used is \textsc{EParResL} and thus
    $p=\parPA[p_1]{p_2}$, where $\parPA[p_1]{p_2}\trel(\boutAA[a]{b},\parPA[p_1']{p_2})$
    is provable in $\RE$ and $b\#p_2$. By the induction hypothesis,
    $p_1\trelAbs[b](\boutAA[a]{b},p_1')$ is provable in $\REAbs$, and
    \begin{displaymath}
      \parPA[p_1]{p_2}\trelAbs[b](\boutAA[a]{b},\parPA[p_1']{p_2})
    \end{displaymath}
    is provable by rule \textsc{AParL} since $b\#p_2$.
  \item The last rule used is \textsc{ECommL} or \textsc{ECommR}, and thus
    $\ell=\tauAA$. Without loss of generality, we assume that the last rule used is
    \textsc{ECommL} and thus $p=\parPA[p_1]{p_2}$ and $p'=\parPA[p_1']{p_2'}$, where
    $p_1\trel (\outAA[a]{b},p_1')$ and $p_2\trel (\inAA[a]{b},p_2')$ are provable in
    $\RE$. By the induction hypothesis, $p_1\trelAbs [c](\outAA[a]{b},p_1')$ and
    $p_2\trelAbs [d](\inAA[a]{b},p_2')$ are provable in $\REAbs$, and since
    $\bnE(\outAA[a]{b})=\bnE(\inAA[a]{b})=\emptyset$, therefore $c\#p_1'$ and
    $d\#p_2'$. Without loss of generality we assume that $c=d$. Therefore,
    $p\trelAbs [c](\ell, p')$ is provable in $\REAbs$ by rule \textsc{AECommL} where
    $\bnE(\tauAA)=\emptyset$ and $c\#(\tauAA,\parPA[p_1']{p_2'})$.
  \item The last rule used is any of \textsc{ECloseL}, \textsc{ECloseR}, \textsc{ERepComm}
    or \textsc{ERepClose}. Consider that the last rule used is \textsc{ECloseL}, and thus
    $\ell=\tauAA$, $p=\parPA[p_1]{p_2}$ and $p'=\newPA[{[b](\parPA[p_1']{p_2'})}]$, where
    $p_1\trel (\boutAA[a]{b},p_1')$ and $p_2\trel (\inAA[a]{b},p_2')$ are provable in
    $\RE$ and $b\#p_2$. By the induction hypothesis, $p_1\trelAbs [b](\boutAA[a]{b},p_1')$
    and $p_2\trelAbs [c](\inAA[a]{b},p_2')$ are provable in $\REAbs$ where
    $c\#p_2'$. Without loss of generality, we let $c\#p_1'$. Therefore
    $p\trelAbs [c](\tauAA, p')$ is provable in $\REAbs$ by rule \textsc{AECloseL} where
    $\bnE(\tauAA) = \emptyset$ and $c\#\newPA[{[b](\parPA[p_1']{p_2'})}]$. The cases where
    the last rule used is any of \textsc{ECloseR}, \textsc{ERepComm} or \textsc{ERepClose}
    are analogous.
  \item The last rule used is \textsc{Rep}, thus
    $p = \repPA[p_1] \trel (\ell,\parPA[p_1']{\repPA[p_1]})$ where $p_1\trel (\ell,p_1')$
    is provable in $\RE$. By the induction hypothesis, $p_1\trel [a](\ell,p_1')$ is
    provable in $\REAbs$. If $\bnE(\ell)=\emptyset$, then $a\#(\ell,p_1')$. Without loss
    of generality, we let $a\#p_1$, and therefore
    $\repPA[p_1] \trel [a](\ell,\parPA[p_1']{\repPA[p_1]})$ is provable in $\REAbs$ by
    rule \textsc{ARep} and $\bnE(\ell)=\emptyset$ and
    $a\#(\ell,\parPA[p_1']{\repPA[p_1]})$.

    If $\bnE(\ell)=\{a\}$, then $a\#p_1$ since $\RE$ is in ACR-format (see
    page~\pageref{pag:fresh-in-source}), which guarantees that the binding name $a$ in
    transition $p_1\trel (\ell,p_1')$ is fresh in its source $p_1$. Therefore, transition
    $p = \repPA[p_1] \trelAbs [a](\ell,\parPA[p_1']{\repPA[p_1]})$ is provable in $\REAbs$
    by rule \textsc{ARep}.
  \item The last rule used is \textsc{Open}. We have that $\ell=\boutAA[a]{b}$,
    $\bnE(\boutAA[a]{b})=\{b\}$ and $b\#a$. Transition $p\trelAbs[b](\boutAA[a]{b},p')$ is
    provable by applying the induction hypothesis and by rule \textsc{AOpen}.
  \item The last rule used is any of \textsc{SumL}, \textsc{SumR} or \textsc{Res}. These
    cases are analogous to the case \textsc{EParL}.
  \end{itemize}

  \noindent
  The ``only if'' direction can be checked similarly, except for the observation that for
  a transition $p\trelAbs [a](\ell,p')$ provable in $\REAbs$ and where the last rule used
  is \textsc{AParL}, we distinguish the cases where $a\#\ell$ and where $a\in\supp(\ell)$,
  and use the induction hypothesis together with rule \textsc{EParL} or rule
  \textsc{EParResL}, respectively, to prove that $p\trel(\ell,p')$ is provable in $\RE$.
\end{proof}

\begin{proof}[Proof of Lemma~\ref{lem:trl-translates-trlabs}]
  For every transition $p\trel (\ell, p')$, we have to show that $p\trel (\ell, p')$ has a
  proof tree in $\RL$ iff $p\trelAbs [a](\ell, p')$ has a proof tree in $\RLAbs$, where
  either $\bnL(\ell)=\{a\}$, or $\bnL(\ell)=\emptyset$ and $a\#(\ell,p')$. We proceed by
  induction on the height of the proof tree of $p\trel (\ell, p')$. We prove the ``if''
  direction first.

  The base case is when $p\trel (\ell, p')$ is provable by any of the axioms \textsc{LIn},
  \textsc{Out} or \textsc{Tau}. In the case where $p\trel (\ell, p')$ is provable by axiom
  \textsc{LIn}, we know that $\ell=\binAA[a]{b}$, $\bnL(\binAA[a]{b}) = \{b\}$ and
  $b\#a$. Therefore $p\trelAbs [b](\binAA[a]{b}, p')$ is provable by axiom
  \textsc{ALIn}. In the other two cases are already proven in Proof of
  Lemma~\ref{lem:tre-translates-treabs}.

  For the inductive step, we distinguish the following sub-cases depending on the last rule
  used in the proof of $p\trel (\ell, p')$:
  \begin{itemize}
  \item The last rule used is \textsc{LParL} or \textsc{LParR}. Without loss of
    generality, we assume that the last rule used is \textsc{LParL} and thus
    $p=\parPA[p_1]{p_2}$ and $p'=\parPA[p_1']{p_2}$ where $p_1\trel(\ell,p_1')$ is
    provable in $\RL$ where
    $\ell\not\in\{\boutAA[a]{b},\binAA[a]{b}\mid a,b\in\Atom_{\Chan}\}$ and therefore
    $\bnL(\ell)=\emptyset$. By the induction hypothesis, $p_1\trelAbs[a](\ell,p_1')$ is
    provable in $\RLAbs$, where $a\#(\ell,p_1')$. Without loss of generality, we let
    $a\#p_2$. The transition $p\trelAbs[a](\ell,\parPA[p_1']{p_2})$ is provable by rule
    \textsc{AParL}.
  \item The last rule used is \textsc{LParResL} or \textsc{LParResR}. Without loss of
    generality, we assume that the last rule used is \textsc{LParResL} and thus
    $p=\parPA[p_1]{p_2}$, where $\parPA[p_1]{p_2}\trel(\ell,\parPA[p_1']{p_2})$ with
    $\bnL(\ell)=\{b\}$ is provable in $\RL$ and $b\#p_2$. By the induction hypothesis,
    $p_1\trelAbs[b](\ell,p_1')$ is provable in $\RLAbs$, and
    \begin{displaymath}
      \parPA[p_1]{p_2}\trelAbs[b](\ell,\parPA[p_1']{p_2})
    \end{displaymath}
    is provable by rule \textsc{AParL} since $b\#p_2$.
  \item The last rule used is \textsc{LCommL} or \textsc{LCommR}, and thus
    $\ell=\tauAA$. Without loss of generality, we assume that the last rule used is
    \textsc{LCommL} and thus $p=\parPA[p_1]{p_2}$ and $p'=\parPA[p_1']{p_2'}$, where
    $p_1\trel (\outAA[a]{b},p_1')$ and $p_2\trel (\binAA[a]{c},p_2')$ are provable in
    $\RL$. By the induction hypothesis, $p_1\trelAbs [d](\outAA[a]{b},p_1')$ and
    $p_2\trelAbs [c](\binAA[a]{c},p_2')$ are provable in $\RLAbs$, and since
    $\bnL(\outAA[a]{b})=\emptyset$, therefore $d\#p_1'$. Without loss of generality we
    assume that $d\#\ren{p_2'}{\rep{b}{c}}$. Therefore, $p\trelAbs [d](\tauAA, p')$ is
    provable in $\RLAbs$ by rule \textsc{ALCommL} where $\bnL(\tauAA)=\emptyset$ and
    $d\#(\tauAA,\parPA[p_1']{\ren{p_2'}{\rep{b}{c}}})$.
  \item The last rule used is any of \textsc{LCloseL}, \textsc{LCloseR}, \textsc{LRepComm}
    or \textsc{LRepClose}. Consider that the last rule used is \textsc{LCloseL}, and thus
    $\ell=\tauAA$, $p=\parPA[p_1]{p_2}$ and $p'=\newPA[{[b](\parPA[p_1']{p_2'})}]$, where
    $p_1\trel (\boutAA[a]{b},p_1')$ and $p_2\trel (\binAA[a]{b},p_2')$ are provable in
    $\RL$. By the induction hypothesis, $p_1\trelAbs [b](\boutAA[a]{b},p_1')$ and
    $p_2\trelAbs [b](\binAA[a]{b},p_2')$ are provable in $\RLAbs$. Without loss of
    generality, we let $c\#[b](p_1',p_2')$. Therefore $p\trelAbs [c](\tauAA, p')$ is
    provable in $\RLAbs$ by rule \textsc{AECloseL} where $\bnL(\tauAA) = \emptyset$ and
    $c\#\newPA[{[b](\parPA[p_1']{p_2'})}]$. The cases where the last rule used is any of
    \textsc{LCloseR}, \textsc{LRepComm} or \textsc{LRepClose} are analogous.
  \end{itemize}

  \noindent
  The cases where the last rule used is any of \textsc{Rep}, \textsc{Open}, \textsc{SumL},
  \textsc{SumR} and \textsc{Res} are already proven in Proof of
  Lemma~\ref{lem:tre-translates-treabs}.

  The ``only if'' direction can be checked similarly, except for the observation that for
  a transition $p\trelAbs [a](\ell,p')$ provable in $\RLAbs$ and where the last rule used
  is \textsc{AParL}, we distinguish the cases where $a\#\ell$ and where $a\in\supp(\ell)$,
  and use the induction hypothesis together with rule \textsc{LParL} or rule
  \textsc{LParResL}, respectively, to prove that $p\trel(\ell,p')$ is provable in $\RL$.
\end{proof}

\end{document}